\documentclass[sigconf, screen, nonacm, natbib=false]{acmart}

\usepackage{xparse}
\usepackage{enumitem}
\usepackage[ruled,noend]{algorithm2e}
\usepackage{tikz}
\usetikzlibrary{snakes,mindmap}
\usepackage{todonotes}
\usepackage{amsmath}
\usepackage{thmtools}
\usepackage{thm-restate}
\usepackage{mathtools}
\usepackage{biblatex}
\usepackage{oplotsymbl}
\usepackage{subcaption}
\usepackage{refcount}
\usepackage[capitalize]{cleveref}
\usetikzlibrary{arrows,automata,chains,matrix,positioning,scopes,calc,decorations.pathmorphing,shapes.misc,shapes,fit,backgrounds}
\tikzset{mode/.style={font=\scriptsize}}
\tikzset{gadget/.style={->,>=stealth,initial text=,minimum size=7pt,auto,on grid,scale=1,inner sep=1pt,node distance=1cm}}
\tikzset{every state/.style={minimum size=15pt,inner sep=1pt,fill=black!10,draw=black!70,thick}}


\makeatletter
\newcommand{\oset}[3][0ex]{%
  \mathrel{\mathop{#3}\limits^{
    \vbox to#1{\kern-2\ex@
    \hbox{$\scriptstyle#2$}\vss}}}}
\makeatother

\DeclareDocumentCommand{\autstep}{O{}}{\oset{#1}{\rightarrow}}
\DeclareDocumentCommand{\autsteps}{O{}}{\oset{#1}{\rightarrow}}

\DeclareDocumentCommand{\deriv}{O{}}{\,\Rightarrow_{#1}\,}
\DeclareDocumentCommand{\derivs}{O{}}{\,\oset{*}{\Rightarrow}_{#1}\,}

\DeclareDocumentCommand{\zvassnz}{o}{%
	\ensuremath{%
		\IfNoValueTF{#1}{%
			\Z\operatorname{-\mathsf{VASS}}_{\mathsf{nz}}
		}{
			\Z^{\pm}_{#1}\operatorname{-\mathsf{VASS}}_{\mathsf{nz}}
		}
	}
}

\newcommand{\nreach}{\leadsto}
\newcommand{\nreachs}{\oset{*}{\nreach}}

\newcommand{\set}[1]{\{ #1 \}}
\newcommand{\tuple}[1]{\langle #1 \rangle}
\newcommand{\dM}{\mathbb{M}}
\newcommand{\cA}{\mathcal{A}}
\newcommand{\cB}{\mathcal{B}}
\newcommand{\cV}{\mathcal{V}}
\newcommand{\cG}{\mathcal{G}}

\newcommand{\bA}{\mathbf{A}}
\newcommand{\bB}{\mathbf{B}}
\newcommand{\bC}{\mathbf{C}}
\newcommand{\ba}{\mathbf{a}}
\newcommand{\bb}{\mathbf{b}}
\newcommand{\bd}{\mathbf{d}}
\newcommand{\bc}{\mathbf{c}}
\newcommand{\be}{\mathbf{e}}
\newcommand{\bg}{\mathbf{g}}
\newcommand{\bp}{\mathbf{p}}
\newcommand{\bu}{\mathbf{u}}
\newcommand{\bv}{\mathbf{v}}

\newcommand{\bx}{\mathbf{x}}
\newcommand{\by}{\mathbf{y}}

\newcommand{\bzero}{\mathbf{0}}

\newcommand{\B}{\mathbb{B}}
\newcommand{\M}{\mathbb{M}}
\newcommand{\N}{\mathbb{N}}
\newcommand{\Z}{\mathbb{Z}}

\newcommand{\REACH}{\mathsf{REACH}}
\newcommand{\REVREACH}{\mathsf{BIREACH}}
\newcommand{\SUBMEM}{\mathsf{SUBMEM}}

\newcommand{\SCpm}{\mathsf{SC}^\pm}

\newcommand{\ILD}{\mathsf{ILD}}

\newcommand{\compL}{\mathsf{L}}
\newcommand{\compP}{\mathsf{P}}
\newcommand{\EXPSPACE}{\mathsf{EXPSPACE}}
\newcommand{\NL}{\mathsf{NL}}
\newcommand{\EXPTIME}{\mathsf{EXP}}
\newcommand{\NEXPTIME}{\mathsf{NEXP}}
\newcommand{\NP}{\mathsf{NP}}
\newcommand{\PTIME}{\mathsf{P}}

\newcommand{\Cfour}{\mathsf{C4}}
\newcommand{\Pfourlooped}{\mathsf{P4}^\circ}
\newcommand{\Cfourlooped}{\mathsf{C4}^\circ}

\newcommand{\Eff}{\mathsf{Eff}}
\newcommand{\Reach}{\mathsf{Reach}}

\newcommand{\rt}[1]{{#1}^{\tiny{\trianglepa}}}
\newcommand{\nort}[1]{{#1}^{\tiny{\trianglepafillha}}}
\newcommand{\vertup}{\hat{V}}
\newcommand{\vertdown}{\check{V}}
\newcommand{\projup}{\hat{\pi}}
\newcommand{\projdown}{\check{\pi}}
\newcommand{\inv}[1]{{#1}^\dagger}

\newcommand{\xparagraph}[1]{\subsection*{#1}}

\newcommand{\zvass}[2]{
	\begin{tikzpicture}
		\path (0,-#2) -- (0,#2);
		\node[state] (s) {$s$};
		\node[state,above right=0.2cm and 1cm of s] (q1) {};
		\node[state,below right=0.2cm and 1cm of s] (q2) {};
		\node[state,right=2cm of s] (t) {$t$};
		\path[->] 
		(s) edge[very thick] node[above] {$5$} (q1)
		(s) edge[very thick] node[below] {$2$} (q2)
		(q1) edge node[left] {$3$} (q2)
		(q1) edge node[above] {$2$} (t)
		(q2) edge[very thick] node[below] {$1$} (t)
		;
	\end{tikzpicture}
}

\newcommand{\pnpZero}[1]{
\begin{tikzpicture}[every circle/.style={}, scale=#1]
\fill (0,0) circle (2pt) node (a) {}    (1,0) circle (2pt) node (b)  {}   (2,0) circle (2pt) node (c) {};
\draw (a.center) -- (b.center) -- (c.center);
\end{tikzpicture}
}

\newcommand{\threezcounters}[1]{
\begin{tikzpicture}[every circle/.style={}, scale=#1]
\fill (0:0) circle (2pt) node (aa) {}    +(60:1) circle (2pt) node (ab) {}      +(0:1) circle (2pt) node (ac) {} ;
\draw (aa.center) -- (ab.center) -- (ac.center) -- (aa.center);
\draw (aa.center) ++(-150:3pt) circle (3pt);
\draw (ab.center) ++(90:3pt) circle (3pt);
\draw (ac.center) ++(-30:3pt) circle (3pt);
\end{tikzpicture}
}

\newcommand{\twononezcounter}[1]{
\begin{tikzpicture}[every circle/.style={}, scale=#1]
\fill (0:0) circle (2pt) node (aa) {}    +(60:1) circle (2pt) node (ab) {}      +(0:1) circle (2pt) node (ac) {} ;
\draw (ab.center) ++ (90:3pt) circle (3pt);
\draw (aa.center) -- (ab.center) -- (ac.center) -- (aa.center);
\end{tikzpicture}
}

\newcommand{\pnpOne}[1]{
\begin{tikzpicture}[every circle/.style={}, scale=#1]
\fill (0,0) circle (2pt) node (a) {}    (1,0) circle (2pt) node (b)  {}   (2,0) circle (2pt) node (c) {};
\draw (a.center) -- (b.center) -- (c.center);
\draw (a.center) ++(90:3pt) circle (3pt);
\end{tikzpicture}
}

\newcommand{\pnpTwo}[1]{
\begin{tikzpicture}[every circle/.style={}, scale=#1]
\fill (0,0) circle (2pt) node (a) {}    (1,0) circle (2pt) node (b)  {}   (2,0) circle (2pt) node (c) {};
\draw (a.center) -- (b.center) -- (c.center);
\draw (a.center) ++(90:3pt) circle (3pt);
\draw (c.center) ++(90:3pt) circle (3pt);
\end{tikzpicture}
}

\newcommand{\unloopedcycle}[1]{
\begin{tikzpicture}[every circle/.style={}, scale=#1]
\fill (0,0) circle (2pt) node (a) {}    (1,0) circle (2pt) node (b)  {}   (0,-1) circle (2pt) node (c) {}   (1,-1) circle (2pt) node (d) {};
\draw (a.center) -- (b.center) -- (d.center) -- (c.center) -- (a.center);
\end{tikzpicture}
}

\newcommand{\twopushdowns}[1]{
\begin{tikzpicture}[every circle/.style={}, scale=#1]
\fill (0,0) circle (2pt) node (a) {}    (1,0) circle (2pt) node (b)  {}   (0,-1) circle (2pt) node (c) {}   (1,-1) circle (2pt) node (d) {};
\draw (a.center) -- (b.center) -- (c.center) -- (d.center) -- (a.center);
\end{tikzpicture}
}

\newcommand{\loopedcycle}[1]{
\begin{tikzpicture}[every circle/.style={}, scale=#1]
\fill (0,0) circle (2pt) node (a) {}    (1,0) circle (2pt) node (b)  {}   (0,-1) circle (2pt) node (c) {}   (1,-1) circle (2pt) node (d) {};
\draw (a.center) -- (b.center) -- (d.center) -- (c.center) -- (a.center);
\draw (a.center) ++(135:3pt) circle (3pt);
\draw (b.center) ++(45:3pt) circle (3pt);
\draw (d.center) ++(-45:3pt) circle (3pt);
\draw (c.center) ++(225:3pt) circle (3pt);
\end{tikzpicture}
}

\newcommand{\loopedpathFour}[2]{
\begin{tikzpicture}[every circle/.style={}, scale=#1]
\path (0,-#2) -- (0,#2);
\fill (0,0) circle (2pt) node (a) {}    (1,0) circle (2pt) node (b)  {}   (2,0) circle (2pt) node (c) {}   (3,0) circle (2pt) node (d) {};
\draw (a.center) -- (b.center) -- (c.center) -- (d.center);
\draw (a.center) ++(90:3pt) circle (3pt);
\draw (d.center) ++(90:3pt) circle (3pt);
\draw (b.center) ++(90:3pt) circle (3pt);
\draw (c.center) ++(90:3pt) circle (3pt);
\end{tikzpicture}
}

\newcommand{\treeLevels}[1]{
\begin{tikzpicture}[scale=#1]

\tikzstyle{level}=[dotted, black!60!green]
\draw (-2.5,0.3) edge[level] (4.5,0.3);
\draw (-2.5,-0.5) edge[level] (4.5,-0.5);
\draw (-2.5,-1.3) edge[level] (4.5,-1.3);
\draw (-2.5,-2.1) edge[level] (4.5,-2.1);
\draw (-2.5,-2.9) edge[level] (4.5,-2.9);
\draw (-2.5,-3.7) edge[level] (4.5,-3.7);

\node [black!60!green] at (4.8,0.3) {\scriptsize level 5};
\node [black!60!green] at (4.8,-0.5) {\scriptsize level 4};
\node [black!60!green] at (4.8,-1.3) {\scriptsize level 3};
\node [black!60!green] at (4.8,-2.1) {\scriptsize level 2};
\node [black!60!green] at (4.8,-2.9) {\scriptsize level 1};
\node [black!60!green] at (4.8,-3.7) {\scriptsize level 0};

\node[fill=white, inner sep = 0pt] at (1.3,0.3) {$\trianglepa$};
\node[fill=white, inner sep = 0pt] at (1.3,-0.5) {$\trianglepafillha$};
\node[fill=white, inner sep = 0pt] at (3,-1.3) {$\trianglepa$};
\node[fill=white, inner sep = 0pt] at (3,-2.1) {$\trianglepafillha$};
\node[fill=white, inner sep = 0pt] at (-0.7,-2.9) {$\trianglepa$};
\node[fill=white, inner sep = 0pt] at (-0.7,-3.7) {$\trianglepafillha$};
\node[fill=white, inner sep = 0pt] at (1.9,-2.9) {$\trianglepa$};
\node[fill=white, inner sep = 0pt] at (1.9,-3.7) {$\trianglepafillha$};
\node[fill=white, inner sep = 0pt] at (4,-2.9) {$\trianglepa$};
\node[fill=white, inner sep = 0pt] at (4,-3.7) {$\trianglepafillha$};

\tikzstyle{bag}=[fill = blue!15, circle, draw = blue!15]
\tikzstyle{con}=[circle connection bar switch color=from (blue!15) to (blue!15)]

\node[bag, minimum size = 3.4em] (a) at (0.3,-0.15) {};
\node[bag, minimum size = 3.6em] (b) at (2.1,-1.7) {};
\node[bag, minimum size = 3.4em] (c) at (-1.5,-3.3) {};
\node[bag, minimum size = 3.4em] (d) at (1.1,-3.3) {};
\node[bag, minimum size = 3.4em] (e) at (3.1,-3.3) {};

\path (a) to[con] (b);
\path (a) to[con] (c);
\path (b) to[con] (d);
\path (b) to[con] (e);

\tikzset{every node/.style={fill = black, circle, inner sep = 2pt}}
\node (a1) at (0,0) {};
\node (a2) at (0.6,0) {};
\node (a3) at (0.3,-0.5) {};
\node (b1) at (-1.5,-3.0) {};
\node (b2) at (-1.8,-3.5) {};
\node (b3) at (-1.2,-3.5) {};
\node (c1) at (2.35,-1.45) {};
\node (c2) at (1.85,-1.45) {};
\node (c3) at (2.35,-1.95) {};
\node (c4) at (1.85,-1.95) {};
\node (d1) at (1.1,-2.9) {};
\node (d2) at (1.5,-3.17) {};
\node (d3) at (0.7,-3.17) {};
\node (d4) at (1.35,-3.65) {};
\node (d5) at (0.85,-3.65) {};
\node (e1) at (3.1,-3.0) {};
\node (e2) at (3.4,-3.5) {};
\node (e3) at (2.8,-3.5) {};

\draw (a1.center) ++(135:3pt) circle (3pt);
\draw (a2.center) ++(45:3pt) circle (3pt);
\draw (a3.center) ++(-90:3pt) circle (3pt);
\draw (a1) -- (a2);
\draw (a1) -- (a3);
\draw (a2) -- (a3);

\draw (b1.center) ++(90:3pt) circle (3pt);
\draw (b1) -- (b2);
\draw (b1) -- (b3);
\draw (b2) -- (b3);

\draw (c1.center) ++(45:3pt) circle (3pt);
\draw (c2.center) ++(135:3pt) circle (3pt);
\draw (c3.center) ++(-45:3pt) circle (3pt);
\draw (c4.center) ++(-135:3pt) circle (3pt);
\draw (c1) -- (c2);
\draw (c1) -- (c3);
\draw (c1) -- (c4);
\draw (c2) -- (c3);
\draw (c2) -- (c4);
\draw (c3) -- (c4);

\draw (d1) -- (d2);
\draw (d1) -- (d3);
\draw (d1) -- (d4);
\draw (d1) -- (d5);
\draw (d2) -- (d3);
\draw (d2) -- (d4);
\draw (d2) -- (d5);
\draw (d3) -- (d4);
\draw (d3) -- (d5);
\draw (d4) -- (d5);

\draw (e1.center) ++(90:3pt) circle (3pt);
\draw (e2.center) ++(-45:3pt) circle (3pt);
\draw (e3.center) ++(225:3pt) circle (3pt);
\draw (e1) -- (e2);
\draw (e1) -- (e3);
\draw (e2) -- (e3);

\end{tikzpicture}
}

\newcommand{\pushdownzvass}[1]{
	\newcommand{\distance}{0.7}
	\begin{tikzpicture}[scale=#1]
		
		\tikzset{every node/.style={fill = black, circle, inner sep = 1.5pt}}
		\foreach \x/\angle in {1/90, 2/162, 3/234, 4/306, 5/378}
			\node (d\x) at (\angle:\distance) {};

		\foreach \x in {1,2,5}
			\foreach \y in {3,4}
				\draw (d\x) -- (d\y);

		\foreach \x in {1,2,5}
			\foreach \y in {1,2,5}
				\draw (d\x) -- (d\y);
		\draw (d1.center) ++(90:3pt) circle (3pt);
		\draw (d2.center) ++(180:3pt) circle (3pt);
		\draw (d5.center) ++(0:3pt) circle (3pt);
	\end{tikzpicture}
}

\newcommand{\pzvassTree}[1]{
\begin{tikzpicture}[scale=#1]

\tikzstyle{level}=[dotted, black!60!green]
\draw (0,-1.3) edge[level] (4.5,-1.3);
\draw (0,-2.1) edge[level] (4.5,-2.1);
\draw (0,-3) edge[level] (4.5,-3);

\node [black!60!green] at (4.8,-1.3) {\scriptsize level 2};
\node [black!60!green] at (4.8,-2.1) {\scriptsize level 1};
\node [black!60!green] at (4.8,-3) {\scriptsize level 0};

\node[fill=white, inner sep = 0pt] at (3,-1.3) {$\rt{c}$};
\node[fill=white, inner sep = 0pt] at (3,-2.1) {$\nort{c}$};
\node[fill=white, inner sep = 0pt] at (0.5,-3) {$\nort{a}$};
\node[fill=white, inner sep = 0pt] at (3.7,-3) {$\nort{b}$};

\tikzstyle{bag}=[fill = blue!15, circle, draw = blue!15]
\tikzstyle{con}=[circle connection bar switch color=from (blue!15) to (blue!15)]

\node[bag, minimum size = 3em] (b) at (2.1,-1.7) {};
\node[bag, minimum size = 3em] (d) at (1.3,-3) {};
\node[bag, minimum size = 3em] (e) at (2.9,-3) {};

\path (b) to[con] (d);
\path (b) to[con] (e);

\tikzset{every node/.style={fill = black, circle, inner sep = 2pt}}
\node[label = {[label distance=-1pt]left:$c$}] (c1) at (2.1,-1.7) {};
\node[label = {[label distance=-1pt]left:$a$}] (d1) at (1.3,-3) {};
\node[label = {[label distance=-1pt]left:$b$}] (e1) at (2.9,-3) {};

\draw (c1.center) ++(90:3pt) circle (3pt);

\end{tikzpicture}
}

\newcommand{\citet}[1]{\textcite{#1}}
\usepackage{newfile}
\newoutputstream{pages}
\openoutputfile{test.pg}{pages}

\crefformat{equation}{(#2#1#3)}
\crefrangeformat{equation}{(#3#1#4) to~(#5#2#6)}
\crefmultiformat{equation}{(#2#1#3)}%
{ and~(#2#1#3)}{, (#2#1#3)}{ and~(#2#1#3)}

\newtheorem{theorem}{Theorem}[section]
\newtheorem{lemma}[theorem]{Lemma}
\newtheorem{proposition}[theorem]{Proposition}
\newtheorem{corollary}[theorem]{Corollary}

\Crefname{theorem}{Theorem}{Theorems}
\Crefname{lemma}{Lemma}{Lemmas}
\Crefname{proposition}{Proposition}{Propositions}
\Crefname{lemma}{Lemma}{Lemmas}
\Crefname{corollary}{Corollary}{Corollaries}

\title[The Complexity of Bidirected Reachability in Valence Systems]{The Complexity of Bidirected Reachability in Valence Systems}

\newcommand{\ouraffiliation}{%
\institution{Max Planck Institute for \\ Software Systems (MPI-SWS)}
\city{Kaiserslautern}
\country{Germany}
}

\author{Moses Ganardi}
\orcid{0000-0002-0775-7781}
\affiliation{\ouraffiliation}
\author{Rupak Majumdar}
\orcid{0000-0003-2136-0542}
\affiliation{\ouraffiliation}
\author{Georg Zetzsche}
\orcid{0000-0002-6421-4388}
\affiliation{\ouraffiliation}

\copyrightyear{2022}
\acmYear{2022}
\setcopyright{rightsretained}
\acmConference[LICS '22]{37th Annual ACM/IEEE Symposium on Logic in Computer Science (LICS)}{August 2--5, 2022}{Haifa, Israel}
\acmBooktitle{37th Annual ACM/IEEE Symposium on Logic in Computer Science (LICS) (LICS '22), August 2--5, 2022, Haifa, Israel}\acmDOI{10.1145/3531130.3533345}
\acmISBN{978-1-4503-9351-5/22/08}


\begin{document}

\begin{abstract}

Reachability problems in infinite-state systems are often subject to extremely
high complexity. %
This motivates the investigation of efficient overapproximations, where we add transitions to
obtain a system in which reachability can be decided more efficiently. We consider
\emph{bidirected} infinite-state systems, where for every transition there is a
transition with opposite effect.  

We study bidirected reachability in the
framework of valence systems, an abstract model featuring finitely many control
states and an infinite-state storage that is specified by a finite graph.
By picking suitable graphs, valence systems can uniformly model counters as in vector addition systems,
pushdowns, integer counters, and combinations thereof.

We provide a comprehensive complexity landscape for bidirected reachability and show that
the complexity drops (often to polynomial
time) from that of general reachability, for almost every storage mechanism
where reachability is known to be decidable. 

\end{abstract}

\maketitle

\section{Introduction}
The reachability problem is one of the most fundamental tasks for
infinite-state systems: It asks for a given infinite-state system and two
configurations $c_1$ and $c_2$, whether it is possible to reach $c_2$ from
$c_1$. It is a basic component in many types of verification (both for safety
and liveness), and has been studied intensively for decades. Unfortunately,
the reachability problem often exhibits prohibitively high complexity or is
subject to long-standing open problems: In \emph{vector addition
systems with states} (VASS), which are systems with counters over the natural
numbers that can be incremented and decremented, reachability is
Ackermann-complete~\cite{DBLP:conf/focs/Leroux21,DBLP:conf/focs/CzerwinskiO21}.
Furthermore, if we allow a pushdown in addition to the counters, we obtain
\emph{pushdown vector addition systems} (PVASS), for which decidability of the
reachability problem is a long-standing open
problem~\cite{DBLP:conf/icalp/LerouxST15}.

One way to circumvent this is to consider overapproximations of the
reachability problem: Instead of deciding reachability in the original system,
we decide reachability in a system that has more transitions. Then, a negative
answer still certifies safety. A promising overapproximation is \emph{bidirected reachability}: We assume that in our system,
for each transition, there is another with opposite effect.
For example, in pushdown reachability, a push of a symbol on the forward edge is reverted by the pop of the symbol on the backward edge.
In a vector addition system, incrementing a counter is reverted by decrementing it by the same amount. 
In addition to its theoretical interest, bidirected reachability is also of practical interest:
for example, several program analysis problems can be formulated or practically approximated
as bidirected pushdown reachability \cite{ChatterjeeCP2018,DBLP:conf/pldi/ZhangLYS13} 
or bidirected interleaved pushdown reachability 
\cite{DBLP:conf/ecoop/XuRS09,DBLP:conf/issta/YanXR11,DBLP:conf/popl/ZhangS17,DBLP:conf/pldi/LiZR20,DBLP:journals/pacmpl/LiZR21,DBLP:journals/pacmpl/KP22}. A remarkable recent result is that bidirected reachability in PVASS with one counter is decidable~\cite{DBLP:journals/pacmpl/KP22}.

An important measure for the utility of bidirectedness is: \emph{Does
bidirectedness reduce complexity?} It is known that bidirected reachability in
VASS is $\EXPSPACE$-complete and in logarithmic space for a fixed number of
counters~\cite{MAYR1982305}.  Bidirected pushdown reachability can be solved in
almost linear time \cite{ChatterjeeCP2018} whereas a truly subcubic algorithm
for pushdown reachability is a long-standing open problem \cite{Chaudhuri}.
However, the general complexity landscape is far from understood. For example,
it is hitherto not even known if bidirected pushdown reachability is
$\compP$-hard.  Further, it is not known whether bidirected reachability in
$\Z$-VASS (which allow counter values to become negative) is $\NP$-complete as
for reachability~\cite{DBLP:conf/rp/HaaseH14}.

\xparagraph{Valence systems} In this paper, we systematically compare the
complexity of bidirected reachability with general reachability. To this end,
we use the framework of \emph{valence systems over graph
monoids}~\cite{DBLP:journals/eatcs/Zetzsche16,Zetzsche2016c,DBLP:journals/iandc/Zetzsche21,DBLP:conf/rp/Zetzsche21}, which are
an abstract model that features finitely many
control states and an infinite storage specified by a finite graph
(with self-loops allowed). By picking a suitable graph, one can obtain
classical infinite-state models: A clique of $d$ unlooped nodes
corresponds to VASS with $d$ counters. If the vertices are looped, one obtains
$\Z$-VASS.  Two isolated unlooped vertices yield pushdown systems.

In order to compare bidirected with general reachability, we focus on storage
mechanisms where general reachability is known to be decidable. These
mechanisms correspond to a well-understood class of graphs studied
in~\cite{DBLP:journals/iandc/Zetzsche21}.  The latter work characterizes a
class of graphs that precisely capture PVASS with one counter.  Then, it is
shown in~\cite{DBLP:journals/iandc/Zetzsche21} that for every graph that avoids
these PVASS graphs as induced subgraphs, reachability is decidable iff the
graph is a transitive forest. This class of graphs is called $\SCpm$.

\xparagraph{Contribution} We provide a comprehensive complexity landscape for
bidirected reachability for valence systems over graphs in $\SCpm$.  For every
graph in $\SCpm$, we prove that bidirected reachability is either
$\compL$-complete or $\compP$-complete.  More generally, we consider the
setting where the graph is part of the input and drawn from a class
$\cG\subseteq\SCpm$: This is the case, e.g.\ when the input can consist of
$\Z$-VASS with an arbitrary number of counters. Even in this general setting,
we obtain an almost complete complexity picture: Then, bidirected reachability
falls within five possible categories: $\compL$-complete, $\ILD$-complete,
$\compP$-complete, in $\EXPTIME$, or $\EXPSPACE$-complete.  Here, the
complexity class $\ILD$ captures those problems that are logspace-reducible to
the feasibility problem for integer linear Diophantine equations (it lies
between $\compL$ and $\compP$).  

To obtain the $\compP$-hardness, we exploit a connection to
algorithmic group theory, which allows us to reduce from the subgroup
membership problem for free groups~\cite{avenhaus1984nielsen}. The same
connection yields undecidability of bidirected reachability if the graph has an
induced 4-cycle via a well-known group theoretic undecidability due to
Mikhailova~\cite{mikhailova1966occurrence}. In particular, bidirected
reachability is undecidable in systems with two pushdowns, which answers an
open problem from~\cite{DBLP:journals/pacmpl/LiZR21}\footnote{This problem was
	answered independently in~\cite{DBLP:journals/pacmpl/KP22}.}.
	This connection also yields an example of a graph where reachability is
	\emph{undecidable}, but bidirected reachability is \emph{decidable}.
	Finally, we can translate back to group
	theory: Our results imply that the subgroup membership problem in graph
	groups over any transitive forest is in $\compP$.

Our results show that the complexity drops for almost every
storage (pure pushdowns without any additional counters are the only
exceptions).  For example, reachability for $\Z$-VASS and for pushdown
$\Z$-VASS (which have a pushdown and several $\Z$-valued counters) are both
known to be $\NP$-complete~\cite{DBLP:conf/rp/HaaseH14,HagueLin2011}, but we
place bidirected reachability in $\PTIME$. The same holds for $\Z$-VASS that,
in addition, have a fixed number of $\N$-counters.  Our results also apply to
stacks where each entry contains several counter values.
Moreover, in addition to such stacks, one can have $\Z$-counters, then build
stacks of such configurations, etc.  Our characterization implies that when the
number of alternations between building stacks and adding $\Z$-counters is not
fixed, then bidirected reachability can still be solved in $\EXPTIME$; in
contrast, general reachability is
$\NEXPTIME$-complete~\cite{HaaseZetzsche2019}.

\xparagraph{Key techniques} These lower complexities are achieved using novel
techniques. The aforementioned connection to algorithmic group theory is
the logspace inter-reducibility with the subgroup membership problem if the
storage graph has self-loops everywhere.  This connection between bidirected
automata and subgroups has been observed in~\cite{lohrey2010automata}; we
provide logspace reductions.

While the connection to group theory is used for lower bounds, our upper bounds
also require new methods. The decidability for reachability in $\SCpm$
in~\cite{DBLP:journals/iandc/Zetzsche21} employs VASS with nested zero
tests~\cite{reinhardt2008reachability,DBLP:conf/mfcs/Bonnet11}, which we manage
to avoid completely. Instead, we rely on results about bidirected VASS
reachability sets~\cite{DBLP:conf/fct/KoppenhagenM97} to essentially eliminate
$\N$-counters in our systems first. However, the main innovation is a
modified approach to reachability in systems with stacks and
$\Z$-counters~\cite{HaaseZetzsche2019,HagueLin2011}. Those algorithms use a
technique of \citet{DBLP:conf/cade/VermaSS05}, which constructs,
given a context-free grammar, an existential Presburger formula for its Parikh
image. While existential Presburger arithmetic is equivalent to systems of
integer linear \emph{inequalities} (where feasibility is $\NP$-complete), we
show that for the new notion of \emph{bidirected grammars}, the Parikh image
can, in some appropriate sense, be described using only \emph{equations}.  This leads
to a $\compP$ upper bound, since feasibility of systems of integer linear
equations is in $\compP$ \cite{DBLP:journals/siamcomp/ChouC82}.

\begin{acks}
The authors are grateful to Markus Lohrey for
discussions about the algorithm in~\cite{lohrey2010automata}.
This work was sponsered in part by the \grantsponsor{DFG}{DFG}{https://www.dfg.de/} under project~\grantnum{DFG}{389792660 TRR 248--CPEC}.
\end{acks}

\section{Bidirected Valence Systems}

\xparagraph{Algebraic Preliminaries}
We assume familiarity with basic notions of monoids, groups, etc.~\cite{Lang}.
A \emph{subgroup} of a group $G$ is a subset of the elements of $G$ that themselves form a group; i.e., it is a subset of elements closed under the 
binary operation as well as inverses.
A subgroup $H$ of a group $G$ can be used to decompose the underlying set of $G$ into disjoint equal-size subsets called \emph{cosets}.
The \emph{left cosets} (resp.\ \emph{right cosets}) of $H$ in $G$ are the sets obtained by multiplying each element of 
$H$ by a fixed element $g$ of $G$: $gH = \set{g\cdot h \mid h \in H}$ (resp.\ $Hg = \set{h \cdot g\mid h\in H}$).
For a subset $S$, we write $\tuple{S}$ for the smallest subgroup containing $S$; this is the set of all elements of $G$ that can be written as finite
products of elements from $S$ and their inverses.
If $\tuple{S} = G$, we say $S$ \emph{generates} $G$ and call the elements of $S$ the \emph{generators} of $G$.

A \emph{presentation} $(\Sigma \mid R)$ of a monoid is a description of a monoid in terms of a set $\Sigma$ of generators and 
a set of binary relations $R \subseteq \Sigma^* \times \Sigma^*$ on the free monoid $\Sigma^*$ generated by $\Sigma$. 
For a set $R \subseteq \Sigma^* \times \Sigma^*$ define the step relation $\to_R$ by $sut \to_R svt$ for all $(u,v) \in R$ and $s,t \in \Sigma^*$.
Define $\equiv_R$ to be the smallest equivalence relation containing $\to_R$.
Then $\equiv_R$ is a \emph{congruence}, meaning that if $u\equiv_R v$, then $sut\equiv_R svt$ for every $s,t\in\Sigma^*$.
The monoid is then presented as the quotient of $\Sigma^*$ by the congruence $\equiv_R$.
For a word $w\in\Sigma^*$, we write $[w]_{\equiv_R}$ for the equivalence class of $w$ under $\equiv_R$.
A {\em commutative semigroup presentation} is a presentation $( \Sigma \mid R )$
where $(xy,yx) \in R$ for all $x \neq y \in \Sigma$.
The {\em word problem for commutative semigroups} asks,
given a commutative semigroup presentation $( \Sigma \mid R )$ and two words $u,v \in \Sigma^*$,
does $u \equiv_R v$ hold?
This problem is known to be $\EXPSPACE$-complete \cite{MAYR1982305}.

\xparagraph{Graph Monoids} 
A \emph{graph} is a tuple $\Gamma=(V, I)$, 
where $V$ is a finite set of vertices and $I \subseteq \{e\subseteq V \mid 1\le |e|\le 2\}$ is a finite set of undirected edges, 
which can be self-loops. 
Thus, if $\{v\}\in I$, we say that $v$ is \emph{looped}; otherwise, $v$ is \emph{unlooped}.
If all nodes of $\Gamma$ are (un)looped, we call $\Gamma$ \emph{(un)looped}.
By $\Gamma^\circ$ and $\Gamma^-$, we denote the graph obtained from $\Gamma$ by adding (removing) self-loops on all vertices.
The edge relation is also called an \emph{independence relation}. 
We also write $uIv$ for $\{u,v\}\in I$.
A subset $U\subseteq V$ is a \emph{clique} if $uIv$ for any two distinct $u,v\in U$.
We say that $U\subseteq V$ is an \emph{anti-clique} if we do not have $uIv$ for any distinct $u,v\in U$.
Given a graph $\Gamma$, we define a monoid as follows. 
We define the alphabet $X_\Gamma=V \cup \bar V$ where $\bar V = \{\bar{v}\mid v\in V\}$.
We define $R_\Gamma = \{ (v \bar v, \varepsilon) \mid v \in V\} \cup \{ (xy,yx) \mid x\in\{u,\bar{u}\}, \, y\in\{v,\bar{v}\}, \, uIv \}$.
We write $\to_\Gamma$ instead of $\to_{R_\Gamma}$ and
$\equiv_\Gamma$ for the smallest equivalence relation containing $\to_\Gamma$.
As observed above, $\equiv_\Gamma$ is a congruence.
In particular, if $v$ has a self-loop, then $\bar{v}v\equiv_\Gamma\varepsilon$.
We define the monoid $\dM\Gamma:=X_{\Gamma}^*/\equiv_\Gamma$.
We write $[w]$ for the equivalence class of $w$ under $\equiv_\Gamma$ and
$1$ for $[\varepsilon]$.
For each word $w\in X_\Gamma^*$, we define its \emph{inverse} $\bar{w}$ as follows.
If $w=v$ for some $v\in V$, then $\bar{w}$ is the letter $\bar{v}$. If $w=\bar{v}$ for $v\in V$, then $\bar{w}=v$. 
Finally, if $w=w_1\cdots w_n$ with $w_1,\ldots,w_n\in X_\Gamma$, then $\bar{w}=\bar{w}_n\cdots \bar{w}_1$.
It is known that $w \equiv_\Gamma \varepsilon$ is witnessed by a derivation $w = w_0 \to_\Gamma  w_1 \to_\Gamma \dots \to_\Gamma  w_n = \varepsilon$~\cite[Equation (8.2)]{Zetzsche2016c}.

\xparagraph{Valence systems and reachability}
Valence systems are an abstract model for studying finite-state transition systems with ``storage''.
They consist of a state transition system on a finite set of states, as well as a monoid that represents an auxiliary storage
and determines which paths in the automata form valid computations in the presence of the auxiliary storage. 
For example, if the underlying storage is a stack, the monoid can encode push and pops and determine computations 
that produce an empty stack.

In this work, we only consider graph monoids as the underlying monoids. 
Many classes of infinite-state systems involving combinations of stacks and counters can be 
modeled as valence systems over graph monoids; see~\cite{Zetzsche2016c} for detailed examples. They have been studied in terms of expressiveness~\cite{BuckheisterZetzsche2013a,Zetzsche2013a}, computing downward closures~\cite{Zetzsche2015a}, and various forms of reachability problems~\cite{DOsualdoMeyerZetzsche2016a,MeyerMuskallaZetzsche2018a,ShettyKrishnaZetzsche2021a,DBLP:journals/iandc/Zetzsche21}, see \cite{DBLP:conf/rp/Zetzsche21} for a survey on the latter.

A \emph{valence system} $\cA$ over a graph $\Gamma$ consists of a finite set of states $Q$, and 
a finite transition relation $\rightarrow\subseteq Q \times X_{\Gamma}^* \times Q$.
We also write a transition $(p,u,q)$ as $p \autsteps[u] q$.
A {\em run} is a sequence $(q_0,u_1,q_1)(q_1,u_2,q_2) \dots (q_{n-1},u_n,q_n)$ of transitions,
also abbreviated $q_0 \autsteps[u] q_n$ if $u=u_1\cdots u_n$.

The \emph{reachability problem} ($\REACH$) for valence systems is the following:
\begin{description}
 	\item[Given] A valence system $\cA$ and states $s,t$ in $\cA$.
 	\item[Question] Is there a run $s\autsteps[w]t$ for some $w\in X_\Gamma^*$ with $[w]=1$?
\end{description}

If the reachability problem is restricted to valence systems over a particular graph $\Gamma$, then we denote the problem by $\REACH(\Gamma)$. 
If we restrict the input systems to a class $\cG$ of graphs, then we write $\REACH(\cG)$. 
For example, if $\cV$ is the class of all unlooped cliques, then $\REACH(\cV)$ is the reachability problem for vector addition systems with states (VASS).

A valence system $\cA$ is \emph{bidirected} if for any transition
$p\autstep[w]q$, we also have $q\autstep[\bar{w}]p$.  The \emph{bidirected
reachability problem} ($\REVREACH$) is the reachability problem where $\cA$ is
restricted to bidirected valence systems. As above, we consider the case where
the system is over a particular graph $\Gamma$, denoted $\REVREACH(\Gamma)$, or
where the graph is drawn from a class $\cG$, denoted $\REVREACH(\cG)$.

\subsection{Decidability Landscape for Reachability}

\xparagraph{PVASS-graphs}
A graph $\Gamma$ is a \emph{PVASS-graph} if it is isomorphic to one of the
following three graphs: 

{\centering
  $ \displaystyle
    \begin{aligned}
	\pnpZero{1} && \pnpOne{1} && \pnpTwo{1} 
    \end{aligned}
  $
\par}
\noindent We say that a graph is \emph{PVASS-free} if it has no
PVASS-graph as an induced subgraph. 
Here, a graph $\Gamma'=(V',I')$ is an \emph{induced subgraph} of $\Gamma=(V,I)$ if there is an injective map $\iota\colon V'\to V$ with $\{\iota(u),\iota(v)\}\in I$ iff $\{u,v\}\in I'$ for $u,v\in V'$.
Observe that a graph $\Gamma$ is PVASS-free
if and only if in the neighborhood of each unlooped vertex, any two vertices
are adjacent.  The terminology stems from the fact that if $\Gamma$ is a
PVASS-graph, then $\REACH(\Gamma)$ is inter-reducible with reachability for
PVASS with one counter~\cite{DBLP:journals/iandc/Zetzsche21}. Whether
reachability is decidable for these is a long-standing open
problem~\cite{DBLP:conf/icalp/LerouxST15} (however, bidirected reachability is
decidable~\cite{DBLP:journals/pacmpl/KP22}).

\xparagraph{Transitive forests}
A graph $\Gamma$ is a \emph{transitive forest} if it can be built as follows.
First, the empty graph is a transitive forest. Moreover, if $\Gamma_1$
and $\Gamma_2$ are transitive forests, then (i)~the disjoint union of
$\Gamma_1$ and $\Gamma_2$ is a transitive forest and (ii)~if $\Gamma$ is the
graph obtained by adding one looped or unlooped vertex $v$ to $\Gamma_1$ so that $v$ is adjacent
to every vertex in $\Gamma_1$, then $\Gamma$ is also a transitive forest.

The complexity of our algorithms will depend on the height of the trees in
transitive forests. Formally, every non-empty transitive forest is either (i)~a
disjoint union of connected transitive forests, or (ii)~has a \emph{universal
vertex}, i.e. a vertex that is adjacent to all other vertices (take the root of
the underlying tree). This induces a successive decomposition of the transitive
forest into smaller ones: For a disjoint union, take the disjoint connected
transitive forests. If there is a universal vertex, remove that vertex to
obtain a smaller transitive forest.

The decomposition is unique up to isomorphism: This is obvious in the
case of a disjoint union. If there are several universal vertices,
then all removals result in isomorphic graphs. 
We define the \emph{height} $h(\Gamma)$ of a
transitive forest $\Gamma=(V,I)$: If $V=\emptyset$, then $h(\Gamma)=0$. If
$\Gamma$ is a disjoint union of connected transitive forests
$\Gamma_1,\ldots,\Gamma_n$, then
$h(\Gamma)=\max\{h(\Gamma_i)\mid i\in[1,n]\}+1$. If $\Gamma$
has a universal vertex whose removal leaves $\Gamma'$, then
$h(\Gamma)=h(\Gamma')$.

Among the PVASS-free graphs, it is well-understood when reachability
is decidable:
\begin{theorem}[\cite{DBLP:journals/iandc/Zetzsche21}]\label{pvass-free-reachability}
	Let $\Gamma$ be PVASS-free. Then $\REACH(\Gamma)$ is decidable iff $\Gamma$ is a transitive forest.
\end{theorem}

By $\SCpm$, we denote the class of PVASS-free graphs $\Gamma$ that are
transitive forests. The abbreviation $\SCpm$ reflects that valence systems over
$\SCpm$ are equivalent to \emph{stacked counter machines}, as explained in
\cref{sec:main-results}.  Hence, \cref{pvass-free-reachability} says that for
every graph in $\SCpm$, the reachability problem is decidable.  Moreover, for
every graph $\Gamma$ outside of $\SCpm$, either $\Gamma$ contains a PVASS-graph
(thus showing decidability of $\REACH(\Gamma)$ would in particular solve a
long-standing open problem) or $\REACH(\Gamma)$ is known to be undecidable.
Therefore, $\SCpm$ is the largest class of graphs $\Gamma$ for which
$\REACH(\Gamma)$ is currently known to be decidable.

\section{Main Results}\label{sec:main-results}

We assume familiarity with the basic complexity classes  $\compL$
(deterministic logspace), $\compP$ (deterministic polynomial time), $\NP$ (non-deterministic polynomial time), $\EXPTIME$
(deterministic exponential time), $\NEXPTIME$ (non-deterministic exponential
time), and $\EXPSPACE$ (exponential space). 
By $\ILD$, we denote the class of problems that are 
logspace-reducible to the problem of solvability 
of integer linear Diophantine equations (ILD): 
\begin{description}
	\item[Given] A matrix $\bA\in\Z^{m\times n}$ and a vector $\bb\in\Z^m$.
	\item[Question] Is there a vector $\bx\in\Z^n$ with $\bA\bx=\bb$?
\end{description}
It is well known that ILD is solvable in polynomial time.
\begin{theorem}[\cite{DBLP:journals/siamcomp/ChouC82}]\label{ild-in-p}
	ILD is solvable in polynomial time.
\end{theorem}
See \cite[Theorems 1 and 13]{DBLP:journals/siamcomp/ChouC82}.
In particular, the class $\ILD$ lies in between $\compL$ and $\compP$.
The exact complexity of ILD seems to be open~\cite{DBLP:journals/cc/AllenderBO99}.
It is conceivable that $\ILD$ coincides with $\compL$ or $\compP$ or that
it lies strictly in between. Hence, we have the inclusions
\[ \compL\subseteq\ILD\subseteq\compP\subseteq\NP\subseteq\EXPTIME\subseteq\NEXPTIME\subseteq\EXPSPACE. \]

In order to formulate our main result about the complexity of
$\REVREACH$, we need some terminology.   We say that
$\cG$ is \emph{UC-bounded} if there is a $k$ such that for every $\Gamma$ in
$\cG$, every unlooped clique in $\Gamma$ has size at most $k$. Otherwise, it is
called \emph{UC-unbounded}. Similarly, \emph{LC-bounded} (\emph{LC-unbounded},
respectively) if the same condition holds for looped cliques. We say that $\cG$
is \emph{height-bounded} if there is a $k$ with $h(\Gamma)\le k$ for every
$\Gamma$ in $\cG$. Otherwise, $\cG$ is \emph{height-unbounded}.
We now present an almost complete complexity classification of
$\REVREACH(\cG)$, where $\cG$ is a subclass of $\SCpm$.  Here, we assume that
$\cG$ is closed under taking induced subgraphs.  This is a mild assumption that
only rules out some pathological exceptions.

\begin{theorem}[Classification Theorem]\label{main-result}
	Let $\cG\subseteq \SCpm$ be closed under induced subgraphs. Then $\REVREACH(\cG)$ is
	\begin{enumerate}
		\item\label{main-bounded-cliques} $\compL$-complete if $\cG$ consists of cliques of bounded size,
		\item\label{main-lc-unbounded-cliques} $\ILD$-complete if $\cG$ consists of cliques, is UC-bounded, and LC-unbounded,
		\item\label{main-bounded-height} $\compP$-complete if $\cG$ contains a graph that is not a clique, and $\cG$ is UC-bounded and height-bounded,
		\item\label{main-unbounded-height} in $\EXPTIME$ if $\cG$ is UC-bounded and height-unbounded, 
		\item\label{main-uc-unbounded} $\EXPSPACE$-complete otherwise.
	\end{enumerate}
\end{theorem}	
From \cref{main-result}, we can deduce our dichotomy for individual graphs $\Gamma$:
Take as $\cG$ the set of graphs containing $\Gamma$ and its induced subgraphs.
Then $\cG$ is UC-bounded, LC-bounded, and height-bounded and thus falls into
case (\ref{main-bounded-cliques}) or (\ref{main-bounded-height})
above. 
\begin{corollary}[Dichotomy for $\REVREACH$]
	Let $\Gamma\in\SCpm$ be a graph. If $\Gamma$ is a clique, then
	$\REVREACH(\Gamma)$ is $\compL$-complete. Otherwise, the problem
	$\REVREACH(\Gamma)$ is $\compP$-complete.
\end{corollary}

\begin{figure}[t]
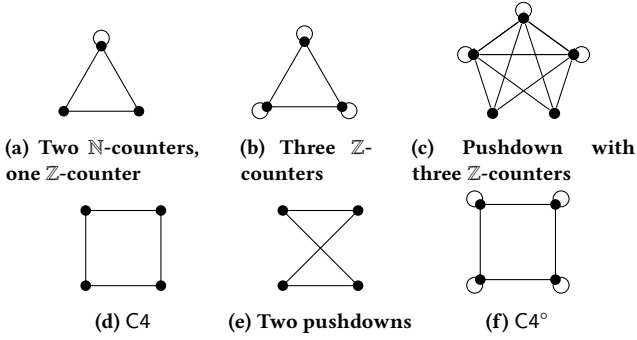

	\subcaptionbox{Two $\N$-counters, one $\Z$-counter\label{graph-onentwoz}}[0.3\columnwidth]{\twononezcounter{1}}
\hspace{0.05\columnwidth}
\subcaptionbox{Three $\Z$-counters\label{graph-threez}}[0.2\columnwidth]{\threezcounters{1}}
\hspace{0.05\columnwidth}
\subcaptionbox{Pushdown with three $\Z$-counters\label{graph-pthreez}}[0.35\columnwidth]{\pushdownzvass{1}}

\subcaptionbox{$\Cfour$\label{graph-cfour}}[0.2\columnwidth]{\unloopedcycle{1}}
\subcaptionbox{Two pushdowns\label{graph-cfourdifferent}}[0.4\columnwidth]{\twopushdowns{1}}
\subcaptionbox{$\Cfourlooped$\label{graph-cfourlooped}}[0.2\columnwidth]{\loopedcycle{1}}
\caption{Example graphs for storage mechanisms}
\end{figure}

We complement \cref{main-result} with the following undecidability result. The graph $\Cfour$ is shown in \cref{graph-cfour}.

\begin{restatable}{theorem}{undecidableunderlyingcfour}\label{undecidable-underlying-c4}
	If $\Gamma^-$ contains $\Cfour$ as an induced subgraph, then $\REVREACH(\Gamma)$ is undecidable.
\end{restatable}
In the case $\Gamma=\Cfour$, undecidability also follows from an independent
result of \citet{DBLP:journals/pacmpl/KP22} that uses different techniques.

\xparagraph{Intuition on graph classes} Let us phrase our results in terms
of infinite-state systems.  In~(\ref{main-bounded-cliques})
and~(\ref{main-lc-unbounded-cliques}), we have cliques. A clique with $d$
unlooped vertices and $e$ looped vertices corresponds to counter machines with
$d$-many $\N$-counters and $e$-many $\Z$-counters (see \cref{graph-onentwoz} with $d=2$, $e=1$).  Thus,
(\ref{main-bounded-cliques}) and~(\ref{main-lc-unbounded-cliques}) show
$\compL$-completeness for for fixed $d$ and $e$, and $\ILD$-completeness if
only $d$ is fixed. This is in contrast to reachability: Even for fixed $d$,
reachability in machines with $d$-many $\N$-counters can be
non-elementary~\cite{DBLP:conf/focs/Leroux21,DBLP:conf/focs/CzerwinskiO21}.
Moreover, for $d=0$, these graphs correspond to $\Z$-VASS (see \cref{graph-threez} for $e=3$). Thus, we show: For a
fixed number of counters in $\Z$-VASS, the complexity drops from $\NL$ (for
general reachability) to $\compL$. If the number of counters in $\Z$-VASS is
not fixed, then we have a drop from $\NP$~\cite{DBLP:conf/rp/HaaseH14} to
$\ILD\subseteq\compP$. 

In (\ref{main-bounded-height}), we go beyond just counters.  
Recall the recursive structure of every graph in $\SCpm$.
In terms of storage, taking a disjoint union
$\Gamma_1\cup\Gamma_2$ is the same as \emph{building stacks}: One obtains a
stack, where each entry of the stack is either a storage content of $\Gamma_1$
or $\Gamma_2$.  Moreover, adding a universal looped vertex corresponds to \emph{adding
a $\Z$-counter}~\cite{DBLP:journals/eatcs/Zetzsche16,Zetzsche2016c}: This means, in
addition to the configuration of the previous model, we also have a new
$\Z$-counter. 
After decomposing $\Gamma$ further and further, we are left with
a clique of unlooped vertices.  Therefore, valence systems over $\SCpm$ are
called \emph{stacked counter machines}: We start with a number of $\N$-counters
and then alternate between building stacks and adding
$\Z$-counters%
	.
The height $h(\Gamma)$ is the number of alternations between the
two steps (building stacks and adding $\Z$-counters).  
	
For example, if $\Gamma_k$ has two non-adjacent nodes $v_1,v_2$ and looped
nodes $u_1,\ldots,u_k$ such that each $u_i$ is adjacent to all other vertices
(see \cref{graph-pthreez} for $\Gamma_3$), then valence systems over
$\Gamma_k$ are pushdown $\Z$-VASS with $k$ counters.  Thus,
(\ref{main-bounded-height}) says that in stacked counter automata, with a
bounded number of $\N$-counters and a bounded number of alternations (but
arbitrarily many $\Z$-counters!), bidirected reachability is still in $\compP$.
This is again a striking complexity drop: In pushdown $\Z$-VASS, reachability
is $\NP$-complete~\cite{HagueLin2011}, and the same is true for any fixed
number of alternations (if there are no
$\N$-counters)~\cite{HaaseZetzsche2019}.  In fact, $\NP$-hardness holds for
pushdown $\Z$-VASS already for a single counter~\cite{HaaseZetzsche2019}.
Moreover, our $\compP$ upper bound still holds with a bounded
number of $\N$-counters.

In (\ref{main-unbounded-height}), not even the number of alternations is fixed.
We obtain an $\EXPTIME$ upper bound, which is again a drop from reachability:
In stacked counter automata (without $\N$-counters), reachability is
$\NEXPTIME$-complete~\cite{HaaseZetzsche2019}.

In (\ref{main-uc-unbounded}), we show that if neither the alternations
nor the number of $\N$-counters is fixed, $\REVREACH$ is $\EXPSPACE$-complete. This strengthens $\EXPSPACE$-completeness
of $\REVREACH$ in VASS. Again, general reachability has much higher
complexity: Since our model includes VASS, the problem is
Ackermann-hard~\cite{DBLP:conf/focs/Leroux21,DBLP:conf/focs/CzerwinskiO21}
and could be even higher: The algorithm for reachability
for general $\SCpm$ in~\cite{DBLP:journals/iandc/Zetzsche21} uses VASS with
nested zero tests, for which no complexity upper bound is known.

Finally, our undecidability result (\cref{undecidable-underlying-c4}) says in
particular that reachability in bidirected two-pushdown machines is
undecidable: By drawing $\Cfour$ as in \cref{graph-cfourdifferent}, one can see
that it realizes two stacks (the left two nodes act as a stack, and the right
two nodes act as a stack). %

\section{$\REVREACH$ and subgroup membership}

If $\Gamma$ is looped, then $\dM\Gamma$ is a group. 
The groups of this form are called \emph{graph groups} or \emph{right-angled Artin groups} 
and have been studied intensively over the last decades \cite{lohrey2008submonoid,DiekertM06,lohrey2007graph,LohreyZ18}, in part because of their rich subgroup structure (see, e.g.~\cite{Wise2012}).
The \emph{subgroup membership problem} ($\SUBMEM$) is the following.
\begin{description}
	\item[Given] A looped graph $\Gamma$, words $w_1,\ldots,w_k,w\in X_\Gamma^*$.
	\item[Question] Does $[w]\in\langle [w_1],\ldots,[w_k]\rangle$?
\end{description}
If the input graph $\Gamma$ is fixed, we write $\SUBMEM(\Gamma)$. If $\Gamma$
is drawn from a class $\cG$, then we write $\SUBMEM(\cG)$.  Surprisingly,
describing the class of graphs $\Gamma$ for which $\SUBMEM(\Gamma)$ is
decidable is a longstanding open problem~\cite{lohrey2013rational}.

Our first observation is that if $\Gamma$ is looped, then the complexity of $\REVREACH(\Gamma)$ 
matches that of subgroup membership over $\dM\Gamma$. The connection between subgroups and bidirected valence automata (albeit under different names) is a prominent theme in group theory. It is implicit in the well-known concept of Stallings graphs and was used by \citet{lohrey2010automata} in decidability results. We show that the conversion can be done in logspace, in both directions.

\begin{restatable}{theorem}{revreachsubmem}\label{revreach-submem}
	For any looped graph $\Gamma$, the problems $\REVREACH(\Gamma)$ and $\SUBMEM(\Gamma)$ are logspace inter-reducible.
\end{restatable}

Reducing $\SUBMEM(\Gamma)$ to $\REVREACH(\Gamma)$ is easy:
To test whether $[w]$ is contained in $\langle [w_1], \dots, [w_k] \rangle$
we construct a bidirected valence system $\cA$ with two states $s,t$
and the transitions $s \autsteps[\bar w] t$
and $t \autsteps[w_i] t$ for all $1 \le i \le k$, and the reverse transitions.
Then $[w] \in \langle [w_1], \dots, [w_k] \rangle$ holds if and only if there exists $u \in X_\Gamma^*$
with $s \autsteps[u] t$ and $[u] = 1$.
For the converse, the following \lcnamecref{compute-coset} shows that we can compute in logspace a coset representation of
$\{ [w] \in \dM \Gamma \mid s \autsteps[w] t \}$.

\begin{restatable}{lemma}{computecoset}
	\label{compute-coset}
	Given a looped graph $\Gamma$, a bidirected valence system $\cA$ over $\Gamma$ and states $s,t$ from $\cA$,
	one can compute words $w_0, w_1, \dots, w_n \in X_\Gamma^*$ in logspace
	such that $\{ [w] \in \dM \Gamma \mid s \autsteps[w] t \}= [w_0] \langle [w_1], \dots, [w_n] \rangle$
\end{restatable}

\Cref{compute-coset} reduces reachability to testing membership of $[\bar w_0]$ in the set $\langle [w_1], \dots, [w_n] \rangle$, which is an instance of $\SUBMEM$.
To show \cref{compute-coset}, we compute in logspace a spanning tree of the automaton (using \cite{DBLP:journals/jacm/Reingold08}). Then $w_0$ is obtained from the unique path from $s$ to $t$ in the tree. 
The words $w_1,\ldots,w_n$ are obtained as labels of \emph{fundamental cycles}~\cite{biggs1997algebraic}: These are cycles consisting of one edge outside the tree and a path inside the tree.

\begin{figure}[t]
	\subcaptionbox{\label{figure-example-zvass}}[0.59\columnwidth]{\zvass{1}{0.7}}
	\subcaptionbox{\label{graph-pfour-loops}}[0.4\columnwidth]{\loopedpathFour{1}{0.7}}
	\caption{$\Z$-VASS for \cref{example-zvass} (reverse edges are omitted) and the graph $\Pfourlooped$}
\end{figure}

\begin{example}\label{example-zvass}
\Cref{compute-coset} can be directly used to show that the problem
$\REVREACH(\cG)$ is in $\ILD$ if $\cG$ consists of looped cliques. In other
words, bidirectional reachability for $\Z$-VASS is in $\ILD$, whereas standard
reachability for $\Z$-VASS is $\NP$-complete~\cite{DBLP:conf/rp/HaaseH14}.
As an example, consider the bidirected $\Z$-VASS in \cref{figure-example-zvass}. We compute the value of an arbitrary $s$-$t$-path, e.g. $w_0 = 5+2 = 7$,
a spanning tree (the bold edges), and the fundamental cycles
with values $w_1 = 5+3-2=6$ and $w_2 = 5+2-1-2 = 4$.
Since $7 + 6x + 4y = 0$ has no integer solution,
there is no $s$-$t$-path with value 0.
\end{example}

The connection between $\REVREACH$ and $\SUBMEM$ can be used to identify a
graph for which the problem $\REACH(\Gamma)$ is undecidable but $\REVREACH(\Gamma)$ is
decidable: Let $\Pfourlooped$ be the graph in \cref{graph-pfour-loops}. As shown by~\citet{lohrey2008submonoid}, the problem $\REACH(\Pfourlooped)$ is undecidable. However, a
result by \citet{kapovich2005foldings} implies that if a graph $\Gamma$ is
looped and chordal, then $\SUBMEM(\Gamma)$ is decidable (a simpler proof was
then given by Lohrey and Steinberg~\cite{lohrey2010automata}).  Thus, since
$\Pfourlooped$ is looped and chordal, $\SUBMEM(\Pfourlooped)$, and hence
$\REVREACH(\Pfourlooped)$, is decidable.

For the proof of \Cref{undecidable-underlying-c4},
we rely on an undecidability result of \citet{mikhailova1966occurrence}:
\begin{theorem}[\cite{mikhailova1966occurrence}]\label{mikhailova}
	$\SUBMEM(\Cfourlooped)$ is undecidable.
\end{theorem}
The graph $\Cfourlooped$ is shown in \Cref{graph-cfourlooped}.  The proof of
\cref{mikhailova} is a simple (but not obvious) reduction of the word problem
of any finitely presented group (since \cite{mikhailova1966occurrence} is in
Russian, we refer to \cite[Chapter~IV, Lemma~4.2]{LyndonSchupp} or \cite[proof of Lemma~4.2]{LohreyZ18} for short expositions of this reduction).  Since there are
finitely presented groups with an undecidable word
problem~\cite[Chapter~IV, Theorem~7.2]{LyndonSchupp}, this implies undecidability of
$\SUBMEM(\Cfourlooped)$.

Via \cref{revreach-submem}, \cref{mikhailova} implies that $\REVREACH(\Cfourlooped)$ is undecidable. 
Using standard arguments one can transfer the undecidability of $\SUBMEM(\Cfourlooped)$ to $\SUBMEM(\Cfour)$,
and thus obtain \Cref{undecidable-underlying-c4}.

\label{sec:subgp}

\section{Lower Bounds}\label{sec:lower-bounds}

We now prove the lower bounds in \cref{main-result}.

\xparagraph{$\compL$-hardness}
We can reduce from the reachability on undirected graphs to the problem $\REVREACH(\Gamma)$ for any $\Gamma$
by replacing each undirected edge by bidirected $\varepsilon$-labeled transitions.
Since the former problem is $\compL$-complete under $\mathsf{AC}^0$ many-one reductions \cite{DBLP:journals/cc/AlvarezG00,DBLP:journals/jacm/Reingold08},
so is $\REVREACH(\Gamma)$.

\xparagraph{$\ILD$-hardness}
Next, assume that $\cG$ is LC-unbounded.
Observe that ILD is the subgroup membership problem over $\Z^m$, where $m$ is part of the input,
and hence log-space reducible to $\SUBMEM(\cG)$.
Since the equivalence of $\SUBMEM$ and $\REVREACH$ for looped graphs is uniform in the graph $\Gamma$ (see \Cref{sec:subgp}),
$\REVREACH(\cG)$ is $\ILD$-hard.

\xparagraph{$\compP$-hardness}
Suppose $\Gamma$ has two non-adjacent vertices $u,v$.
In $\Gamma^\circ$, $u$ and $v$ generate a free group over two generators,
for which subgroup membership is $\compP$-hard~\cite{avenhaus1984nielsen}.
By \cref{revreach-submem}, this implies $\compP$-hardness of $\REVREACH(\Gamma^\circ)$.
Using standard arguments, $\REVREACH(\Gamma^\circ)$ reduces to $\REVREACH(\Gamma)$.
Thus, $\REVREACH(\Gamma)$ is $\compP$-hard.

\xparagraph{$\EXPSPACE$-hardness}
We reduce from the word problem over commutative semigroups, which is $\EXPSPACE$-hard \cite{MAYR1982305}.
Since $\cG$ is UC-unbounded and closed under induced subgraphs, it contains an unlooped clique $\Gamma$ of size $|\Sigma|$.
We can assume that $\Sigma$ is its node set.
Let $\cA$ be the bidirected valence system over $\Gamma$ with three states $q_0, q_1, q_2$,
the transitions $q_0 \autsteps[\bar u] q_1$, $q_1 \autsteps[v] q_2$,
and the transitions $q \autsteps[\bar x y] q$ for all $(x,y) \in R$,
and their reverse transitions.
Then $u \equiv_R v$ holds if and only if $q_0 \autsteps[w] q_2$ for some $w \in X_\Gamma^*$ with $[w] = 1$.

\section{Upper Bounds I: $\compL$ and $\ILD$} \label{sec:upper-bound-ild}

In this section we will study $\REVREACH(\cG)$ for classes $\cG$ of cliques,
and prove the $\compL$ and $\ILD$ upper bounds from \Cref{main-result}.
If $\Gamma$ is an unlooped clique then $\REVREACH(\Gamma)$ is
the reachability problem over {\em bidirected vector addition systems with states}
or, equivalently, the word problem for commutative semigroups \cite{MAYR1982305}.

Fix a bidirected valence system $\cA = (Q,\to)$ over a clique $\Gamma = (V,I)$.
Let $U$ and $L$ be the sets of unlooped and looped vertices in $\Gamma$, respectively,
and let $s,t \in Q$.
We can view the unlooped vertices as $\N$-counters,
which may not fall below zero,
and the looped vertices as $\Z$-counters.
Formally, the monoid $\dM \Gamma$ is isomorphic to $\B^U \times \Z^L$
where $\B$ is the {\em bicyclic monoid}, i.e. the set $\N^2$ equipped with the associative operation
$(a^-,a^+) \oplus (b^-,b^+) = (a^-+b^--\min(a^+,b^-), a^++b^+-\min(a^+,b^-))$.
We identify each element $(0,a^+)$ in $\B$ with $a^+\in\N$ and each $(a^-,0)$ with $-a^-\in-\N$.
Let $\Phi \colon X_\Gamma^* \to \Z^L$ be the function which computes the value of the $\Z$-counters,
i.e. $\Phi(w)(v) = |w|_v - |w|_{\bar v}$ for all $w \in X_\Gamma^*$, $v \in L$.
Let $\Psi \colon X_\Gamma^* \to \B^U$ be 
the morphism defined by $\Psi(v)(v) = (0,1)$ and $\Psi(\bar v)(v) = (1,0)$ for all $v \in U$,
and $\Psi(x)(v) = (0,0)$ for all $x \in (V \cup \bar V) \setminus \{v, \bar v\}$.
Then $\dM \Gamma \to \B^U \times \Z^L$, $[w] \mapsto (\Psi(w),\Phi(w))$ for $w\in X_\Gamma^*$, is an isomorphism.

\xparagraph{$\N$-counters}
If we consider only $\N$-counters, i.e. $\cA$ is a bidirected VASS,
the reachability problem and the structure of reachability sets are well understood:

\begin{restatable}[\cite{MAYR1982305}]{lemma}{revvasspath}
	\label{rev-vass-path}
	One can decide in deterministic space $2^{O(|U|)} \cdot \log \|\cA\|$
	whether there exists a path $s \autsteps[w] t$ with $\Psi(w) = \bzero$
	and, if so, return such a path.
\end{restatable}

We define $\Reach(p,q) = \{ \Psi(w) \mid p \autsteps[w]_\cA q \} \cap \N^U$ for any states $p,q \in Q$.
\Cref{rev-vass-reach} lets us partition the set $U$ of $\N$-counters into (exponentially) bounded components $B$ and
{\em simultaneously} unbounded components $U \setminus B$.
In the proof, we employ a representation of $\Reach(p,q)$ as a {\em hybrid linear set}
$\bigcup_{i=1}^m \{ \bb_i + \sum_{j=1}^\ell \lambda_j \bp_j \mid \lambda_1, \dots, \lambda_\ell \in \N \}$
with $\bb_i, \bp_j \in \N^\Sigma$ as shown in~\cite{DBLP:conf/fct/KoppenhagenM97}. See \cref{appendix-upper-bound-n-counters}.
\begin{restatable}{lemma}{revvassreach}
	\label{rev-vass-reach}
	One can compute in deterministic space $2^{O(|U|)} \cdot \log \|\cA\|$ a set $B \subseteq U$ and a number $b \le 2^{O(|U|)} \cdot \|\cA\|$
	such that for all $q \in Q$ we have:
	\begin{itemize}
	\item $\bv(u) \le b$ for all $\bv \in \Reach(s,q)$ and $u \in B$,
	\item for every $c \in \N$ there exists $\bv \in \Reach(s,q)$ with $\bv(u) \ge c$ for all $u \in U \setminus B$.
	\end{itemize}
\end{restatable}

\xparagraph{Adding $\Z$-counters}
We will decide reachability in $\cA$ by computing a representation for
\[
	\Eff(s,t) = \{\Phi(w) \mid s \autsteps[w]_\cA t, \, \Psi(w) = \bzero\}.
\]
Since $[w] = 1$ if and only if $\Psi(w) = \bzero$ and $\Phi(w) = \bzero$
we only need to test $\bzero \in \Eff(s,t)$.
Observe that $\Eff(s,t)$ is either empty or a coset $\Eff(s,t) = \bu + \Eff(s,s)$ for any $\bu \in \Eff(s,t)$.
Using \Cref{rev-vass-path} we can test whether there is a path $s \autsteps[w] t$ with $\Psi(w) = \bzero$,
and, if so, we find $\bu := \Phi(w) \in \Eff(s,t)$.
It remains to compute a representation of the subgroup $\Eff(s,s)$.

\begin{restatable}{proposition}{computeeff}
	\label{compute-eff}
	In deterministic space $2^{O(|U|)} \cdot \log \|\cA\|$, we can compute
	$L'\supseteq L$, $|L'| \le |\Gamma|$, and $\bv_1, \dots, \bv_n \in \Z^{L'}$ 
	with $\Eff(s,s) = \{ \bv|_L \mid \bv \in \langle \bv_1, \dots, \bv_n \rangle,~\bv|_{L' \setminus L} = \bzero \}$.
\end{restatable}

\noindent\emph{Proof Sketch.}
We eliminate the $\N$-counters in $\cA$ by replacing the unbounded counters by $\Z$-counters and
maintaining the bounded ones in the finite state.
We obtain a system $\cA'$ over a looped clique, for which the statement follows from \Cref{compute-coset}.
Clearly, every valid $\cA$-run translates into a valid $\cA'$-run.
However, in a valid $\cA'$-run, the formerly unbounded counters can now take negative values.
We can prepend a cycle run which takes sufficiently large values in the unbounded components,
and append the reverse cycle run to cancel its effect.
This ensures that all counters remain nonnegative during the run, which can thus be translated into an $\cA$-run.
\qed

We can now prove the upper bounds for cases (\ref{main-bounded-cliques}) and (\ref{main-lc-unbounded-cliques})
in \cref{main-result}.
As explained above we can test in exponential space (logarithmic space if $\cG$ consists of cliques of bounded size)
whether $\Eff(s,t)$ is nonempty and, if so compute a vector $\bu \in \Eff(s,t)$.
It remains to test $\bzero \in \Eff(s,t)$, which is equivalent to $- \bu \in \Eff(s,s)$.
Using \Cref{compute-eff}, this can be solved in $\ILD$.
If $|\Gamma|$ is bounded then this is in $\compL$ (even $\mathsf{TC}^0$) \cite[Theorem~13]{DBLP:conf/stacs/ElberfeldJT12}.

\begin{restatable}{proposition}{ucboundedild}
	\label{uc-bounded-ild}
	If $\cG$ is a UC-bounded class of cliques,
	then the problem $\REVREACH(\cG)$ belongs to $\ILD$.
	If $\cG$ is a class of cliques of bounded size,
	then $\REVREACH(\cG)$ belongs to $\compL$.
\end{restatable}

Similarly we can prove the following result, which will be used in \cref{sec:upper-bound-p}.

\begin{restatable}{theorem}{computecosetcliques}\label{compute-coset-cliques}
	Given a clique $\Gamma \in \cG$, a bidirected valence system $\cA = (Q,\to)$ over $\Gamma$, and states $s,t \in Q$,
	one can test in exponential space (polynomial time if $\cG$ is UC-bounded)
	if $\Eff(s,t)$ is nonempty and, if so,
	compute a coset representation $\bu + \langle \bv_1, \dots, \bv_n \rangle$ for $\Eff(s,t)$. 
\end{restatable}

\section{Upper Bounds II: $\compP$, $\EXPTIME$, $\EXPSPACE$}
\label{sec:upper-bound-p}

We now prove the upper bounds of (\ref{main-bounded-height}) and
(\ref{main-unbounded-height}) in \cref{main-result}. Let
$\SCpm_{d}$ be the class of graphs in $\SCpm$ where each unlooped clique has
size at most $d$, and let $\SCpm_{d,\ell}$ be the class of those graphs
$\Gamma$ in $\SCpm_d$ with $h(\Gamma)\le \ell$. We prove that
(i)~$\REVREACH(\SCpm)$ is in $\EXPSPACE$, (ii)~for every $d$,
$\REVREACH(\SCpm_d)$ is in $\EXPTIME$ and (iii)~for every $d,\ell\ge 0$,
the problem $\REVREACH(\SCpm_{d,\ell})$ is in $\compP$.

\xparagraph{Key ideas and outline}
The graphs in $\SCpm_{d,\ell}$ correspond to the following storage mechanisms.
The simplest case, graphs in $\SCpm_{d,0}$, consist in collections of $d$ counters with values in $\N$, or $d$-VASS. 
The storage mechanism corresponding to $\SCpm_{d,1}$ is a stack, where each entry contains $d$ such $\N$-counters.
In addition, they have $\Z$-counters in parallel to such stacks (note that using
more $\N$-counters in addition to a stack would be a PVASS, which we are
avoiding). Thus, $\SCpm_{d,1}$ corresponds to ``stacks of $d$ $\N$-counters, plus $\Z$-counters''. 
The \emph{building stacks} and \emph{adding $\Z$-counters} can now be iterated to obtain the storage mechanisms for $\SCpm_{d,\ell}$ for higher $\ell$: 
A storage mechanism with stacks, where each
entry is a configuration of an $\SCpm_{d,1}$ storage, plus additional
$\Z$-counters, is one in $\SCpm_{d,2}$, etc.

Our algorithms work as follows. We first translate a valence system over
$\SCpm_{d,\ell}$ into a type of grammar that we call $k$-grammar. This
translation is quite similar to existing approaches to show that general
reachability over $\SCpm_{1,\ell}$ is decidable~\cite{lohrey2008submonoid} and
even $\NP$-complete~\cite{HaaseZetzsche2019}. The difference to the previous
approaches is that for \emph{bidirected} valence systems, we can construct
\emph{bidirected $k$-grammars}, which are $k$-grammars that satisfy a carefully
chosen set of symmetry conditions. In this translation, we also eliminate the
$d$ $\N$-counters: Roughly speaking, \cref{compute-coset-cliques} lets us
replace them by a gadget that preserves their interaction with the
$\Z$-counters.

After the reduction to grammars, the algorithm in
\cite{HaaseZetzsche2019} (and also the $\NP$ algorithm for pushdown
$\Z$-VASS~\cite{HagueLin2011}, which correspond to $\SCpm_{0,1}$) employs
a result by \citet{DBLP:conf/cade/VermaSS05}. It says that given a context-free
grammar, one can construct in polynomial time an existential Presburger formula
for its Parikh image.
Existential Presburger arithmetic corresponds to systems of
integer linear \emph{inequalities}, for which feasibility is $\NP$-complete.  
(To be precise, \cite{HaaseZetzsche2019} uses 
Presburger arithmetic extended with Kleene stars, to deal with $\SCpm_{0,\ell}$ for
$\ell>1$.) 
We show that for \emph{bidirected} grammars, we have an analogous result that yields
systems of linear integer \emph{equations}, which can be solved in $\compP$.

This analogue of \cite{DBLP:conf/cade/VermaSS05} lets us express emptiness of
bidirected grammars using \emph{coset circuits}, which are (newly introduced)
circuits to describe cosets $C\subseteq\Z^m$. 
They can be seen as compressed
representations of integer linear equation systems. In the final step, we
observe that those coset circuits translate into exponential-sized linear
equation systems. Moreover, for every fixed $\ell$, on input from
$\SCpm_{d,\ell}$, the coset circuits have bounded depth, resulting in
polynomial-sized equation systems.

\xparagraph{Grammars}
Let us now define $k$-grammars.
They have a set $N$ of nonterminal symbols (which can be rewritten by a grammar rule) and a set $T$ of terminal symbols (which can not be rewritten).
We allow letters in $T$ to occur negatively, but the letters in $N$ can only occur non-negatively. 
Hence, we derive vectors in $\N^N+\Z^T$, or equivalently, 
vectors $\bu\in\Z^{N\cup T}$, where $\bu(a)\ge 0$ for each $a\in N$.
We say that a vector $\bv\in\N^N$ \emph{occurs} in such a $\bu$ if $\bu(a)\ge \bv(a)$ for every $a\in N$.

\subsubsection*{Formal definition}
A \emph{$k$-grammar} is a tuple $G=(N,T,P)$, where
\begin{itemize}
	\item $N$ is a finite alphabet of \emph{nonterminals}, which is a disjoint union $N=\bigcup_{i=0}^k N_i$,
	\item $T$ is a finite alphabet of \emph{terminals}, which is a disjoint union $T=\bigcup_{i=0}^k T_i$,
	\item $P$ is a finite set of \emph{productions} of one of two forms:
		\begin{itemize}
			\item $a\to b$ with $a\in N_i$, $b\in N_{j}$, $|i-j|=1$.
			\item $a\to \bu$ with $a\in N_0$ and $\bu\in \N^{N_{0}}+\Z^{T}$.
		\end{itemize}	
\end{itemize}
The letters in $T_i$ ($N_i$) are the \emph{level-$i$ (non)terminals}.
We write 
\[ N_{[i,j]}=\bigcup_{i\le r\le j} N_r, \]
and
analogously for $T_{[i,j]}$. 
Define $R\subseteq N$ as the subset
of $a\in N$ that appear on some right-hand side of the grammar. We also
write $R_i= R\cap N_i$ and use the notation $R_{[i,j]}$ as for $N$ and $T$.

\subsubsection*{Derivations}
In these grammars, derivations produce vectors in $\N^{N}+\Z^{T}$ instead of words.\footnote{
	Usually, grammars derive strings rather than vectors.
	We could also develop our theory using grammars that generate
	strings, but since we are only interested in the generated strings up to
	reordering of letters, we simplify the exposition by working directly with vectors. 
}
A \emph{configuration} is a vector $\bv\in\N^{N}+\Z^T$. For $i\in[0,k]$, we define the $i$-derivation
relation $\deriv[i]$ as follows. We begin with defining $\deriv[0]$ and then define each $\deriv[i]$ based on $\deriv[i-1]$.
For configurations $\bv,\bv'\in\N^{N}+\Z^{T}$,
we have $\bv\deriv[0]\bv'$ if there is some $a\in N_0$ and a
production $a\to\bu$ with $\bu\in\Z^{N_0\cup T}$ or $\bu\in N_1$ such that
$\bv(a)>0$ and $\bv'=\bv-a+\bu$. 

In order to define $\deriv[i]$ inductively for $i>0$, we first need to define
generated sets. Let $\derivs[i]$ denote the reflexive transitive closure of
$\deriv[i]$. 
We define the generated set $L(a)$ for each $a\in N_i$:
\[ L(a)=\{\bu\in\N^{N_{[i+1,k]}}+\Z^{T_{[i+1,k]}} \mid \text{$a\derivs[i] \bu$}\}. \]
We define $\deriv[i]$ based on $\deriv[i-1]$: We have $\bv\deriv[i]\bv'$ if
there is an $a\in N_{i}$ with $\bv(a)>0$ and a production $a\to a'$ for some
$a'\in N_{i-1}$ and a $\bu\in \Z^{N\cup T}$ with $\bu\in L(a')$ and
$\bv'=\bv-a+\bu$. If $a\to b$ is a production with $a\in N_i$ and $b\in
N_{i+1}$ and $\bv(a)>0$ and $\bv'=\bv-a+b$, then we also have
$\bv\deriv[i]\bv'$.  The \emph{emptiness problem} asks, given a $k$-grammar $G$
and a nonterminal $a$ of $G$, is $L(a)\ne\emptyset$?

\subsubsection*{Intuition}
We give some intuition on how $\REACH(\SCpm_{0,\ell})$ is translated into a
grammar. Each nonterminal is of the form $(p,s^x,q)$, where $p,q$ are states of
the valence system and $s^x$ specifies some part of the storage mechanism (see
\cref{appendix-sec-bireach-to-grammars} for details). Then $(p,s^x,q)$
represents the set of all runs of the system from $p$ to $q$ such that on part
$s^x$, the run has overall neutral effect.

Specifically, in the case of $\SCpm_{0,1}$, i.e.\ pushdown $\Z$-VASS,
reachability amounts to checking whether a context-free language contains a word
whose letter counts satisfy some linear condition.  In a $0$-grammar, we have
context-free rewriting rules ($a\to \bu$), which involve nonterminals ($N$) and
terminals ($T$). In addition, the definition of $L(a)$ requires that a derived
contains no terminal from $T_0$ anymore: This is used to implement the linear
condition on letter counts for pushdown $\Z$-VASS.

For $\SCpm_{0,\ell}$, $\ell>1$, we need to simulate stacks with
$\SCpm_{0,\ell-1}$ configurations in each stack entry. Here, we can use a
higher $k$: A derivation step at level $k$ involves an entire derivation at
level $k-1$. This corresponds to the fact that between a pair of push and pop
of a $\SCpm_{d,\ell}$ machine, there is an entire run of an $\SCpm_{d,\ell-1}$
machine.

\xparagraph{Bidirected grammars}
We are now ready to present the symmetry conditions of bidirected grammars.
For a $k$-grammar $G$, an \emph{involution} is a map
$\inv{\cdot}\colon N\to N$ such that for $a\in N_i$, we
have $\inv{a}\in N_i$ for $i\in[0,k]$ and $\inv{(\inv{a})}=a$.
Then, for $\bu\in\Z^{N\cup T}$, we define
$\inv{\bu}\in\Z^{N\cup T}$ as $\inv{\bu}(a)=\bu(\inv{a})$ for
$a\in N$ and $\inv{\bu}(a)=-\bu(a)$ for $a\in T$.
Here, $\inv{\bu}$ can be thought of as the inverse of $\bu$.

We say that a $k$-grammar $G$ is \emph{bidirected} if there is an
involution $\inv{\cdot}\colon N\to N$ such that 
\begin{enumerate}[label=(\arabic*)]
	\item\label{symmetry-production-inverse} for every
production $a\to \bu$ in $P$, we have a production $\inv{a}\to\inv{\bu}$ in $P$,

\item\label{symmetry-production-reverse} for every $a\in R_0$ and every production $a\to
	b+\bu$ with $b\in R_0$, $\bu\in\N^{R_0}+\Z^T$, we have $b\derivs[0] a+\inv{\bu}$,

\item\label{symmetry-reverse-nton} for every production $a\to b$ with $a\in N_i$, $b\in N_j$ with $|i-j|=1$ and $a\in R$, we also have $b\to a$,
\item\label{symmetry-rhs} for every $a\in R$, we have $L(a)\ne\emptyset$,
\item\label{symmetry-add-inverse-pairs}\label{rev-add-inverse-pairs} if $a\in R_0$, then $a\derivs[0] a+a+\inv{a}$.
\end{enumerate}

\subsubsection*{Intuition on bidirectedness}
We give some intuition on these symmetry conditions. 
Recall that in the
translation from automata to grammars, each nonterminal is of the form
$(p,s^x,q)$, where $p,q$ are states and $s^x$ represents part of the storage
mechanism.  The involution is given by $\inv{(p,s^x,q)}=(q,s^x,p)$.  Since the
map $\inv{\cdot}$ negates terminal letters, which represent counter values,
condition~\ref{symmetry-production-inverse} reflects the existence of paths that go in
the opposite direction with opposite effect.

Conditions \ref{symmetry-production-reverse},\ref{symmetry-reverse-nton} let us
reverse productions: Consider a production $(p,s^x,q)\to
(p,s^x,t)+(t,s^x,q)$.  It says one can get a run from $p$ to $q$ by
combining one from $p$ to $t$ with one from $t$ to $q$. Now
\ref{symmetry-production-reverse} yields a derivation $(p,s^x,t)\derivs[0]
(p,s^x,q)+\inv{(t,s^x,q)}=(p,s^x,q)+(q,s^x,t)$. This reflects that paths from
$p$ to $t$ can be obtained from ones from $p$ to $q$ and ones from $q$ to $t$.

\emph{Cross-level productions} $a\to b$ with $a\in N_i$, $b\in N_j$ with $|i-j|=1$
describe the relationship between
comparable parts $s^x$ of the storage. Here, ``comparable'' means ``a subset of the
counters.''  For example, if $r^y$ represents a subset of the counters in $s^x$,
we have productions $(p,s^x,q)\to(p,r^y,q)$, which tell us: A run $\rho$ from
$p$ to $q$ that is neutral on $r^y$ is also neutral on $s^x$ \emph{if} $\rho$'s
effect on the additional counters in $s^x$ is zero (recall that a derivation on
level $i-1$ can only be used as a step on level $i$ if its effect on the
terminal letters is zero).  Now \ref{symmetry-reverse-nton} says that this is
symmetric: The production $(p,r^y,q)\to (p,s^x,q)$ tells us that any run that
is even neutral on $s^x$ is in particular neutral on $r^y$.

The last two conditions \ref{symmetry-rhs} and
\ref{symmetry-add-inverse-pairs} stem from the fact that we sometimes
construct derivations that, as a byproduct, create vectors
$a+\inv{a}=(p,s^x,q)+(q,s^x,p)$. On the one hand, we want to argue
that such cycles can always be eliminated by further derivation. This is
guaranteed by \ref{symmetry-rhs}, which lets us derive some vector $\bu$ from $a$ and
then because of \ref{symmetry-production-reverse}, the inverse $\inv{\bu}$ from
$\inv{a}$. This results in the vector $\bu+\inv{\bu}$, but the nonterminals in
$\bu$ are on a higher level than $a$. Thus, inside of a derivation on level
$k$, we can completely get rid of it. Finally,
\ref{symmetry-add-inverse-pairs} complements this by letting us
create such cycles. This simplifies the set of derived
vectors.

\subsubsection*{Constructing bidirected grammars}
We obtain the following reduction. It is technically involved, but
since it follows similar ideas to existing approaches
(\cite{lohrey2008submonoid,HaaseZetzsche2019}), we defer details to the
appendix. 
\begin{proposition}\label{bireach-to-grammars}
	There is a polynomial-time Turing reduction from $\REVREACH(\SCpm_{d})$
	to the emptiness problem for bidirected $k$-grammars. If the
	input graphs are from $\SCpm_{d,\ell}$, then we only use $k$-grammars
	with $k\le 2\ell$.
\end{proposition}
Note that the reduction described in \cref{bireach-to-grammars} is a Turing reduction.
This is because we need to ensure \ref{symmetry-rhs}. To this end, the reduction
involves a saturation that successively enlarges the set of nonterminals that
are known to generate a non-empty set: At first, it only allows a small set of
triples $(p,s^x,q)\in N$, whose language is non-empty by construction, to
appear on right-hand sides. It then invokes the emptiness check, which yields
more triples (nonterminals) that can then appear on right-hand sides in the
next iteration, etc.

\begin{example} Let us see \cref{bireach-to-grammars}
in an example. Consider the graph $\Gamma$ in \cref{example-pzvass-graph} and the
valence system in \cref{example-pzvass-system}. Observe that $\Gamma$
corresponds to a pushdown $\Z$-VASS: The nodes $a,b$ realize a stack, and $c$
acts as a $\Z$-counter. For simplicity, we show the translation in the final
step of the saturation (described after \cref{bireach-to-grammars}), i.e. where
the set of triples $(p,s^x,q)\in N$ with  $L((p,s^x,q))\ne\emptyset$ has
stabilized. The entry $s^x$ in the nonterminals represents to a subset of the
nodes in the graph: First, $\nort{a}$ and $\nort{b}$ represent only the node
$a$ and $b$, respectively. Second, $\nort{c}$ represents the set $\{a,b\}$.
Third, $\rt{c}$ represents $\{a,b,c\}$. This decomposition into subsets is
derived from the tree structure of $\Gamma$ as a transitive forest,
see~\cref{example-pzvass-decomposition}. The decomposition also determines the
nonterminal levels: We have $N_0=\{(p,s^x,q) \mid
	s^x\in\{\nort{a},\nort{b}\}\}$, $N_1=\{(p,s^x,q) \mid s^x=\nort{c}\}$
	and $N_2=\{(p,s^x,q) \mid s^x=\rt{c}\}$.  The terminal letters
	correspond to looped nodes, and their levels also stem from the
	decomposition into a tree, i.e.\ $T_0=T_1=\emptyset$ and $T_2=\{c\}$.

\newcommand{\exampleHeight}{0.5}

\begin{figure}[t]
\subcaptionbox{The graph $\Gamma$\label{example-pzvass-graph}}{
\begin{tikzpicture}[every circle/.style={}, scale=1]
\path (0,-\exampleHeight) -- (0,\exampleHeight);
\fill (0,0) circle (2pt) node (a) {}    (1,0) circle (2pt) node (c)  {}   (2,0) circle (2pt) node (b) {};
\node at (0,-0.3) {$a$};
\node at (1,-0.3) {$c$};
\node at (2,-0.3) {$b$};
\draw (a.center) -- (c.center) -- (b.center);
\draw (c.center) ++(90:3pt) circle (3pt);
\end{tikzpicture}
}
\subcaptionbox{Valence system over $\Gamma$ (reverse edges are not shown)\label{example-pzvass-system}}{
	\begin{tikzpicture}
		\path (0,-\exampleHeight) -- (0,\exampleHeight);
		\node[state] (q0) {$q_0$};
		\node[state,right=1cm of q0] (q1) {$q_1$};
		\path[->] 
		(q0) edge node[above] {$a$} (q1)
		(q0) edge [loop left] node{$b$} (q0)
		(q1) edge [loop right] node{$c$} (q1)
		;
	\end{tikzpicture}
}
\vspace{0.3cm}

\subcaptionbox{Decomposition of $\Gamma$ into a tree.\label{example-pzvass-decomposition}}{\pzvassTree{1}}
\caption{Graph and example valence system}
\end{figure}
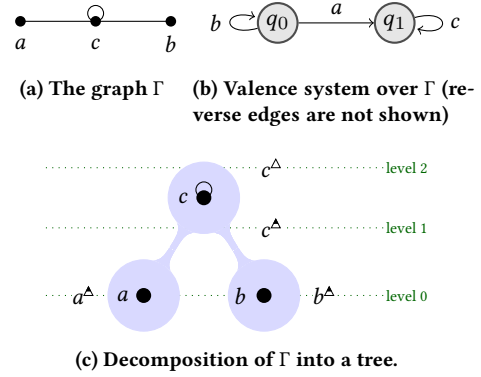

Intuitively, a nonterminal $(p,s^x,q)$ represents runs of the valence system in
which only transitions over nodes (in $\Gamma$) in $s^x$ and its ancestors in
the tree are used and where the effect on the storage is neutral w.r.t. nodes
in $s^x$.  However, the derived multisets of $(p,s^x,q)$ can also contain
nonterminals $(p',r^y,q')$ as ``placeholders'', where $r^y$ is above $s^x$ in
the tree decomposition. This allows derivations at a higher level to insert
neutral runs into each other to simulate runs that are neutral on the union of
two subtrees.

Note that $(q_0,\nort{u},q_1)$ and $(q_1,\nort{u},q_0)$ have empty languages for $u\in\{a,b\}$, because there is no run between $q_0$ and $q_1$ with a neutral effect w.r.t $a$ and $\bar{a}$. 
Therefore, we will not use these nonterminals on the right-hand side. However, all other nonterminals $(p,s^x,q)$ have non-empty languages. Thus, the set $R\subseteq N$ contains all nonterminals $(p,s^x,q)$ with $p=q$.

We begin with productions for nonterminals $(p,\nort{a},p)$:
\begin{align*}
	(p,\nort{u},p) & \to (p,\nort{u},p)+(p,\nort{u},p)  & \text{$p\in\{q_0,q_1\}$, $u\in\{a,b\}$} \\
	(p,\nort{a},p)&\to(q,\nort{a},q) & \text{$p,q\in\{q_0,q_1\}$}
\end{align*}
Note that the second production stems from the fact that we can start from $q_0$,
then move to $q_1$ with $a$, then execute a loop from $q$ to $q$, and then go back to $q_0$ with $\bar{a}$.
We also have productions for the $c$ and $-c$ loops at $q_1$:
\begin{align*}
	(q_1,\nort{u},q_1) & \to c, & (q_1,\nort{u},q_1) & \to -c & \text{for $u\in\{a,b\}$}
\end{align*}
Finally, we have cross-level productions:
\begin{align*}
	(p,\nort{c},p)&\to (p,\nort{u},p) &  (p,\nort{u},p)&\to (p,\nort{c},p) & \text{(level 0/1)} \\
	(p,\rt{c},p)&\to (p,\nort{c},p) &  (p,\nort{c},p)&\to (p,\rt{c},p) & \text{(level 1/2)} 
\end{align*}
for every $p\in\{q_0,q_1\}$ and $u\in\{a,b\}$.
\end{example}

\xparagraph{Reversing derivations on each level} We will show an analogue of
\ref{symmetry-production-reverse} at levels $i>0$,  by induction on $i$. Since
a step on level $i>0$ involves an entire derivation on level $i-1$ (which can
result in several nonterminals), the induction works with a stronger property
where the reverse derivation can not only start in one produced nonterminal (as
in \ref{symmetry-production-reverse}), but arbitrary vectors in $\N^{R_i}$.
But then, suppose we want to reverse 
\[ (p,s^x,q)\deriv[0] (p,s^x,t)+(t,s^x,q). \]
Deriving $(p,s^x,q)$ from $(p,s^x,t)+(t,s^x,q)$ would require either
$(p,s^x,t)$ or $(t,s^x,q)$ to derive a vector without nonterminals, which would
be too strong a condition.  However, we can use a slight relaxation: We will
show that derivations can be reversed \emph{up to a byproduct of cycles
$b+\inv{b}$}. For instance, in our example, the condition~\ref{symmetry-production-reverse} yields 
$(p,s^x,t)\derivs[0](p,s^x,q)+\inv{(t,s^x,q)}$ and hence 
	\[ (p,s^x,t)+(t,s^x,q)\derivs[0] (p,s^x,q)+\inv{(t,s^x,q)}+(t,s^x,q). \]
To
make ``up to cycles $b+\inv{b}$'' precise, we now introduce the
equivalence relations $\approx_a$.

First, we define the relation $\nreach$ on $N$ as follows.  For $a\in N_i$ and
$b\in N_{[i,k]}$, we have $a\nreach b$ if{}f there is a configuration
$\bu\in\N^{N_{[i,k]}} + \Z^{T_{[i,k]}}$ with $\bu(b)\ge 1$ and $a\deriv[i]
\bu$. In other words, for $a\in N_i$ and $b\in N_{[i,k]}$, we have $a\nreach b$
if $b$ can be reached from $a$ in one derivation step on level $i$.

Let $\nreachs$ denote the reflexive transitive closure of $\nreach$.
For each $a\in N_i$, we define the monoid $\Delta_a$, which is generated by all $b+\inv{b}$ with $b\in N_{[i+1,k]}$ and $a\nreachs b$. Hence, $\Delta_a\subseteq\N^{N_{[i+1,k]}}$.
For $\bu,\bu'\in\N^{N_{[i,k]}}+\Z^{T_{[i,k]}}$,
we have $\bu\approx_a\bu'$ if and only if there are $\bd,\bd'\in\Delta_a$ with $\bu+\bd=\bu'+\bd'$. 
A $k$-grammar is called \emph{$i$-bidirected} if for
every $a\in R_i$ and every derivation $a\derivs[i]\bu+\bv$ with
$\bu\in\N^{R_{i}}$, $\bu\ne\bzero$, $\bv\in\N^{R_{[i,k]}}+\Z^{T_{[i,k]}}$, then
$\bu\derivs[i] a+\bv'$ for some $\bv'$ with $\bv'\approx_a\inv{\bv}$. 
In short, up to differences in $\Delta_a$, we can
reverse derivations on level $i$ that start in nonterminals in $R$. 
Note that if $G$ is $i$-bidirected, then on $R_i$, the relation $\nreachs$
is symmetric and thus an equivalence. 
We show the following:
\begin{restatable}{lemma}{levelBidirectedness}\label{level-bidirectedness}
	If $G$ is bidirected, then it is $i$-bidirected for each $i\in[0,k]$.
\end{restatable}

\xparagraph{Expressing emptiness using cosets}\label{grammars-to-cosets}
\newcommand{\cosetL}{\mathscr{L}}
\newcommand{\cosetM}{\mathscr{M}}
\newcommand{\cosetZ}{\mathscr{Z}}
\newcommand{\cosetK}{\mathscr{K}}
\newcommand{\cosetO}{\mathscr{O}}
An important ingredient in our proof is to work with integer linear
\emph{equations} instead of inequalities. While solution sets of systems of
inequalities are semilinear sets in $\Z^m$, solution sets of equation systems
are cosets of $\Z^m$. Recall that a coset in $\Z^m$ is a set of the form
$\bu+U$, where $\bu\in\Z^m$ and $U\subseteq\Z^m$ is a subgroup. In this
section, we will translate emptiness in bidirected grammars into a problem
about cosets of $\Z^{N\cup T}$. In general, we employ three operations for
constructing cosets. First, we can take $\bu+\langle S\rangle$ for
$\bu\in\Z^{N\cup T}$ and any set $S\subseteq\Z^{N\cup T}$.  Recall that
$\langle S\rangle$ is the subgroup generated by $S$.  Second, if
$C_1,C_2\subseteq\Z^{N\cup T}$ are cosets, then their point-wise sum $C_1+C_2$
is a coset, and the set $C_1\cap C_2$ is either a coset or empty. 

A central role will be played by the group $H_a$, which we define for each $a\in R_i$:
\begin{equation} H_a = \langle -b+\bu \mid \text{$b\in R_i$ and $b\deriv[i] \bu$ and $a\nreachs b$} \rangle \label{generating-set-zero} \end{equation}
Hence, $H_a$ is the group generated by all differences that are added when applying derivation steps $b\deriv[i]\bu$ for $a\nreachs b$. We observe that if $a\derivs[i]\bu$, then $-a+\bu\in H_a$. We will need a coset to express that in such a derivation, there are no level-$i$ letters (resp.\ no level-$i$ terminals) left. We will do this by intersecting with the cosets
\[ S_i= %
	\Z^{N_{[i+1,k]}\cup T_{[i+1,k]}},\hspace{0.2cm} S'_i = \Z^{N_{[i+1,k]}}+\Z^{T_{[i,k]}}.
\]
Using $H_a$ and $S_i$, we can now define the coset that will essentially
characterize $L(a)$.  For $a\in R_i$, we set
\begin{equation} L_a = (a+H_a)\cap S_i. \label{coset-equation-la}\end{equation}
One of the main results of this section will be that for $a\in R_i$, 
we can describe $L(a)$ in terms of $L_a$. Following the theme that we can do things
only ``up to differences in $\Delta_a$'' on each level, we need a group version of $\Delta_a$.
For every $a\in N_i$, we have the subgroup $D_a\subseteq\Z^{N_{i+1}}$, which we
define next.  If $a\in N_i$ with $i\in[0,k-1]$, we set $D_a=\langle b+\inv{b}
\mid b\in N_{i+1},~a\nreachs b\rangle$. For $a\in N_k$, we define
$D_a=\{\mathbf{0}\}$. In other words, $D_a$ is the group generated by $\Delta_a$, for every $a\in N$.
With this notation, one of the crucial ingredients in our proof is to show that $L_a=L(a)+D_a$, i.e. $L_a$ describes $L(a)$ up
to differences in $D_a$ (see \cref{grammars-to-groups}). By our observation about $H_a$, it is obvious that
$L(a)+D_a\subseteq L_a$. The key step is that also $L_a\subseteq
L(a)+D_a$.

While the cosets $L_a$ are coset versions of the sets $L(a)$, we will also need
coset versions of a slight variant of $L(a)$. For each $a\in R_i$, by $M(a)$, we denote the vectors that
are derivable from $a$, but may still contain level-$i$ terminals:
\[ M(a) = \{\bu\in\N^{N_{[i+1,k]}}+\Z^{T_{[i,k]}} \mid a\derivs[i] \bu \}. \]
Thus, $M(a)$ differs from $L(a)$ by collecting all derivable $\bu\in\N^{N_{[i+1,k]}}+\Z^T$, where all level-$i$ nonterminals have been eliminated, but not necessarily all level-$i$ terminals, meaning $L(a)=M(a)\cap S_i$. Just as $L_a$ is a coset analogue of $L(a)$, we have a coset analogue of $M(a)$:
\[ M_a = (a+H_a)\cap S'_i. \]
The cosets $M_a$ will be needed to express, using cosets, whether $L(a)$ is empty for nonterminals $a\in N_i$ that do not necessarily belong to $R_i$: Those are the nonterminals for which we do not know whether $L(a)$ is empty. Such sets $L(a)$ do not directly correspond to cosets. However, we will be able to use cosets to characterize when $L(a)$ is empty. Here, we use the cosets $K_a$ for $a\in N_i$, $i\in[0,k-1]$:
\[ K_a = (L_a + \langle -b+M_b\mid b\in R_{i+1},~a\nreachs b\rangle)\cap S_{i+1}. \]
We shall see that the coset $K_a$ corresponds to those vectors in $\N^{N_{[i+1,k]}}+\Z^{T_{[i+1,k]}}$ that can be derived using $a$. 
Using the cosets $K_b$ for $b\in N_{[0,k-1]}$, we will be able to characterize
emptiness of $L(a)$ for $a\in N_{[1,k]}$. In order to do the same for $a\in
N_0$, we need a final type of cosets. For each production $a\to
\bu$ with $\bu=b_1+\cdots+b_n+\bv$ in our grammar with $b_1,\ldots,b_n\in N_0$ and
$\bv\in\Z^T$, we define
\[ K_{a\to \bu} = (M_{b_1}+\cdots+M_{b_n}+\bv)\cap S_0. \]

The key ingredient in our proof is the following analogue of the
translation of \citet{DBLP:conf/cade/VermaSS05} to Presburger arithmetic. Here,
we express the set $L(a)$ of derivable vectors of a nonterminal $a$ as a coset.
\begin{theorem}\label{grammars-to-groups}
	If $G=(N,T,P)$ is bidirected, then $M(a)+D_a=M_a$ for every $a\in R$.
	In particular, $L(a)+D_a=L_a$.
\end{theorem}

\subsubsection*{Intuition for the proof of \cref{grammars-to-groups}}
Before we sketch the proof, let us compare it  with the construction of~\cite{DBLP:conf/cade/VermaSS05}. They show that if we
assign to each production $p$ in a context-free grammar a number $\bx(p)$
saying how often $p$ is applied, then there is a derivation consistent with
$\bx$ if and only if: 
\begin{enumerate}[label=(\roman*)]
	\item each
nonterminal is produced as many times as it is consumed (except for the
start-symbol, which is consumed once more) and 
\item the participating
nonterminals (i.e. for which $\bx(p)>0$) must be \emph{connected}: The latter
means, for each nonterminal, it must be possible to reach its consumed
nonterminal by way of productions that occur in $\bx$. 
\end{enumerate}

Essentially, we will argue that for bidirected grammars, one can drop~(ii).
This is because, independently of a particular derivation, each set $N_i$
decomposes into connected components with respect to $\nreachs$.  Then, instead
of stipulating connectedness of the set of used nonterminals, we only need all
used nonterminals to be in the same $\nreachs$-component. To construct a
derivation  inside one $\nreachs$ component, we need to show that there exist
``connecting derivations'' between each pair in $\nreachs$.  
This is made precise in the notion of Kirchhoff graphs, which we define next.

\subsubsection*{Kirchhoff graphs}
Let $a\in R_i$. A \emph{Kirchhoff graph for $a$} is a directed graph whose set
of vertices is $\{b\in R_i\mid a\nreachs b\}$, such that there exists an edge
$(b,c)$ for each $b,c\in R_i$ with $a\nreachs b$ and $a\nreachs c$, and where
an edge $(b,c)$ is weighted by an element $\bg_{b,c}\in
\N^{N_{[i+1,k]}}+\Z^{T_{[i+1,k]}}$ such that the following holds:
\begin{enumerate}[label=(\roman*)]
	\item $b\derivs[i] c+\bg_{b,c}$ for every $b,c$,
	\item $\bg_{b,b}=\bzero$ for every $b$, and 
	\item for any vertices $b,c,d$,
we have $\bg_{b,c}+\bg_{c,d}\approx_a \bg_{b,d}$.
\end{enumerate}
The term stems from the fact that these graphs satisfy (up to $\approx_a$) a
condition like Kirchhoff's law on voltage drops: The weight sum of every cycle is zero.
In our case, adding up the $\bg_{b,c}$ along a cycle will yield zero up to $\approx_a$.
This implies the following property:
\begin{restatable}{lemma}{kirchhoffProperty}\label{kirchhoff-property}
	If $e_1,\ldots,e_\ell,f_1,\ldots,f_\ell\in R_i$ satisfy $a\nreachs e_j$
	and $a\nreachs f_j$ for $j\in\{1,\ldots,\ell\}$ and $\pi$ is a
	permutation of $\{1,\ldots,\ell\}$, then in a Kirchhoff graph for $a$,
	we have
	\[ \bg_{e_1,f_1}+\cdots+\bg_{e_\ell,f_\ell}\approx_a \bg_{e_1,f_{\pi(1)}}+\cdots +\bg_{e_\ell,f_{\pi(\ell)}}. \]
\end{restatable}
\begin{proof}
	Let us first observe that for any $e\in\{e_1,\ldots,e_\ell\}$ and $f\in\{f_1,\ldots,f_\ell\}$, the relation
\begin{align*}
	\bg_{e,f} + \bg_{e',f'} & \approx_a (\bg_{e,e'}+\bg_{e',f})+(\bg_{e',e}+\bg_{e,f'}) \\
	&\approx_a (\bg_{e,e'}+\bg_{e',e}) + (\bg_{e',f}+\bg_{e,f'}) \\
	&\approx_a \bg_{e',f}+\bg_{e,f'}
\end{align*}
	follows from the definition of Kirchhoff graphs. This is the case of 
	\begin{equation} \bg_{e_1,f_1}+\cdots+\bg_{e_\ell,f_\ell}\approx_a \bg_{e_1,f_{\pi(1)}}+\cdots +\bg_{e_\ell,f_{\pi(\ell)}} \label{kirchhoff-perm}\end{equation}
	where $\pi$ is a transposition, i.e.\ a permutation that swaps two points and lets all others fixed. Since every permutation of $\{1,\ldots,\ell\}$ can be written as a composition of several transpositions, \cref{kirchhoff-perm} follows in full generality.
\end{proof}

\vspace{-0.1cm}
\begin{restatable}{lemma}{kirchhoffGraph}\label{kirchhoff-graph}
	If $G$ is $i$-bidirected, then for each $a\in R_i$, there exists a Kirchhoff graph for $a$.
\end{restatable}
\begin{proof}
	Write $\{b\in R_i\mid a\nreachs b\}=\{b_1,\ldots,b_n\}$. To simplify
	notation, we write $\bg_{r,s}$ instead of $\bg_{b_r,b_s}$. We
	have to pick $\bg_{j,j}=\bzero$. Since $G$ is $i$-bidirected, we
	know that $\nreachs$ is symmetric. In particular, for any
	$j\in[1,n-1]$, there exists a $\bg_{j,j+1}$ such that $b_j
	\derivs[i] b_{j+1}+\bg_{j,j+1}$.  Moreover, $i$-bidirectedness of $G$
	guarantees that there exists a $\bg_{j+1,j}$ with
	$\bg_{j+1,j}\approx_{b_j}\inv{\bg_{j,j+1}}$ such that
	$b_{j+1}\derivs[i] b_j+\bg_{j+1,j}$. Note that since $a\nreachs b_j$
	and $\nreachs$ is symmetric, we have $\Delta_{b_j}=\Delta_a$ and thus
	$\bg_{j+1,j}\approx_a\inv{\bg_{j,j+1}}$.  Finally, for
	$r,s\in[1,n]$ with $r<s$, we pick
	$\bg_{r,s}=\bg_{r,r+1}+\cdots+\bg_{s-1,s}$ and similarly
	$\bg_{s,r}=\bg_{s,s-1}+\cdots+\bg_{r+1,r}$. 

	Let us now show that this is indeed a Kirchhoff graph for $a$.  We
	clearly have $b_r\derivs[i]b_s+\bg_{r,s}$ for any $r,s\in[1,n]$.  It
	remains to show that 
	\[ \bg_{r,s}+\bg_{s,t}\approx_a\bg_{r,t} \]
	for any $r,s,t\in[1,n]$. It suffices to do this 
	in the case that $|s-t|=1$, because the other cases follow by induction.
	Consider the case $t=s+1$ (the case $s=t+1$ is analogous). We have to
	show that 
	\[ \bg_{r,s}+\bg_{s,s+1}\approx_a\bg_{r,s+1}. \]
	If
	$r\le s$, then both sides are identical by definition. If $r>s$, then
	$\bg_{r,s}$ is defined as
	$\bg_{r,s}=\bg_{r,s+1}+\bg_{s+1,s}$. By the choice of
	$\bg_{s+1,s}$, we know that
	$\bg_{s+1,s}\approx_a\inv{\bg_{s,s+1}}$ and thus
	$\bg_{s+1,s}+\bg_{s,s+1}\approx_a\bzero$. Hence:
	\[ \bg_{r,s}+\bg_{s,s+1}=\bg_{r,s+1}+\bg_{s+1,s}+\bg_{s,s+1}\approx_a\bg_{r,s+1}. \qedhere\]
\end{proof}

For \cref{grammars-to-groups}, we need two additional lemmas. The first follows
from property \ref{symmetry-add-inverse-pairs} and induction on $i$.
\begin{restatable}{lemma}{bidirectedDaIa}\label{bidirected-da-ia}
	For bidirected $G=(N,T,P)$ and $a\in R_i$, we have $D_a\subseteq H_a$.
\end{restatable}

Next observe that $L_a$ contains vectors obtained by
adding, but also \emph{subtracting} effects of derivation steps. We now show that if our
grammar is $i$-bidirected, then each such subtraction can be realized by a
sequence of ordinary derivation steps: The lemma says that every element of
$H_a$ can be written as a positive sum of derivation effects (up to a
difference in $D_a$).
\begin{restatable}{lemma}{bidirectedGroupMonoid} \label{bidirected-group-monoid}
	If $G$ is $i$-bidirected, then %
	\[ H_a=\{-b+\bu \mid \exists b\in R_i\colon b\deriv[i]\bu~\text{and}~a\nreachs b\}^*+D_a\]
	for every $a\in R_i$
\end{restatable}
Here, $F^*$ denotes the submonoid of $\Z^{N\cup T}$ generated by the set $F$.
We are now ready to prove \cref{grammars-to-groups}.
\begin{proof}[Proof of \cref{grammars-to-groups}]
	We begin with the inclusion
	``$\subseteq$''. A simple induction on the length of a derivation shows
	that every element of $M(a)$ belongs to $M_a$. 
	\Cref{bidirected-da-ia} tells us that $D_a\subseteq H_a$, and since
	$D_a\subseteq S_i$, this implies $M_a+D_a\subseteq M_a$, hence
	$M(a)+D_a\subseteq M_a$.
	
	We now prove ``$\supseteq$''. An element of $M_a$ is of the form $a+\bv$ with $\bv\in H_a$ and $a+\bv\in S'_i$. We claim that then $a+\bv$ belongs to $M(a)+D_a$.
	Since $\bv\in H_a$, \cref{bidirected-group-monoid} tells us that
	\[ \bv=\sum_{j=1}^n 	-b_j+\bu_j+\bx_j \]
	with	$b_j\in R_i$, $a\nreachs b_j$, $\bu_j\in\N^{R_i}$, $\bx_j\in\N^{N_{[i+1,k]}}+\Z^{T_{[i,k]}}$ where 
	$b_j\deriv[i] \bu_j+\bx_j$ for $j\in[1,n]$.

	Since $G$ is $i$-bidirected by \cref{level-bidirectedness},
	\cref{kirchhoff-graph} yields a Kirchhoff graph for $a$
	with weights $\bg_{b,c}$ for any $b,c\in R_i$ with $a\nreachs b$ and
	$a\nreachs c$.
	Let us now construct a derivation in $G$. Without loss of generality, we may assume
	that $\bu_1,\ldots,\bu_\ell\ne\mathbf{0}$ and $\bu_{\ell+1}=\cdots=\bu_n=\mathbf{0}$.
	For each $j\in[1,\ell]$, we pick some nonterminal $c_j\in R_i$ such
	that $c_1=a$ and $\bu_{j-1}(c_{j})>0$ for $j\in[2,\ell]$. By our choice
	of the $\bg$'s, we can now derive as follows. We use the derivation
	steps $b_j\deriv[i]\bu_j+\bx_j$. But since it is possible that
	$b_{j+1}\deriv[i]\bu_{j+1}+\bx_{j+1}$ cannot be applied after
	$b_j\deriv[i]\bu_j+\bx_j$, we use derivations $c_j\derivs[i]
	b_j+\bg_{c_j,b_j}$ as connecting derivations. Here, we think of the
	$\bg_{c_j,b_j}$ as ``garbage'' that we produce in order to use the
	connecting derivations. Afterwards, we will
	cancel out these garbage elements. We begin with a connecting derivation
	in order to apply $b_1\deriv[i]\bu_1+\bx_1$:
	\begin{align*} a=c_1&\derivs[i]b_1+\bg_{c_1,b_1} \\
		&\deriv[i] \bu_1 + \bx_1 + \bg_{c_1,b_1} = a+(-c_1+\bu_1) + \bg_{c_1,b_1}. \end{align*}
	Since $c_2$ must occur in $a+(-c_1+\bu_1+\bx_1) + \bg_{c_1,b_1}$, we can apply $c_2\derivs[i] b_2+\bg_{c_2,b_2}$, etc. Doing this $\ell$ times yields
	\[ a\derivs[i] a+\sum_{j=1}^\ell -c_j+\bu_j+\bx_j+\bg_{c_j,b_j}. \]
	Let us denote the sum on the right-hand side by $\by$.
	Since want to derive $a+\sum_{j=1}^\ell (-b_j+\bu_j+\bx_j)$ instead of $a+\by$, we now need to correct two aspects: (i)~Our derivation subtracted $c_1,\ldots,c_\ell$ instead of $b_1,\ldots,b_\ell$, so we need to add $c$'s and subtract $b$'s and (ii)~we need to cancel out the garbage elements $\bg_{c_j,b_j}$.
	When replacing $b$'s by $c$'s, it could be that some $c$'s are equal to $b$'s, so
	for (i), we don't have to change those.
	So we pick a permutation $\pi$ of $\{1,\ldots,\ell\}$ and a number $r\in[1,\ell]$ so that (a)~$c_j=b_{\pi(j)}$ for $j\in[1,r]$ and (b)~$\{b_{\pi(r+1)},\ldots,b_{\pi(\ell)}\}$ and $\{c_{r+1},\ldots,c_{\ell}\}$ are disjoint. Now observe that the nonterminals $\{b_{\pi(r+1)},\ldots,b_{\pi(\ell)}\}$ are never consumed in the derivation arriving at $a+\by$. However, since $a+\bv\in S_i$, we know that $b_1+\cdots+b_\ell$ must occur in $\sum_{j=1}^\ell \bu_j$. Therefore, in particular $b_{\pi(r+1)}+\cdots+b_{\pi(\ell)}$ must occur in $\by$.
	But this means we can, for each $j\in[r+1,\ell]$, apply the derivation $b_{\pi(j)}\derivs[i] c_j+\bg_{b_{\pi(j)},c_j}$ to $\by$ (in any order). Thus, from $a$, using $\derivs[i]$, we can derive
	\begin{align*} &a+\sum_{j=1}^\ell -c_j+\bu_j+\bx_j + \bg_{c_j,b_j}+\sum_{j=r+1}^\ell -b_{\pi(j)}+c_j+\bg_{b_{\pi(j)},c_j} \\
		&= a+\sum_{j=1}^\ell -b_j + \bu_j+\bx_j + \left(\sum_{j=1}^\ell \bg_{c_j,b_j}+\sum_{j=r+1}^\ell\bg_{b_{\pi(j)},c_j}\right).
	\end{align*}
	Moreover, since $a+\bv\in S_i$, we know that $b_{\ell+1}+\cdots+b_n$ must occur in $\sum_{j=1}^\ell -b_j+\bu_j+\bx_j$, and so we can just apply the steps $b_j\deriv[i]\bu_j+\bx_j$ for $j\in[\ell+1,n]$ in any order to obtain:
	\[ a\derivs[i] a+\sum_{j=1}^n -b_j + \bu_j+\bx_j + \left(\sum_{j=1}^\ell \bg_{c_j,b_j}+\sum_{j=r+1}^\ell\bg_{b_{\pi(j)},c_j}\right). \]
	Now since $a+\bv\in S'_i$, the right-hand side contains no more level-$i$ nonterminals. Hence, the right-hand side belongs to $M(a)$.
	Finally, since $\pi(j)=j$ for $j\in[1,r]$ and $\bg_{c_j,c_j}\approx\mathbf{0}$ for $j\in[1,\ell]$, the term in parentheses belongs to $D_a$ by \cref{kirchhoff-property}.  Hence, $a+\bv\in M(a)+D_a$. 

	Finally, note that $L(a)+D_a=L_a$ follows from $M(a)+D_a=M_a$, because $D_a\subseteq S_i$ and thus $L(a)+D_a=(M(a)\cap S_i)+D_a=(M(a)+D_a)\cap S_i=M_a\cap S_i=L_a$.
\end{proof}

While \cref{grammars-to-groups} describes what nonterminals $a\in R_i$ can derive,
we also need an analogue $a\in N_i\setminus R_i$.
It is a straightforward consequence of previous steps:

\begin{restatable}{corollary}{emptinessNotRHS}\label{emptiness-not-rhs}
	Suppose $G$ is $i$-bidirected.
	(1) For $a\in N_i$ for $i\in[1,k]$, 
	we have $L(a)\ne\emptyset$ if{}f there is some $a'\in N_{i-1}$
	and a production $a\to a'$ such that $K_{a'}\ne\emptyset$.
	(2) For $a\in N_0$, we have $L(a) \ne\emptyset$ if{}f 
	there is some production $a\to\bu$ with
	$\bu\in\N^{N_0}+\Z^T$ such that $K_{a\to\bu}\ne\emptyset$.
\end{restatable}

\xparagraph{Constructing coset circuits}%
We will
express our cosets in compact representations called \emph{coset circuits}.  Let $Y$ be a finite
set.  A \emph{matrix representation} of a coset $S\subseteq \Z^Y$ is a matrix
$\bA\in\Z^{X\times Z}$ and a vector $\bb\in\Z^X$ such that $Y\subseteq Z$ and
$S=\{\pi_Y(\bx) \mid \bA\bx=\bb\}$, where $\pi_Y\colon\Z^Z\to\Z^Y$ is the
projection onto $\Z^Y$. Note that for a coset given as $S=\bv+\langle
\bu_1,\ldots,\bu_n\rangle$, we can directly compute a matrix representation.
A \emph{coset circuit over $\Z^Y$} is a directed
acyclic graph $C$, whose vertices are called \emph{gates} and
\begin{enumerate}
	\item Leaves, i.e.\ gates $g$ with in-degree $0$, are labeled by a matrix representation of a coset $C(g)$.
	\item Gates with in-degree $>0$ are labeled by $+$ or $\cap$.
\end{enumerate}
In a coset circuit, each gate $g$ evaluates to a coset $C(g)$ of $\Z^Y$ (or the
empty set): The leaves evaluate to their labels.  A gate $g$ with
incoming edges from $g_1,\ldots, g_m$, $m\ge 1$, evaluates to either
$C(g_1)+\cdots+C(g_m)$ or $C(g_1)\cap \cdots\cap C(g_m)$ depending on whether
$v$'s label is $+$ or $\cap$.

We will construct a coset circuit that has for each $a\in N$, a gate for
$H_a,L_a,M_a,K_a$ and also for $K_{a\to\bu}$ for productions $a\to\bu\in P$.
By \cref{emptiness-not-rhs}, this lets us check
emptiness of $L(a)$ for each $a\in N$.  For $H_a,L_a,K_{a\to\bu}$ with $a\in
N_0$, we can directly create leaves with matrix representations. For the
others, the definitions do not tell us directly how to compute the cosets using
sums and intersections, e.g.: $H_a$ is defined by a (potentially infinite)
generating set. Thus, we first show that $H_a$ is generated by finitely many
cosets.
\begin{restatable}{lemma}{hInTermsOfL}\label{h-intermsof-l}
	Let $G=(N,T,P)$ be a bidirected $k$-grammar.
	For every $a\in R_{i}$, $i\in[1,k]$, we have
	\[ H_a = \langle -b+L_c \mid b\in R_{i}, c\in R_{i-1}, a\nreachs b, b\to c\in P\rangle. \]
\end{restatable}
\begin{proof}
	By definition, we have $H_a=\langle -b+L(c) \mid b\in R_i, c\in
	R_{i-1}, a\nreachs b, b\to c\in P\rangle$. Since
	\cref{grammars-to-groups} tells us that $L_c=L(c)+D_c$, the inclusion
	``$\subseteq$'' is immediate. For ``$\supseteq$'', because of
	$L(c)+D_c$, we shall prove that $D_c\subseteq H_a$. (Note that this is
	not an immediate consequence of \cref{bidirected-da-ia}, because $a\in
	R_i$ and $c\in R_{i-1}$ are on different levels). 

	For $D_c\subseteq H_a$, it suffices to prove that every
	generator $e+\inv{e}$ with $e\in R_i$, $c\nreachs e$, of $D_c$ belongs
	to $H_a$.  Since $c\nreachs e$, $e$ appears on a right-hand side of a
	production and thus $L(e)\ne\emptyset$. Hence, there is some
	$\bu\in\Z^{N_{[i+1,k]}\cup T}$ with $e\derivs[i]\bu$. This implies that
	$e+\inv{e}\derivs[i] \bu+\inv{\bu}$. Since $\bu+\inv{\bu}\in
	D_a\subseteq H_a$ (\cref{bidirected-da-ia}) and
	$e+\inv{e}-(\bu+\inv{\bu})\in H_a$, this proves that $e+\inv{e}\in
	H_a$.
\end{proof}
This is a straightforward consequence of \cref{grammars-to-groups}.  
We have
now described each coset in terms of other cosets using sum, intersection, but
also \emph{generated subgroup} (such as in \cref{h-intermsof-l}). In order to
describe cosets only using sums and intersections, we use:
\begin{restatable}{lemma}{localGlobalGroup}\label{local-global-group}
Let $g_1,\ldots,g_n\in \Z^m$ and let $U,S\subseteq\Z^m$ be subgroups with $(g_i+U)\cap S\ne\emptyset$ for all $i\in[1,n]$. Then
\begin{equation} \langle (g_1+U)\cap S,\ldots, (g_n+U)\cap S\rangle = (\langle g_1, \ldots, g_n\rangle+U)\cap S. \label{eq-local-global-group}\end{equation}
\end{restatable}
In other words, \cref{local-global-group} says that instead of imposing the
condition of belonging to $S$ \emph{locally} at each summand $g_i+U$, it
suffices to impose it \emph{globally} on the sum $\langle
g_1,\ldots,g_n\rangle+U$.
\begin{proof}[Proof of \cref{local-global-group}]
The inclusion ``$\subseteq$'' is obvious because $U$ and $S$ are subgroups. 
Conversely, suppose $h=x_1g_1+\cdots+x_n g_n+u\in S$ for some $x_1,\ldots,x_n\in\Z$ and $u\in U$. 
Since $(g_i+U)\cap S\ne\emptyset$, we can choose $u_i\in U$ for each $i\in[1,n]$ such that $g_i+u_i\in S$. 
We compute in the quotient $A/S$. Note that since $h\in S$ and $g_i+u_i\in S$, we have
$[u]=-[x_1g_1+\cdots+x_ng_n]=[x_1u_1+\cdots+x_nu_n]$
and thus $u-x_1u_1-\cdots-x_nu_n\in S$. 
Therefore, we have
\begin{align}
\begin{split}
h=\underbrace{x_1(g_1+u_1)}_{\in \langle (g_1+U)\cap S\rangle} &~+~\cdots~+~\underbrace{x_n (g_n+u_n) }_{\in \langle(g_n+U)\cap S\rangle} \\ &~+~ \underbrace{u-x_1 u_1-\cdots-x_nu_n}_{\in U\cap S}.\label{h-decomp}\end{split} \end{align}
This proves that $h$ belongs to the left-hand side of \cref{eq-local-global-group}, since that set is closed under adding $U\cap S$.
\end{proof}

Hence, if we have built gates for $U$ and $S$ and are given vectors
$g_1,\ldots,g_n$, then to create a gate for $\langle
(g_1+U),\ldots,(g_n+U)\rangle\cap S$, we can create one for $(\langle
g_1,\ldots,g_n\rangle+U)\cap S$. Since we can directly compute a matrix representation for $\langle g_1,\ldots,g_n\rangle$, we can
introduce a leaf for this subgroup.

Let us now describe how, using \cref{local-global-group},  each coset can be
expressed using only sums and intersection.  To simplify notation, for any
$a\in N_i$, we write $P_a = \{(b,c) \mid b\in N_i, c\in N_{i-1}, a\nreachs b,
b\to c\in P\}$. Note that
\begin{align*}
	H_a &= \langle -b+\bu \mid a\nreachs b,~b\to \bu \in P\rangle & &\text{for $a\in R_0$}, \\
	H_a &= \langle -b+L_c \mid (b,c)\in P_a\rangle & &\text{for $a\in R_{[1,k]}$}.
\end{align*}
Thus, for $a\in R_0$, we have a finite generating set for $H_a$ given explicitly in the grammar and can thus create a leaf gate for $H_a$ labeled by an explicit matrix representation for $H_a$. However, for $a\in R_{[1,k]}$, we need to eliminate the
the $\langle\cdot\rangle$ operator. To this end, we write
\begin{equation}
	\begin{split}
	H_a &= \langle -b+L_c \mid (b,c)\in P_a\rangle \\
	&=\sum_{(b,c)\in P_a} \langle (-b+c+H_c)\cap S_{i-1}\rangle \\
	&=\sum_{(b,c)\in P_a} (\langle -b+c\rangle + H_c)\cap S_{i-1}
\end{split}  \label{coset-equation-ha}
\end{equation}
where in the first step, we plug in the definition of $L_c$ and in the second step, we apply \cref{local-global-group}. Here, the sum on the right only uses those $(b,c)\in P_a$ for which $(\langle -b+c\rangle+H_c)\cap S_{i-1}\ne\emptyset$. Now, each $\langle -b+c\rangle$ is a group for which we can create a gate with an explicit representation.

We also need to express $K_a$ in terms of sum and intersection. First note that for $b\in R_{i+1}$,
\begin{equation}\begin{split} 
\langle -b+M_b\rangle&=\langle -b+((b+H_b)\cap S'_{i+1})\rangle \\
	&=\langle (-b+S'_{i+1})\cap H_b\rangle \\
	&= H_b\cap \langle -b+S'_{i+1}\rangle 
\end{split}\label{coset-equation-ma}
\end{equation}
provided that $(-b+S'_{i+1})\cap H_b\ne\emptyset$; otherwise, $\langle -b+M_b\rangle$ is the trivial group $\{\bzero\}$.
In the first equality, we plug in the definition of $M_b$. The second is due to the definition of $S'_{i+1}$, and the third applies \cref{local-global-group}, relying on $(-b+S'_{i+1})\cap H_b$ being non-empty. This implies that for $a\in N_i$,
\begin{equation}\label{coset-equation-ka}
\begin{split}
	K_a&=(L_a + \langle -b+M_b \mid b\in R_{i+1}, a\nreachs b\rangle)\cap S_{i+1}  \\
	&=\left(L_a+\sum_{b\in R_{i+1}, a\nreachs b} \langle -b+M_b\rangle\right)\cap S_{i+1} \\
	&=\left(L_a+\sum_{b\in R_{i+1}, a\nreachs b} H_b\cap \langle -b+S'_{i+1}\rangle\right)\cap S_{i+1},
\end{split}
\end{equation}
where the last sum only uses those $H_b\cap \langle -b+S'_{i+1}\rangle$ for which $(-b+S'_{i+1})\cap H_b\ne\emptyset$. The equality follows from the definition of $K_a$ and \cref{coset-equation-ma}. Note that for $-b+S'_{i+1}$, it is again easy to construct a matrix representation. 

Finally, observe that all these coset definitions rely on the relation $\nreach$. On $N_0$, we can compute $\leadsto$ directly. On $N_{[1,k]}$, we have to rely on cosets. To this, note that for $a,b\in N_i$, we have $a\nreach b$ if and only if there is a $c\in R_{i-1}$ with a production $a\to c$ and some $\bu\in L_c$ with $\bu(b)=1$. In other words, if and only if $L_c\cap O_b\ne\emptyset$, where $O_b$ is the coset $\{\bu\in\Z^{N\cup T} \mid \bu(b)=1\}$.

\begin{algorithm}[t]
	\caption{Construction of coset circuit for a bidirected $k$-grammar}\label{construct-coset-circuit}
	\SetKwInOut{Input}{Input}
	\Input{Bidirected $k$-grammar $G=(N,T,P)$}
	Create gates for $S_i$, $S'_i$, $O_a$, and $\langle -a+b\rangle$ for each $i\in[1,k]$, $a,b\in N$ \\
	Compute $\nreach$ on $N_0$ \\
	Create gates for $H_a$, $L_a$, and $M_a$ for each $a\in R_0$ \\
	Create gates for $K_{a\to \bu}$ for productions $a\to\bu$, $a\in N_0$, $\bu\in\N^{N_0}+\Z^T$ \\
	\For{$i=1,\ldots,k$}{
		Create gate for $L_c\cap O_b$ for each $c\in R_{i-1}$, $b\in N_i$ \\
		Compute the relation $\nreach$ on $N_i$ based on non-emptiness of $L_c\cap O_b$ with $c\in N_{i-1}$, $b\in N_i$ \\		
		Compute $\nreachs$ on $N_i$ and then $P_a$ for each $a\in N_i$ \\
		Create gate for $H_a$, then for $L_a$, then for $M_a$ for each $a\in R_i$ according to \cref{coset-equation-ha,coset-equation-la,coset-equation-ka} \\
		Create gate for $K_{a}$ for each $a\in R_{i-1}$
	}
\end{algorithm}

The exact order in which the gates are constructed is given in
\cref{construct-coset-circuit}.  We observe that the constructed coset circuit
has depth $\le ck$ for some constant $c$: The gates for $H_a$, $L_a$, $M_a$
only depend on gates created in the last iteration of the for-loop. Moreover,
each of them only adds a constant depth to the gates produced before.  The
gates $K_a$ depend on gates $H_b$ created in the same iteration, thus also
adding only constant depth.  Finally, note that the entries in the matrices in
the labels of the leaves require at most polynomially many bits: The gates for
$S_i$, $S'_i$, $O_a$, $\langle -a+b\rangle$ and for $H_a$, $L_a$, and $M_a$ for
$a\in R_0$ are obtained directly from the productions of $G$.  The numbers in
those, in turn, have at most polynomially many bits as they are computed using
\cref{compute-coset-cliques}.

\xparagraph{From coset circuits to equations}%
In each of our three algorithms for $\REVREACH$, we check emptiness of coset
circuits by translating them into integer linear equation systems.  Let
$n=|Y|$.  We compute, for each gate $g$, a matrix $\bA_g\in\Z^{s_g\times t_g}$
and a vector $\bb_g\in\Z^{s_g}$ such that $C(g)=\{\pi_n(\bx) \mid
	\bx\in\Z^{t_g},~\bA_g\bx=\bb_g\}$, where for any $r\in\N$, by $\pi_n$,
	we denote the projection $\pi_n\colon\Z^r\to\Z^n$ onto the last $n$
	coordinates for any $r\ge n$.  If $g$ has in-degree $0$, then $g$ is
	already labeled with such a matrix $\bA$ and vector $\bb$. Now suppose
	$g$ has incoming gates $g_1,\ldots,g_r$.  Let $\bA_1,\ldots,\bA_r$ and
	$\bb_1,\ldots,\bb_r$ with $\bA_i\in\Z^{s_i\times t_i}$,
	$\bb\in\Z^{s_i}$, denote the matrices and vectors constructed for the
	gates $g_1,\ldots,g_r$. Then the matrix $\bA\in\Z^{s\times t}$ and
	$\bb\in\Z^s$ for $g$ have the shape
\begin{align*}
	&\bA=\left (\begin{array}{c|c}
		\begin{array}{ccc}
		\bA_1 &         &          \\
		      &  \ddots &          \\
		      &         & \bA_r
	      \end{array} & \bzero \\\hline
	\bB	& \bC
      \end{array} \right), && \bb=\left(\begin{array}{c}  \bb_1 \\ \vdots \\ \bb_r \\ \bzero \end{array}\right)
 \end{align*}
 where $\bB$ and $\bC$ are chosen depending on whether $g$ is labeled with $+$
 or $\cap$. 
 If the label is $+$, then $\bB\in\Z^{n\times (t_1+\cdots+t_r)},\bC\in\Z^{n\times n}$, are chosen so that $\bB\bx=\bc$ expresses that in $\bx$, the last $n$ coordinates are the sum of $\by_1+\cdots+\by_r$, where for each $i\in[1,r]$, the vector $\by_i\in\Z^n$ contains the coordinates 
 $t_1+\cdots+t_i-n,\ldots,t_1+\cdots+t_i$, i.e. the last $n$
 coordinates corresponding to $\bA_i$. Hence, we have $s=s_1+\cdots+s_r+n$ and $t=t_1+\cdots+t_r+n$.
 If the label is $\cap$, then $\bB\in\Z^{rn\times (t_1+\cdots+t_r)},\bC\in\Z^{rn\times n}$ are chosen so that $\bB\bx=\bc$ expresses that for each
 $i\in[1,r]$, the last $n$ coordinates of $\bx$ agree with
 coordinates $t_1+\cdots+t_i-n,\ldots,t_1+\cdots+t_i$ of $\bx$, i.e. the last $n$
 coordinates corresponding to $\bA_i$. Thus, we obtain $s=s_1+\cdots+s_r+rn$ and $t=t_1+\cdots+t_r+n$.

 Thus, in any case, we have $\max\{s,t\}\le (r+1)\cdot\max_i \max\{s_i,t_i\}$.
 If $\bA$ is a matrix, then we define its \emph{norm}, denoted  $\|\bA\|$, as the maximal
 absolute value of any entry.
 Observe that the norm in the above constructions does not grow at all: We have
 $\|\bA\|\le\max\{\|\bA_i\|\mid i\in[1,r]\}$ and $\|\bb\|\le\max\{\|\bb_i\|
 \mid i\in[1,r]\}$.  Suppose $m$ is an upper bound on the number of rows and
 columns of the matrices in the leafs of $C$, and $M$ is an upper bound on
 their norm. Then for any gate $g$ of depth $i$, the resulting matrix
 $\bA_g\in\Z^{s_g\times t_g}$ and vector $\bb_g\in\Z^{s_g}$ satisfy
 $\max\{s_g,t_g\}\le (r+1)^i\cdot m$ and $\|\bA_g\|,\|\bb_g\|\le M$.
 Hence, the dimensions of the matrices grow
 exponentially in the depth of the circuit, but polynomially for fixed depth.
 Moreover, their entries require no more bits than the matrices in the leaves.

 \begin{restatable}{proposition}{upperBoundsPExptimeExpspace}\label{upper-bounds-p-exptime-expspace}
	(1)~For every fixed number $d\in\N$, the problem $\REVREACH(\SCpm_{d})$ is in $\EXPTIME$.
	(2)~If $\ell\in\N$ is fixed as well, then $\REVREACH(\SCpm_{d,\ell})$ is in $\compP$.
	(3)~$\REVREACH(\SCpm)$ is in $\EXPSPACE$.
 \end{restatable}
 \begin{proof}[Proof sketch]
	 We begin with~(1) and~(2). According to \cref{bireach-to-grammars}, we need 
	 to show that emptiness for bidirected $k$-grammars is in
	 $\EXPTIME$, and in $\compP$ for fixed $k$.

	 Given a bidirected $k$-grammar, we construct a coset circuit $C$ as
	 described in \cref{construct-coset-circuit} by alternating emptiness
	 checks for gates and building new gates.  To check gates for
	 emptiness, recall that the circuit has depth $\le ck$ for some
	 constant $c\in\N$. Thus, it remains to decide emptiness of a gate $g$
	 in exponential time, resp.\ in polynomial time for bounded depth
	 circuits.  For each gate $g$ in this coset circuit, we compute a
	 matrix representation for the coset $C(g)$. As argued above, the
	 resulting matrix $\bA\in\Z^{s\times t}$ and vector $\bb\in\Z^s$ will
	 satisfy $s,t\le(r+1)^{ck}$, where $r$ is the largest in-degree of a
	 gate in $C$. The norm of $\bA$ and $\bb$ is at most the norm
	 of the matrices in the leaves of $C$, which means the entries of $\bA$
	 and $\bb$ only require polynomially many bits.  Therefore, the size of
	 $\bA$ and $\bb$ is polynomial in $(r+1)^{ck}$.  Hence, by
	 \cref{ild-in-p}, we can decide emptiness of $C(g)$ in time polynomial
	 in $(r+1)^{ck}$.
	 Hence, \cref{ild-in-p} yields the desired algorithms.

	 For (3), we proceed slightly differently. Instead of
	 \cref{bireach-to-grammars}, we use a variant with an exponential space
	 reduction.  The resulting grammars still use polynomially many
	 productions, but the occurring numbers can be doubly exponential.  Our
	 matrix encoding thus yields matrices with exponential dimension and
	 their entries have exponentially many bits.  This still allows us to
	 check for emptiness of coset circuits in exponential time, resulting
	 in an $\EXPSPACE$ algorithm overall.
 \end{proof}

\label{beforebibliography}\addtostream{pages}{\getpagerefnumber{beforebibliography}}
\closeoutputstream{pages}
\printbibliography

\pagebreak
\appendix

\section{Additional material for Section~\ref{sec:subgp}}

\subsection{$\REVREACH$ and subgroup membership}

\computecoset*

\begin{proof}
Consider a bidirected valence system $\cA$ over $\Gamma$ with state set $Q$ and two states $s,t \in Q$.
Let $(Q,E)$ be the undirected graph where $E = \{ \{p,q\} \mid p \autsteps[x] q \}$.
First we compute a spanning forest $F \subseteq E$ in logspace:
Let $e_1, \dots, e_m$ be an enumeration of the edges in $E$.
Then we include an edge $e_i = (p,x,q)$ in $F$ if and only if there exists no path between $p$ and $q$
in the subgraph $(Q,\{e_1, \dots, e_{i-1}\})$.
This can be tested in deterministic logspace \cite{DBLP:journals/jacm/Reingold08}.
If $s$ and $t$ are disconnected in $F$, we can return any no-instance of $\SUBMEM(\Gamma)$.
Otherwise we restrict $\cA$ and $F$ to the connected component of $s$.
Then we compute a set of transitions $T$ which contains
for every edge $\{p,q\} \in F$ a transition $p \autsteps[x] q$ and its reverse transition $q \autsteps[\bar x] p$.
Observe that the undirected version of $T$ is a tree.
Between any two states $p,q$ there is a unique simple path using only transitions from $T$.
Any other path from $p$ to $q$ which uses only transitions from $T$ has the same effect
since it is obtained from the simple path from $p$ to $q$ by repeatedly inserting loops $r \autsteps[x] r' \autsteps[\bar x] r$.
Let $w_0$ be the label of the unique simple path from $s$ to $t$ using transitions from $T$.

Let $C = \{ [w] \in \dM \Gamma \mid s \autsteps[w] t \}$ and $S = \{ [w] \in \dM \Gamma \mid t \autsteps[w] t \}$.
Clearly $S$ is a subgroup of $\dM \Gamma$ and $C$ is a left coset of $S$ in $\dM \Gamma$,
namely $C = [w_0]S$. It remains to compute a generating set for $S$.
For any transition $\tau = (p,x,q)$ of $\cA$ with $\tau \notin T$ consider the unique cycle $\gamma_\tau \colon t \autsteps[u] p \autsteps[x] q \autsteps[v] t$
where the paths $t \autsteps[u] p$ and $q \autsteps[v] t$ are simple and only use transitions from $T$.
We claim that the effects of the cycles $\gamma_\tau$ for all $\tau \in E \setminus T$ generates $S$.
Consider any cycle $\gamma \colon q_0 \autsteps[x_1] q_1 \autsteps[x_2] \dots \autsteps[x_n] q_n$ in $E$ where $q_0 = q_n = t$.
We show that $x_1 \dots x_n$ can be written as the product of effects of cycles $\gamma_\tau$ and their inverses
by induction on the number of transitions in $\gamma$ not contained in $T$.
If the cycle only uses transitions from $T$ then we must have $x_1 \dots x_n = 1$.
Otherwise, suppose that $i$ is minimal with $\tau = (q_{i-1},x_i,q_i) \notin T$.
Let $\gamma_\tau \colon t \autsteps[u] q_{i-1} \autsteps[x_i] q_i \autsteps[v] t$, which satisfies $[u] = [x_1 \dots x_{i-1}]$.
We can replace $\gamma$ by the cycle
\[
	t \autsteps[u] q_{i-1} \autsteps[x_i] q_i \autsteps[v] t \autsteps[\bar v] q_i \autsteps[x_{i+1}] \dots \autsteps[x_n] q_n
\]
without changing its effect.
Since $t \autsteps[\bar v] q_i \autsteps[x_{i+1}] \dots \autsteps[x_n] q_n$ uses strictly fewer transitions not contained in $T$ than $\gamma$
we can write $[\bar v x_{i+1} \dots x_n]$ as a product of effects of cycles $\gamma_\tau$ and their inverses.
This concludes the proof.
\end{proof}

\revreachsubmem*

\begin{proof}
Let $\Gamma$ be a looped graph.
Reducing $\SUBMEM(\Gamma)$ to the problem $\REVREACH(\Gamma)$ is easy:
To test $[w] \in \langle [w_1], \dots, [w_k] \rangle$
we construct a bidirected valence system $\cA$ with two states $s,t$
and the transitions $s \autsteps[\bar w] t$
and $t \autsteps[w_i] t$ for all $1 \le i \le k$, and the reverse transitions.
Then $[w] \in \langle [w_1], \dots, [w_k] \rangle$ holds if and only if there exists $u \in X_\Gamma^*$
with $s \autsteps[u] t$ and $[u] = 1$.

Conversely, we can compute in logspace a coset representation $\{ [w] \in \dM \Gamma \mid s \autsteps[w] t \}= [w_0] \langle [w_1], \dots, [w_n] \rangle$
by \Cref{compute-coset}.
Then it remains to test whether $[\bar w_0] \in \langle [w_1], \dots, [w_n] \rangle$.
\end{proof}

\subsection{Undecidability}

In the following we will give the complete proof of \Cref{undecidable-underlying-c4}.
The observation behind our proof is that using two non-adjacent vertices we can simulate a free group of rank two,
i.e. $\dM \Gamma_2$ where $\Gamma_2$ consists of two non-adjacent looped vertices $a$ and $b$.
The $\compP$-hardness in \cref{main-result} also follows from this observation.
Suppose that $u,v$ are non-adjacent vertices in $\Gamma$.
If both $u$ and $v$ are looped, then they already generate a free group of rank two.
Now suppose that $u$ is unlooped.
We first simulate a pushdown over four letters using the words $W = \{ uv^i \mid i \in [1,4] \}$,
i.e. for all $x \in (W \cup \bar W)^*$ we have $x \equiv_\Gamma \varepsilon$ if and only if $x \leftrightarrow_R^* \varepsilon$
where $R = \{(w_i \bar w_i, \varepsilon) \mid i \in [1,k] \}$, see \Cref{sim-pushdown}.
Then we can transform a bidirected valence system $\cA$ over $\Gamma_2$
into a bidirected valence system $\cB$ over $\Gamma$.

If $\Gamma^-$ contains $\Cfour$ as an induced subgraph then we have two pairs $u_1,v_1$ and $u_2,v_2$ of non-adjacent vertices
where vertices from different pairs commute.
Applying the above reduction to every pair, we can reduce $\REVREACH(\Cfourlooped)$ to $\REVREACH(\Gamma)$,
and obtain undecidability of $\REVREACH(\Gamma)$.

\begin{lemma}
\label{sim-pushdown}
Let $\Gamma$ be a graph with two non-adjacent vertices $u$ and $v$ where $u$ is unlooped.
For any $k \in \N$ there exists $W = \{w_1, \dots, w_k\} \subseteq \{u,v\}^*$
such that for all $x \in (W \cup \bar W)^*$
we have $x \equiv_\Gamma \varepsilon$ if and only if $x \leftrightarrow_R^* \varepsilon$
where $R = \{(w_i \bar w_i, \varepsilon) \mid i \in [1,k] \}$.
\end{lemma}

\begin{proof}

First we claim that, if $x \equiv_\Gamma \varepsilon$ where $x \in X_\Gamma^*$ then $x$ does not contain any factor
of the form $u v^i \bar u$ or $u \bar v^i \bar u$ with $i \ge 1$.
This is true because such a factor in $x$ cannot be eliminated by deletions of $u \bar u$, $v \bar v$ or $\bar v v$.

Define $w_i = u v^i$ and let $W = \{ w_i \mid i \in [1,k] \}$.
We clearly have $\equiv_R \, \subseteq \, \equiv_\Gamma$.
For the converse we prove the following stronger statement:
For all $x \in (W \cup \bar W)^*$ with $x \equiv_\Gamma \varepsilon$ we have
$x \equiv_R \varepsilon$ and any occurrence of $w_i \bar w_j$ in $x$ satisfies $i = j$,
We proceed by induction on $|x|$ where the case $|x| = 0$ is clear.
If $|x| > 0$ then $x$ must contain either $u \bar u$, $v \bar v$, or, if $v$ is looped, $\bar v v$.
Since any occurrence of $u$ in a word from $W$ is followed by $v$,
and any occurrence of $\bar v$ in a word from $W$ is followed by $\bar v$ or $\bar u$,
it must be the case that $v \bar v$ occurs in $x$.
Again by the definition of $W$, the word $x$ is of the form $x = s uv^m \bar v^n \bar u t$
where $m,n \in [1,k]$ and $s,t \in (W \cup \bar W)^*$.
We have $x \equiv_\Gamma s uv^{m-1} \bar v^{n-1} \bar u t$.
By the remark from the beginning we must have $m=n=1$ or $m,n \ge 2$.
If $m = n = 1$ then $x \equiv_\Gamma st$.
If $m,n \ge 2$ then $x \equiv_\Gamma s w_{m-1} \bar w_{n-1} t \in (W \cup \bar W)^*$
and, by induction hypothesis, we obtain $m-1 = n-1$, which also implies $x \equiv_\Gamma st$.
In both cases we obtain by the induction hypothesis that $x \equiv_R st \equiv_R \varepsilon$ and that any occurrence of $w_i \bar w_j$ in $st$ satisfies $i=j$.
Hence also occurrence of $w_i \bar w_j$ in $x = s w_n \bar w_n t$ satisfies $i=j$.
This concludes the proof.
\end{proof}

\begin{lemma}
\label{sim-f2}
Let $\Gamma$ be a graph with two non-adjacent vertices $u$ and $v$.
There exists a finite set $W \subseteq \{u,v\}^*$ and a morphism
$\varphi \colon (W \cup \bar W)^* \to X_\Gamma^*$ such that
for all $y \in X_\Gamma^*$ we have
$y \equiv_{\Gamma^\circ} \varepsilon$ if and only if there exists
$x \in (W \cup \bar W)^*$ such that $x \equiv_\Gamma \varepsilon$
and $\varphi(x) = y$.
Furthermore, $\varphi(w) \in X_\Gamma$ for all $w \in W$ and $\varphi(\bar w) = \overline{\varphi(w)}$ for all $w \in W$
(it {\em respects inverses}).
\end{lemma}

\begin{proof}

If $\Gamma = \Gamma^\circ$, i.e. both $u$ and $v$ are looped, then we set $W = V$ and $\varphi$ to be the identity.
Now assume that $u$ is unlooped and let $W = \{w_1, w_2, w_3, w_4\}$ be a set from \Cref{sim-pushdown},
satisfying $x \equiv_\Gamma \varepsilon$ if and only if $x \leftrightarrow_R^* \varepsilon$
where $R = \{(w_i \bar w_i, \varepsilon) \mid i \in [1,k] \}$.
Let $\varphi \colon (W \cup \bar W)^* \to X_\Gamma^*$ be the morphism defined by
\begin{align*}
	w_1 \mapsto u \qquad w_2 \mapsto v \qquad w_3 \mapsto \bar u \qquad w_4 \mapsto \bar v \\
	\bar w_1 \mapsto \bar u \qquad \bar w_2 \mapsto \bar v \qquad \bar w_3 \mapsto u \qquad \bar w_4 \mapsto v
\end{align*}
We claim that for all $y \in X_\Gamma^*$ we have
$y \equiv_{\Gamma^\circ} \varepsilon$ if and only if there exists
$x \in (W \cup \bar W)^*$ such that $x \leftrightarrow_R^* \varepsilon$
and $\varphi(x) = y$.
For the ``if''-direction observe that any derivation $x \leftrightarrow_R x_1 \leftrightarrow_R \dots \leftrightarrow_R x_n = \varepsilon$
can be translated into a derivation
$\varphi(x) \leftrightarrow_{\Gamma^{\circ}} \varphi(x_1) \leftrightarrow_{\Gamma^{\circ}} \dots \leftrightarrow_{\Gamma^{\circ}} \varphi(x_n) = \varphi(\varepsilon) = \varepsilon$
since $\varphi(w_i) \varphi(\bar w_i) \equiv_\Gamma \varepsilon$.

For the ``only if''-direction consider a derivation $y \to_{\Gamma^\circ} y_1 \to_{\Gamma^\circ} \dots \to_{\Gamma^\circ} y_n = \varepsilon$.
We prove by induction over $n$ that there exists $x \in (W \cup \bar W)^*$ with $x \to_R^* \varepsilon$ and $\varphi(x) = y$.
By induction hypothesis there exists $x_1 \in (W \cup \bar W)^*$ with $x_1 \to_R^* \varepsilon$ and $\varphi(x_1) = y_1$.
There exist $s,t \in (W \cup \bar W)^*$ such that $x_1 = s t$, $y_1 = \varphi(s) \varphi(t)$
and $y_1 = \varphi(s) \, z \, \varphi(t)$ where $z \in \{u \bar u, v \bar v, \bar u u, \bar v v\}$.
We need to choose $x \in (W \cup \bar W)^*$ such that $x \to_R x_1$ and $\varphi(x) = y$:
\begin{itemize}
\item If $z = u \bar u$ then set $x = s w_1 \bar w_1 t$.
\item If $z = v \bar v$ then set $x = s w_2 \bar w_2 t$.
\item If $z = \bar u u$ then set $x = s w_3 \bar w_3 t$.
\item If $z = \bar v v$ then set $x = s w_4 \bar w_4 t$.
\end{itemize}
This concludes the proof.
\end{proof}

\undecidableunderlyingcfour*

\begin{proof}

Let us assume that $\Gamma^- \cong \Cfour$.
We want to show that $\REVREACH(\Gamma)$ is undecidable
by a reduction from $\REVREACH(\Cfourlooped)$, which is undecidable by \Cref{mikhailova} and \cref{revreach-submem}.

Let the vertex set of $\Gamma^-$ be $V = \{u_1,v_1,u_2,v_2\}$ with edges $\{u_1,u_2\}$, $\{u_1,v_2\}$,
$\{v_1,u_2\}$, $\{v_1,v_2\}$.
Let $\Gamma_1$ and $\Gamma_2$ be the subgraphs of $\Gamma$ induced by $V_1 = \{u_1,v_1\}$ and $V_2 = \{u_2,v_2\}$, respectively.
By \Cref{sim-f2} there exist finite sets $W_1 \subseteq V_1^*$, $W_2 \subseteq V_2^*$ and morphisms
$\varphi_1 \colon (W_1 \cup \bar W_1)^* \to (V_1 \cup \bar V_1)^*$,
$\varphi_2 \colon (W_2 \cup \bar W_2)^* \to (V_2 \cup \bar V_2)^*$ such that
for all $i \in \{1,2\}$ and $y \in (V_i \cup \bar V_i)^*$ we have
$y \equiv_{\Gamma_i^\circ} \varepsilon$ if and only if there exists
$x \in (W_i \cup \bar W_i)^*$ such that $x \equiv_{\Gamma_i} \varepsilon$
and $\varphi_i(x) = y$.
Let $W = W_1 \cup W_2$ and let $\varphi \colon (W \cup \bar W)^* \to (V \cup \bar V)^*$
be the unique morphism that extends $\varphi_1$ and $\varphi_2$.
For all $y \in (V \cup \bar V)^*$ it satisfies $y \equiv_{\Gamma^\circ} \varepsilon$
if and only if there exists $x \in (W \cup \bar W)^*$
with $x \equiv_\Gamma \varepsilon$ and $\varphi(x) = y$.
Furthermore, the morphism satisfies $\varphi(\bar w) = \overline{\varphi(w)}$ for all $w \in W$ by \Cref{sim-f2}.

Given a valence system $\cA$ over $\Cfour^\circ \cong \Gamma^\circ$, and states $s,t$
we want to test whether there exists a run $s \autsteps[y]_\cA t$ with $y \equiv_{\Gamma^\circ} \varepsilon$.
We can assume that every transition in $\cA$ is labeled by a single symbol from $V \cup \bar V$, by splitting transitions.
Construct the valence system $\cB$ over $\Gamma$ obtained by replacing every transition $p \autsteps[v] q$
by transitions $p \autsteps[w] q$ for all $w \in W$ with $\varphi(w) = v$.
Observe that $\cB$ is bidirected because the morphism respects inverses.
Moreover, there exists a run $s \autsteps[x]_\cB t$ with $x \equiv_\Gamma \varepsilon$ if and only if 
there is a run $s \autsteps[y]_\cA t$ with $y \equiv_{\Gamma^\circ} \varepsilon$.
This concludes the proof.
\end{proof}

\section{Additional material for Section~\ref{sec:lower-bounds}}

\label{sec:lower-bounds-app}

\begin{lemma}\label{l-hardness}
$\REVREACH(\Gamma)$ is $\compL$-hard under $\mathsf{AC}^0$ many-one reductions for every $\Gamma$.
\end{lemma}

\begin{proof}
We need to reduce from the reachability on undirected graphs to $\REVREACH(\Gamma)$
since the former problem is $\compL$-complete under $\mathsf{AC}^0$ many-one reductions \cite{DBLP:journals/cc/AlvarezG00,DBLP:journals/jacm/Reingold08}.
To do so, we simply replace every undirected edge between nodes $p$ and $q$
by transitions $p \xrightarrow{\varepsilon} q$ and $q \xrightarrow{\varepsilon} p$.
\end{proof}

\begin{lemma}\label{ild-hardness}
If $\cG$ is closed under subgraphs and LC-unbounded then $\REVREACH(\cG)$ is $\ILD$-hard under logspace many-one reductions
\end{lemma}

\begin{proof}
Given a system of linear Diophantine equations $\bA \bx = \bb$ where $\bA \in \Z^{m \times n}$ has columns $\ba_1, \dots, \ba_n \in \Z^m$.
Since $\cG$ is LC-unbounded and closed under induced subgraphs it contains a looped clique $\Gamma_m$ with $m$ nodes $v_1, \dots, v_m$.
Let $\varphi \colon \Z^m \to X_{\Gamma_m}^*$ be the function mapping $(k_1, \dots, k_m)$
to $v_1^{k_1} \dots v_m^{k_m}$.
Then we construct the valence system $\cA$ over $\Gamma_m$ with two states $p$ and $q$
and the transitions $p \xrightarrow{\varphi(\ba_i)} p$ for all $1 \le i \le n$ and $p \xrightarrow{\overline{\varphi(\bb)}} q$,
and their reverse transitions.
Then $\bA \bx = \bb$ has a solution $\bx \in \Z^n$ if and only if there exists $w \in \Gamma_m^*$ with $[w] = 1$ and $p \autsteps[w] q$.
\end{proof}

\begin{lemma}\label{p-hardness}
If $\Gamma$ has two non-adjacent vertices then $\REVREACH(\Gamma)$ is $\compP$-hard.
\end{lemma}

\begin{proof}
We can assume that $\Gamma$ only consists of two non-adjacent vertices $u$ and $v$.
Observe that $\dM \Gamma^\circ$ is a free group over $u$ and $v$.
Subgroup membership in the free group over two generators is $\compP$-hard by \cite{avenhaus1984nielsen},
and hence also $\REVREACH(\Gamma^\circ)$ by \cref{revreach-submem}.
We can reduce from $\REVREACH(\Gamma^\circ)$ to $\REVREACH(\Gamma)$ using \Cref{sim-f2},
similarly to the proof \cref{undecidable-underlying-c4}:
There exists a finite set $W \subseteq \{u,v\}^*$ and a morphism $\varphi \colon (W \cup \bar W)^* \to X_\Gamma^*$
such that for all $y \in X_\Gamma^*$ we have $y \equiv_{\Gamma^\circ}$ if and only if
there exists $x \in (W \cup \bar W)^*$ with $x \equiv_\Gamma \varepsilon$ and $\varphi(x) = y$.
Hence reachability in a bidirected valence system over $\Gamma^\circ$ can be logspace reduced to 
reachability in a bidirected valence system over $\Gamma$ by replacing each transition $p \autsteps[v] q$, $v \in X_\Gamma$
by transitions $p \autsteps[w] q$ for all $w \in W$ with $\varphi(w) = v$.
\end{proof}

\begin{lemma}\label{expspace-hardness}
If $\cG$ is closed under induced subgraphs and UC-unbounded then
$\REVREACH(\cG)$ is $\EXPSPACE$-hard under logspace many-one reductions.
\end{lemma}

\begin{proof}
We reduce from the word problem over commutative semigroups, known to be $\EXPSPACE$-hard \cite{MAYR1982305}.
Since $\cG$ is UC-unbounded and closed under induced subgraphs, it contains an unlooped clique $\Gamma$ of size $|\Sigma|$.
We can assume that $\Sigma$ is its node set.
Let $\cA$ be the bidirected valence system over $\Gamma$ with three states $q_0, q_1, q_2$,
the transitions $q_0 \autsteps[\bar u] q_1$, $q_1 \autsteps[v] q_2$,
and the transitions $q \autsteps[\bar x y] q$ for all $(x,y) \in R$,
and their reverse transitions.
Then $u \equiv_R v$ holds if and only if $q_0 \autsteps[w] q_2$ for some $w \in X_\Gamma^*$ with $[w] = 1$.
\end{proof}

\section{Additional material for Section~\ref{sec:upper-bound-ild}}

We need a few known results on commutative semigroups.
In the following we denote by $\Psi \colon \Sigma^* \to \N^\Sigma$ the Parikh function
where $\Psi(w)(x)$ is the number of occurrences of $x$ in $w$.

\begin{lemma}
	\label{rev-vass}
	Let $(\Sigma\mid R)$ be a commutative semigroup presentation.
	If $u \equiv_R v$ then there exists a derivation $u = u_0 \to_R u_1 \to_R \dots \to_R u_n = v$
	such that
	\[
		|u_i| \le \max\{|u|,|v|\}+\|R\|^{2^{|\Sigma|}} + \|R\|
	\]
	for all $i \in [0,n]$
	where $\|R\| =  |R| \cdot \max_{(x,y) \in R} |xy|$.
	In particular, we can test whether $u \equiv_R v$ and, if so, compute such a witnessing derivation in deterministic space $2^{O(|\Sigma|)} \log (\|R\|+|u|+|v|)$.
\end{lemma}

\begin{proof}
The bound on the lengths $|u_i|$ follows from Pro\-position and Lemma~3 in \cite{MAYR1982305}.
To compute the path we use Reingold's logspace algorithm for undirected connectivity,
which also computes a path between two given vertices \cite{DBLP:journals/jacm/Reingold08}.
Let $G$ be the graph obtained from restricting $(\Sigma^*,\to_R)$ to words
of length $\le \max\{|u|,|v|\}+\|R\|^{2^{|\Sigma|}} + \|R\|$
and identifying words having the same Parikh image.
Then $u \equiv_R v$ if and only if they are connected by a path in $G$.
\end{proof}

The following statement from {\cite[Lemma~17]{DBLP:conf/fct/KoppenhagenM97}}
describes the congruence classes of a commutative semigroup presentation as {\em hybrid linear sets}.

\begin{lemma}
	\label{hybrid}
	Let $(\Sigma\mid R)$ be a commutative semigroup presentation.
	Then for any $u \in \Sigma^*$
	there exist base vectors $\bb_1, \dots, \bb_m \in \N^\Sigma$
	and period vectors $\bp_1, \dots, \bp_\ell \in \N^\Sigma$ such that
	\[
		\Psi([u]_{\equiv_R}) = \bigcup_{i=1}^m \{ \bb_i + \sum_{j=1}^\ell \lambda_j \bp_j \mid \lambda_1, \dots, \lambda_\ell \in \N \}
	\]
	and $\|\bb_1\|, \dots, \|\bb_\ell\| \le (|u|+\|R\|) \cdot 2^{O(|\Sigma|)}$.
\end{lemma}

\subsection{$\N$-counters}\label{appendix-upper-bound-n-counters}

Fix a clique $\Gamma = (V,I)$ where $U$ and $L$ are the sets of unlooped and looped vertices in $\Gamma$, respectively.
Furthermore, we are given a bidirected valence system $\cA = (Q,\to)$ over $\Gamma$,
and two states $s,t \in Q$.
For $Y \subseteq V$ let $\pi_{Y} \colon X_\Gamma^* \to (Y \cup \bar Y)^*$ be the projection to the alphabet $Y \cup \bar Y$.

To transfer the results from \Cref{rev-vass,hybrid} to the valence system $\cA$,
we translate $\cA$ into an equivalent commutative semigroup presentation.
To this end, we use the standard method of eliminating the states of $\cA$ by introducing two additional counters.
Without loss of generality let the state set of $\cA$ be $Q = \{1, \dots, k\}$.
Furthermore, we can ensure that each transition $p \autsteps[w] q$ is of the form $w = \overline{w^-} w^+$
for some $w^-, w^+ \in U^*$, by splitting transitions and adding intermediate states.
Let $U' = U \cup \{\alpha,\beta\}$ where $\alpha,\beta$ are fresh symbols.
Define the commutative semigroup presentation $(U',R)$
where $R$ contains all pairs
\[
	(\alpha^p \beta^{k-p} \pi_U(w^-), \alpha^q \beta^{k-q} \pi_U(w^+))
\]
for all transitions $p \xrightarrow{\overline{w^-} w^+} q$,
and all pairs $(xy,yx)$ for $x,y \in U'$.
Any path $p \autsteps[w] q$ with $\Psi(w) = \bzero$
can be translated into a derivation in $R$ from $\alpha^p \beta^{k-p}$ to $\alpha^q \beta^{k-q}$,
and vice versa.

\revvasspath*

\begin{proof}
	By \Cref{rev-vass} one can determine in deterministic space $2^{O(|\Sigma|)} \log (\|R\|+|u|+|v|)$
	if there is a derivation in $R$ from $\alpha^p \beta^{k-p}$ to $\alpha^q \beta^{k-q}$ and, if so, compute such a derivation,
	which can then be translated into a path $s \autsteps[w] t$ in $\cA$ with $\Psi(w) = \bzero$.
\end{proof}

\revvassreach*

\begin{proof}
	By \Cref{hybrid}
	there exist vectors $\bb_1', \dots, \bb_m' \in \N^{U'}$
	and vectors $\bp_1', \dots, \bp_\ell' \in \N^{U'}$ such that
	\[
		\Psi([\alpha^s \beta^{k-s}]_{\equiv_R}) = \bigcup_{i=1}^m \{ \bb_i' + \sum_{j=1}^\ell \lambda_j \bp_j' \mid \lambda_1, \dots, \lambda_\ell \in \N \}
	\]
	and $\|\bb_1'\|, \dots, \|\bb_\ell'\| \le \|\cA\| \cdot 2^{O(|U|)}$.
	By setting $J_q = \{ i \in [1,m] \mid \bb_i'(x) = q, \, \bb_i'(y) = k-q \}$ we obtain
	\[
		\Reach(s,q) = \bigcup_{i \in J_q} \{ \bb_i + \sum_{j=1}^\ell \lambda_j \bp_j \mid \lambda_1, \dots, \lambda_\ell \in \N \}
	\]
	where $\bb_i$ and $\bp_j$ are the restrictions of $\bb_i'$ and $\bp_j'$ to $U$, respectively.
	
	We set $B = \{ u \in U \mid \bp_j(u) = 0 \text{ for all } j \in [1,\ell] \}$ and $b = \|\cA\| \cdot 2^{O(|U|)}$,
	which satisfy the properties claimed by the lemma.
	Note that since we only use the existence of $\bb_1',\ldots,\bb_m'$ and $\bp_1',\ldots,\bp_\ell'$ (and the complexity of computing them is not clear), it remains to show how to compute $B$.
	Observe that $u \notin B$ if and only if there exists $\bv \in \Reach(s,s)$
	with $\bv(u) > 0$.
	This can be decided as follows:
	Let $\cA_u$ be obtained from $\cA$ by adding
	a bidirected transition $s \autsteps[\bar u] \bot$ to a new state $\bot$
	with bidirected loops $\bot \autsteps[v] \bot$ for $v \in U$.
	Then there exists $\bv \in \Reach(s,s)$ with $\bv(u) > 0$ if and only if
	there exists $w \in X_\Gamma^*$ with $s \autsteps[w] \bot$ in $\cA_u$ and $\Psi(w) = \bzero$.
	The latter can be decided in deterministic space $2^{O(|U|)} \cdot \log \|\cA\|$ by \Cref{rev-vass-path}.
\end{proof}

\subsection{Adding $\Z$-counters}

In the following we use the functions $\Phi_Y \colon X_\Gamma^* \to \Z^Y$
defined by $\Phi_Y(w) = \Phi(w)|_Y$ for $Y \subseteq V$,
and similarly $\Psi_Y \colon X_\Gamma^* \to \B^Y$.

\computeeff*

\begin{proof}
Let $B \subseteq U$ be the set and the number $b = \|\cA\| \cdot 2^{O(|U|)}$ from \Cref{rev-vass-reach}.
The idea is to maintain the $B$-counters in the state
and the $(U \setminus B)$-counters using $\Z$-counters.
Let $\Gamma'$ be the looped clique with vertex set $L' = (U \setminus B) \cup L$.
Let $\cA'$ be the valence system over $\Gamma'$ with the state set $Q \times [0,b]^B$.
In the following, we will often use the projection $\pi_{L'}\colon (V\cup\bar{V})^*\to (L'\cup \bar{L'})^*$; to ease notation, we abbreviate it as $\pi$.
For every transition $p \autsteps[w] q$ in $\cA$ and vectors $\ba,\bb \in [0,b]^B$
we add a transition 
\[ 
(p,\ba) \xrightarrow{\pi(w)} (q,\bb) \]
to $\cA'$ if $\ba \oplus \Psi_B(w) = \bb$. Clearly $\cA'$ is bidirected.
In the remainder of the proof, we will show that
\begin{align}
\begin{split}
	\label{eq-subgp-eff}
	\Eff(s,s) = \{ \Phi_L(w') \mid (s,\bzero) \xrightarrow{w'}_{\cA'} (s,\bzero),~ \Phi_{U \setminus B}(w') = \bzero \}.
\end{split}
\end{align}
Once this is established, we are essentially done: This is because \Cref{compute-coset} allows us to compute vectors $\bv_1, \dots, \bv_n \in \Z^{L'}$
such that $\{ \Phi_{L'}(w') \mid (s,\bzero) \autsteps[w'] (s,\bzero) \} = \langle \bv_1, \dots, \bv_n \rangle$.
With these, we have
\[ \Eff(s,s)=\{ \bv|_L \mid \bv \in \langle \bv_1, \dots, \bv_n \rangle, \, \bv|_{L' \setminus L} = \bzero \}, \] 
which is the desired representation of $\Eff(s,s)$.

It remains to prove \eqref{eq-subgp-eff}.
If
\[
	s = s_0 \xrightarrow{w_1} s_1 \xrightarrow{w_2} \dots \xrightarrow{w_n} s_n = s
\]
is a run in $\cA$ with $\Psi_U(w_1 \dots w_n) = \bzero$
then
\begin{equation}
	\label{eq:bounded-run}
	(s_0,\bb_0) \xrightarrow{\pi(w_1)} (s_1,\bb_1) \xrightarrow{\pi(w_2)} \dots \xrightarrow{\pi(w_n)} (s_n,\bb_n)
\end{equation}
is a run in $\cA'$
where $\bb_i = \Psi_B(w_1 \dots w_i)$ for all $i \in [0,n]$ and $\bb_n = \bzero$.
Observe that $\Psi_U(w_1 \dots w_i) \in \Reach(s,s_i)$ and hence $\bb_i \in [0,b]^B$ by \Cref{rev-vass-reach}.
Furthermore the run \eqref{eq:bounded-run} satisfies
\begin{align*}
\Phi_L(\pi(w_1) \dots \pi(w_n)) &= \Phi_L(w_1 \dots w_n), \\
\Phi_{U \setminus B}(\pi(w_1) \dots \pi(w_n)) &= \Phi_{U \setminus B}(w_1 \dots w_n) = \bzero.
\end{align*}

Conversely, suppose there is a run
\[
	(s_0,\bb_0) \xrightarrow{\pi(w_1)} (s_1,\bb_1) \xrightarrow{\pi(w_2)} \dots \xrightarrow{\pi(w_n)} (s_n,\bb_n)
\]
in $\cA'$
with $s_0 = s$ and $s_n = s$
for some run $s_0 \autsteps[w_1] \dots \autsteps[w_n] s_n$ satisfying $\bb_i = \bigoplus_{j=1}^i \Psi_B(w_j) = \Psi_B(w_1 \dots w_i)$ and $\bb_n = \bzero$.
Furthermore we have $\Phi_{U \setminus B}(w_1 \dots w_n) = \Phi_{U \setminus B}(\pi(w_1) \dots \pi(w_n)) = \bzero$.
Let $\bu_i = \Psi_{U \setminus B}(w_1 \dots w_i)$ for all $i \in [0,n]$.
Let $c \in \N$ such that $c \oplus \bu_i(u) \in \N$ for all $i \in [0,n]$ and $u \in U \setminus B$.
By assumption there exists a run $s \autsteps[w] s$ in $\cA$ such that $\Psi_U(w) \in \N^U$
and $\Psi_U(w)(u) \ge c$ for all $u \in U \setminus B$.
Then
\[
	s \xrightarrow{w} s \xrightarrow{w_1 \dots w_n} s \xrightarrow{\bar w} s
\]
is a run in $\cA$
with 
\[
	\Phi_L(w w_1 \dots w_n \bar w) = \Phi_L(w_1 \dots w_n) = \Phi_L(\pi(w_1) \dots \pi(w_n)).
\]
To prove $\Psi_U(w w_1 \dots w_n \bar w) = \bzero$ we consider $B$ and $U \setminus B$ separately.
We have
\begin{align*}
	\Psi_B(w w_1 \dots w_n \bar w) &= \Psi_B(w) \oplus \Psi_B(w_1 \dots w_n) \oplus \Psi_B(\bar w) \\
	&= \Psi_B(w) \oplus \Psi_B(\bar w) = \bzero.
\end{align*}
where the second equality follows from $\bb_n = \bzero$ and the third equality follows from $\Psi_B(w) \in \N^B$.
To show $\Psi_{U \setminus B}(w w_1 \dots w_n \bar w) = \bzero$,
we only need to prove that $\Psi_{U \setminus B}(w w_1 \dots w_i) \in \N^{U \setminus B}$ for all $i \in [0,n]$
since we know that $\Phi_{U \setminus B}(w_1 \dots w_n) = \bzero$.
This holds because
\begin{align*}
	\Psi_{U \setminus B}(w w_1 \dots w_i)(u) = \, & \Psi_{U \setminus B}(w)(u) \, \oplus \\
	& \Psi_{U \setminus B}(w_1 \dots w_i)(u),
\end{align*}
$\Psi_{U \setminus B}(w)(u) \ge c$ and $c \oplus \Psi_{U \setminus B}(w_1 \dots w_i)(u) \in \N$ for all $u \in U \setminus B$.
This concludes the proof.
\end{proof}

\ucboundedild*

\begin{proof}
Given a clique $\Gamma \in \cG$ and a bidirected valence system $\cA = (Q,\to)$ over $\Gamma$,
and states $s,t \in Q$.
Using \Cref{rev-vass-path} we can test whether there is a path $s \autsteps[w] t$ with $\Psi(w) = \bzero$,
and, if so, we find $\bu := \Phi(w) \in \Eff(s,t)$.
This works in exponential (logarithmic) space if $\cG$ consists of cliques of unbounded (bounded) size.
By \Cref{compute-eff}
we can compute a representation $\Eff(s,s) = \{ \bv|_L \mid \bv \in \langle \bv_1, \dots, \bv_n \rangle, \, \bv|_{L' \setminus L} = \bzero \}$.
for some $\bv_1, \dots, \bv_n \in \Z^{L'}$ where $L \subseteq L'$ and $|L'| \le |\Gamma|$.
We need to test whether $\bzero \in \Eff(s,t)$, which is equivalent to $- \bu \in \Eff(s,s)$.
This holds if and only if there exists $(x_1, \dots, x_n) \in \Z^n$ such that
$\sum_{i=1}^n x_i \bv_i(u) = -\bu(u)$ for all $u \in L$
and $\sum_{i=1}^n x_i \bv_i(u) = 0$ for all $u \in L' \setminus L$,
which is a system of $|\Gamma|$ equations.
This can be solved in $\ILD$.
If $|\Gamma|$ is bounded then this is in $\compL$ (even $\mathsf{TC}^0$) \cite[Theorem~13]{DBLP:conf/stacs/ElberfeldJT12}.
\end{proof}

For the proof of \cref{compute-coset-cliques}, we need a slightly stronger
version of \cref{ild-in-p}. Specifically, the algorithm in
\cite{DBLP:journals/siamcomp/ChouC82} is not only able to decide whether an
equation $\bA\bx=\bb$, $\bA\in\Z^{m\times n}$, $\bb\in\Z^m$, has an integral solution $\bx\in\Z^n$. It can even, in
polynomial time, compute a set $\{\bu_1,\ldots,\bu_k\}$ of integral vectors
such that $\langle \bu_1,\ldots,\bu_k\rangle$ is exactly the set of $\bx\in\Z^n$ with
$\bA\bx=\bzero$. 
Such a set $\{\bu_1,\ldots,\bu_k\}$ is called a \emph{general solution} of the equation $\bA\bx=\bzero$.
\begin{theorem}[\cite{DBLP:journals/siamcomp/ChouC82}]\label{compute-general-solution}
	Given a matrix $\bA\in\Z^{m\times n}$, one can compute in polynomial
	time a general solution for the equation $\bA\bx=\bzero$.
\end{theorem}
See \cite[Theorems 1 and 13]{DBLP:journals/siamcomp/ChouC82}.

\computecosetcliques*

\begin{proof}
The proof is analogous to the proof of~\Cref{uc-bounded-ild}.
To compute a generating set for $\Eff(s,s)$ we compute the vectors $\bv_1, \dots, \bv_n \in \Z^{L'}$
with $\Eff(s,s) = \{ \bv|_L \mid \bv \in \langle \bv_1, \dots, \bv_n \rangle, \, \bv|_{L' \setminus L} = \bzero \}$,
which is a projection of the solution set of a linear Diophantine system.
Therefore, by \cref{compute-general-solution}, we can compute in polynomial time $\bu_1, \dots, \bu_m \in \Z^L$ with
$\Eff(s,s) = \langle \bu_1, \dots, \bu_m \rangle$.
Using \Cref{rev-vass-path} we can test whether $\Eff(s,t) \neq \emptyset$, and,
if so, find $\bu \in \Eff(s,t)$ in exponential space (log-space if $\cG$ is UC-bounded).
This gives us the representation $\bu + \langle \bu_1, \dots, \bu_n \rangle$ for $\Eff(s,t)$.
\end{proof}

\section{Additional material for Section~\ref{sec:upper-bound-p}}
\subsection{The reduction to grammar emptiness}\label{appendix-sec-bireach-to-grammars}
In this section, we prove \cref{bireach-to-grammars}.

The idea behind all these grammar translations is to simulate the runs of
valence systems over the disjoint union of graphs: Such a run is always obtained
by starting from a run with neutral effect over one component, then inserting a run with neutral effect over a
different component, then again inserting a run from some component, etc.
Here, a key trick is to use a nonterminal symbol $a_{p,q}$ for each pair of
states that represents a run from $p$ to $q$. Then, inserting a run from $p$ to
$q$ for $a_{p,q}$ yields a new run. For general reachability, this works even if
we introduce nonterminals $a_{p,q}$ for which there does not exist a run:
Such an $a_{p,q}$ will never be replaced and causes no issues.

However, if we want to guarantee that our grammar is bidirected (we will define
this later precisely), we always have to make sure that every derivation can be
reverted. In particular, every nonterminal that can be produced should be able
to derive something. This property will be captured in our notion of
``realizable placeholder runs'' which we define next.

\subsubsection*{Decomposition into tree} First, we want to make the tree structure in the input graphs $\Gamma\in\SCpm_d$ explicit.
We may assume that our input graph has an unlooped vertex that is adjacent to all other vertices: Otherwise, we can just add such a vertex.
Our graph $\Gamma=(V,I)$ in $\SCpm_d$ has a tree structure. We decompose its vertices into a tree $t$ accordingly:
\begin{enumerate}
	\item Consider the set $U\subseteq V$ of vertices in $\Gamma$ that are adjacent to all other vertices.
	\item If $U$ is a strict subset of $V$, then $\Gamma\setminus U$ has at least two connected components.
		Then we construct the tree for each connected component. The tree for $\Gamma$ is obtained by taking these trees and adding a parent node containing $U$.
	\item Otherwise, 
		$\Gamma$ is a clique with $\le d$ unlooped vertices. Then $t$ contains one node with all of $\Gamma$.
\end{enumerate}

\subsubsection*{Placeholder runs}
Given a valence system $\cA$ over $\Gamma$, a \emph{placeholder} is a triple
$(p,\rt{s},q)$ or $(p,\nort{s},q)$, where where $p$ and $q$ are states in $\cA$
and $s$ is a subtree of $t$.  Intuitively, a placeholder $(p,\rt{s},q)$
represents a run from $p$ to $q$ that is neutral and only uses operations
belonging to $s$. A placeholder $(p,\nort{s},q)$ also represents such a run,
but without the restriction that it has to cancel in the vertices belonging to
the root of $s$.  Hence $\rt{s}$ can be thought of as representing the whole
subtree $s$, whereas $\nort{s}$ represents the set of subtrees directly under
$s$.  The set of placeholders is denoted by $\rho(\cA)$.
    
We say that $r^x$ is \emph{above} $s^y$ if either (i)~$r$ strictly contains $s$ as a subtree or if (ii)~$r^x=\rt{r}$ and $s^y=\nort{r}$.
For each $s^x\in\{\nort{s},\rt{s}\}$, we define two sets of vertices:
\begin{enumerate}
	\item $V_s\subseteq V$ consists of all vertices belonging to the subtree $s$ or to an ancestor of the root of $s$.
	\item $\vertdown_{\rt{s}}\subseteq V$ consists of all vertices in $\Gamma$
		that belong to the subtree $s$.
	\item $\vertdown_{\nort{s}}\subseteq V$ consists of all vertices in $\Gamma$ that belong to the subtree of $s$, but not to the root of $s$.
	\item $\vertup_{s^x}=V_s\setminus \vertdown_{s^x}$ for $s^x=\nort{s}$ or $s^x=\rt{s}$.
\end{enumerate}
Given a subtree $s$ of $t$, the automaton
$\cA_{s}$ is obtained from $\cA$ by deleting all edges labeled with
$v\in\Gamma$ that are incomparable to $s$. 
For each $s^x\in\{\rt{s},\nort{s}\}$, we will use the projection maps $\projup_{s^x}\colon X_\Gamma^*\to\Z^{\vertup_{s^x}}$ and
$\projdown_{s^x}\colon X_\Gamma^*\to\M\Gamma_{\vertdown_{s^x}}$, which project a
string over $X_\Gamma$ to the symbols belonging to the nodes in $\vertup_{s^x}$
and $\vertdown_{s^x}$, respectively, and return their image in
$\Z^{\vertup_{s^x}}$ or $\M\Gamma_{\vertdown_{s^x}}$, respectively.
Here $\Gamma_U$ is the subgraph of $\Gamma$ induced by a vertex set $U$.

A \emph{run in $(p,s^x,q)$} is a sequence of transitions
\[(q_0,w_1,q_1)\cdots (q_{m-1},w_m,q_m) \] in $\cA_s$ such that $q_0=p$, $q_m=q$, and
$\projdown_{s^x}(w_1\cdots w_m)=1$.  Let $R\subseteq\rho(\cA)$ be a subset. An
\emph{$R$-placeholder run in $(p,s^x,q)$} is a sequence 
\begin{equation} \sigma_0(p_1,s_1^{x_1},q_1)\sigma_1\cdots (p_m,s_m^{x_m},q_m)\sigma_m \label{form-placeholder-run}\end{equation}
where (i)~$p=q_0$ and $q=p_{m+1}$, (ii)~for $i\in[1,m]$, $s_i^{x_i}$ is above $s^x$ and $(p_i,s_i^{x_i},q_i)\in R$, (iii)~$\sigma_i\colon q_i\autsteps[w_i]p_{i+1}$ is a run in $\cA_s$ for each $i\in[0,m]$
and (iv)~$\projdown_{s^x} (w_1\cdots w_m)=1$.
An \emph{almost $R$-placeholder run in $(p,s^x,q)$} is an $R$-placeholder run
where we also allow $s_i^{x_i}=s^x$ and only impose that $\projdown_{\nort{s}} (w_1\cdots w_m) = 1$.

The \emph{effect} of the $R$-placeholder run in \cref{form-placeholder-run} is
\begin{align*}
(p_1,s_1^{x_1},q_1)+\cdots&+(p_m,s_m^{x_m},q_m) \\
&+\projup_{s^x}(w_1\cdots w_m)\in
\N^{\rho(\cA)}+\Z^{V}.
\end{align*}
The set of effects of $R$-placeholder runs in
$(p,s^x,q)$ is denoted $E^R_{(p,s^x,q)}\subseteq\N^{\rho(\cA)}+\Z^V$.
The effect of an almost $R$-placeholder runs is defined as
\begin{align*}
(p_1,s_1^{x_1},q_1)+\cdots&+(p_m,s_m^{x_m},q_m) \\
&+ \projup_{\nort{s}}(w_1\cdots w_m)\in \N^{\rho(\cA)}+\Z^{V}.
\end{align*}
By $E'^R_{(p,s^x,q)}$, we denote the set of effects of almost $R$-placeholder runs in $(p,s^x,q)$.

We now inductively define which placeholders and which placeholder runs are realizable:
\begin{itemize}
\item An $R$-placeholder run in $\tau$ is {\em realizable} if all $\tau' \in R$ are realizable.
\item A placeholder $\tau$ is {\em realizable} if there exists a realizable placeholder run in $\tau$.
\end{itemize}
In particular, all ordinary runs in $\tau$ are realizable (set $R = \emptyset$).
The set of realizable placeholder runs in $\tau \in \rho(\cA)$ is denoted by $U_\tau$.
Every realizable placeholder run is of the form \cref{form-placeholder-run} (even those that are ordinary runs, which just have $m=0$).
Hence, they have a well-defined effect.
The set of effects of realizable placeholder runs in
$(p,s^x,q)$ is denoted $E_{(p,s^x,q)}\subseteq\N^{\rho(\cA)}+\Z^V$.

The notion of a realizable placeholder run achieves the idea that we have mentioned
above: A realizable placeholder run can only use realizable placeholders, for
which we have already established the existence of a run.  Clearly, a
placeholder run in $(p,\rt{t},q)$ is just a neutral run in
$(p,\rt{t},q)$.  Thus:
\begin{lemma}\label{revreach-placeholder-runs}
	There exists a neutral run from $p$ to $q$ in $\cA$ if and only if $E_{(p,\rt{t},q)}\ne\emptyset$.
\end{lemma}

Our goal is to describe the sets $E_\tau$ using grammars. For this, it will be useful
to have a characterization of $E_\tau$ that describes how to ``build up'' elements of $E_\tau$ successively. This is the purpose of the sets $W_\tau$, which we define next.
For this, we need to define the subsets $Y_{r^x},Z_{r^x}$ \\ $\subseteq\N^{\rho(\cA)}+\Z^V$:
\begin{multline*}
	Z_{r^x} = \{\bx\in \N^{\rho(\cA)}+\Z^V \mid \text{$\bx(v)=0$ for every $v\in \vertdown_{r^x}$} \\
	~\text{and $\bx(p,s^y,q)=0$ for every $p,q\in Q$} \\
	~\text{and every $s^y$ below $r^x$ or equal to $r^x$}\} 
\end{multline*}
and $Y_{r^x}$ is the intersection of all $Z_{s^y}$ where $s^y$ is below $r^x$.
To simplify notation, if $\tau=(p,r^x,q)$, we also write $Z_\tau$ for $Z_{r^x}$ (analogously for $Y_\tau$).
Let us inductively define the tuple $(W_\tau)_{\tau\in \rho(\cA)}$, where for $W_{(p,s^x,q)}\subseteq \N^{\rho(\cA)} + \Z^{V}$. It is the smallest tuple such that
\begin{enumerate}
	\item\label{w-1} If $s$ is a leaf and $R\subseteq \{\tau\in\rho(\cA)\mid W_\tau\cap Z_\tau\ne\emptyset\}$
		then $E'^R_{(p,\nort{s},q)}\subseteq W_{(p,\nort{s},q)}$.
	\item\label{w-2} If $r^x$ is above $s^y$ and $W_{(p,s^y,q)}\cap Z_{s^y}\ne\emptyset$, then $W_{(p,s^y,q)}\cap Y_{r^x}\subseteq W_{(p,r^x,q)}$.
	\item\label{w-3} If $r^x$ is above $s^y$ and $W_{(p,r^x,q)}\cap Z_{r^x}\ne\emptyset$, then $(p,r^x,q)\in W_{(p,s^y,q)}$.
	\item\label{w-4} If $(p',r^x,q')+\bu\in W_{(p,r^x,q)}$ and $\bv\in W_{(p',r^x,q')}$, then $\bv+\bu\in W_{(p,r^x,q)}$.
\end{enumerate}

Once the notion of realizable placeholder runs and the sets $W_\tau$ is
established, it only requires standard arguments that $E_\tau$ and $W_\tau\cap
Z_\tau$ agree.
\begin{proposition}\label{equivalence-placeholder-runs}
	For every $\tau\in\rho(\cA)$, we have $W_\tau\cap Z_\tau=E_\tau$.
\end{proposition}

\begin{proof}

We say that an almost $R$-placeholder run in $\tau$ is {\em realizable}
if all $\tau' \in R$ are realizable.
Let $E'_\tau$ be the set of all effects of an almost placeholder run in $\tau$.
Observe that $E'_\tau \cap Z_\tau = E_\tau$.
Hence it remains to show that $W_\tau = E'_\tau$.
Since the sets $W_\tau$ and $E'_\tau$ are defined inductively
it is to natural to prove both inclusions of $W_\tau = E'_\tau$
by inductions on the number of rules needed to witness that an element belongs to the sets.

First we prove $W_\tau \subseteq E'_\tau$ for all $\tau\in\rho(\cA)$.
Let $W^{(0)}_\tau \subseteq W^{(1)}_\tau \subseteq \dots$ be the smallest sets satisfying:
\begin{enumerate}
	\item If $s$ is a leaf and $R\subseteq \{\tau\in\rho(\cA)\mid W_\tau^{(k)} \cap Z_\tau\ne\emptyset\}$
		then $E'^R_{(p,\nort{s},q)}\subseteq W_{(p,\nort{s},q)}^{(k+1)}$. \label{w1}
	\item If $r^x$ is above $s^y$ and $W_{(p,s^y,q)}^{(k)}\cap Z_{s^y}\ne\emptyset$, then $W_{(p,s^y,q)}^{(k)}\cap Y_{r^x}\subseteq W_{(p,r^x,q)}^{(k+1)}$. \label{w2}
	\item If $r^x$ is above $s^y$ and $W_{(p,r^x,q)}^{(k)}\cap Z_{r^x}\ne\emptyset$, then $(p,r^x,q)\in W_{(p,s^y,q)}^{(k+1)}$. \label{w3}
	\item If $(p',r^x,q')+\mu\in W_{(p,r^x,q)}^{(k)}$ and $\nu\in W_{(p',r^x,q')}^{(k)}$, then $\nu+\mu\in W_{(p,r^x,q)}^{(k+1)}$. \label{w4}
\end{enumerate}
We prove $W^{(k)}_\tau \subseteq E'_\tau$ for all $\tau \in \rho(\cA)$ by induction on $k$.

\begin{enumerate}
\item Let $R\subseteq \{\tau\in\rho(\cA)\mid W_\tau^{(k)}\cap Z_\tau\ne\emptyset\}$.
Since $E'^R_{(p,\nort{s},q)} \subseteq W^{(k+1)}_{(p,\nort{s},q)}$
we need to show the inclusion $E'^R_{(p,\nort{s},q)} \subseteq E'_{(p,\nort{s},q)}$.
By induction hypothesis we know that $W_\tau^{(k)}\cap Z_\tau \subseteq E'_\tau \cap Z_\tau = E_\tau$ for all $\tau \in \rho(\cA)$.
Therefore all $\tau \in R$ are realizable, which implies $E'^R_{(p,\nort{s},q)} \subseteq E'_{(p,\nort{s},q)}$.

\item Suppose that $r^x$ is above $s^y$ and $W_{(p,s^y,q)}^{(k)}\cap Z_{s^y}\ne\emptyset$.
Take $\kappa \in W_{(p,s^y,q)}^{(k)}\cap Y_{r^x}$.
The goal is to show that $\kappa \in E'_{(p,r^x,q)}$.
By induction hypothesis $\kappa \in E'_{(p,s^y,q)}$, i.e.
$\kappa$ is the effect of a realizable almost placeholder run
\begin{align*}
	\sigma \colon (q_0 \autsteps[w_1] p_1) &(p_1,s_1^{x_1},q_1) (q_1 \autsteps[w_2] p_2) \cdots \\ &(p_m,s_m^{x_m},q_m) (q_m \autsteps[w_m] p_{m+1})
\end{align*}
in $(p,s^y,q)$.
In particular, all $s_i^{x_i}$ are above or equal to $s^y$.
Indeed, all such $s_i^{x_i}$ are above $r^x$ or equal to $r^x$ since $\kappa \in Y_{r^x}$.
Furthermore we claim that $\projdown_{\nort{r}} (w_1\cdots w_m)=1$:
First all runs $q_i \autsteps[w_i] p_{i+1}$ are in $\cA_s$.
We know $\projdown_{\nort{s}} (w_1\cdots w_m)=1$
and also $\projup_{\nort{s}}(w_1\cdots w_m)(v) = 0$ for all $v \in \vertdown_{\nort{r}}$ because $\kappa \in Y_{r^x}$.
Hence $\sigma$ is a realizable almost placeholder run in $(p,r^x,q)$ and therefore $\kappa \in E'_{(p,r^x,q)}$.

\item Suppose that $r^x$ is above $s^y$ and $W_{(p,r^x,q)}^{(k)}\cap Z_{r^x}\ne\emptyset$.
The goal is to show that $(p,r^x,q)\in E'_{(p,s^y,q)}$.
By induction hypothesis we know that there exists $\kappa \in E'_{(p,r^x,q)} \cap Z_{r^x}$,
which is the effect of a realizable almost placeholder run $\sigma$ in $(p,r^x,q)$.
In fact $\sigma$ is a realizable placeholder run since $\kappa \in Z_{r^x}$.
Therefore $(p,r^x,q)$ is realizable, and thus $(p,r^x,q)$ is a realizable almost placeholder run in $(p,s^y,q)$, with effect $(p,r^x,q)\in E_{(p,s^y,q)}'$.

\item Suppose that $(p',r^x,q')+\mu\in W_{(p,r^x,q)}^{(k)}$ and $\nu\in W_{(p',r^x,q')}^{(k)}$.
The goal is to prove $\nu+\mu\in E'_{(p,r^x,q)}$.
By induction hypothesis $(p',r^x,q')+\mu\in E'_{(p,r^x,q)}$ and $\nu\in E'_{(p',r^x,q')}$,
i.e. $(p',r^x,q')+\mu$ is the effect of a realizable almost placeholder run $\sigma_1$ in $(p,r^x,q)$
and $\nu$ is the effect of a realizable almost placeholder run $\sigma_2$ in $(p',r^x,q')$.
Let $\sigma$ be obtained from $\sigma_1$ by replacing any occurrence of the placeholder $(p',r^x,q')$
by $\sigma_2$.
Observe that $\sigma$ is a realizable almost placeholder run in $(p,r^x,q)$ with effect $\nu + \mu$,
and hence $\nu+\mu\in E'_{(p,r^x,q)}$.
\end{enumerate}

For the converse direction, we show for all $\tau\in\rho(\cA)$ and all realizable almost placeholder runs $\sigma$
that its effect belongs to $W_\tau$,
by lexicographic induction, where we first order by the {\em weight} of $\sigma$ and then by the ``below'' order
on the node descriptions $s^x$ in $\tau$.
For every realizable almost placeholder run $\sigma$ in $\tau$ we inductively define a {\em weight}:
If $\sigma = \sigma_0(p_1,s_1^{x_1},q_1)\sigma_1\cdots (p_m,s_m^{x_m},q_m)\sigma_m$
then its weight is $|\sigma_0 \dots \sigma_m|+\sum_{i=1}^m \omega_i$
where $|\sigma_0 \dots \sigma_m|$ is the total number of edges used in the runs $\sigma_i$,
and $\omega_i$ is the minimal weight of a realizable placeholder run in $(p_i,s_i^{x_i},q_i)$.

For the induction base we consider a realizable almost placeholder run consisting of a single transition $(p,w,q)$ in $(p,\nort{s},q)$
for some leaf $s$.
Its effect is contained in $E'^\emptyset_{(p,\nort{s},q)} \subseteq W_{(p,\nort{s},q)}$ by \eqref{w1}.

For the induction step consider a realizable almost placeholder run
\[
	\sigma = \sigma_0 (p_1,s_1^{x_1},q_1) \sigma_1 \cdots (p_m,s_m^{x_m},q_m) \sigma_m
\]
in $(p,s^x,q)$,
i.e. (i)~$p=q_0$ and $q=p_{m+1}$, (ii)~for $i\in[1,m]$, $s_i^{x_i}$ is above or equal to $s^x$ and $(p_i,s_i^{x_i},q_i)$ is realizable,
(iii)~$\sigma_i\colon q_i\autsteps[w_i]p_{i+1}$ is a run in $\cA_s$ for each $i\in[0,m]$
and (iv)~$\projdown_{\nort{s}} (w_1\cdots w_m) = 1$.
\begin{enumerate}
\item If $s^x = \rt{s}$ then $\sigma$ is a realizable placeholder run in $(p,\nort{s},q)$.
By induction hypothesis the effect $\kappa$ of $\sigma$ is contained in $W_{(p,\nort{s},q)}$.
From $\kappa \in Z_{\nort{s}} = Y_{\rt{s}}$ and \eqref{w2} we obtain $\kappa \in W_{(p,\rt{s},q)}$.

\item Now assume that $s^x = \nort{s}$ and let $t_1, \dots, t_k$ be the immediate subtrees of $s$.
Let $u$ be the projection of $w_1 \dots w_m$ to the alphabet $\vertdown_{\nort{s}} \cup \bar \vertdown_{\nort{s}}$,
which satisfies $u \equiv_\Gamma \varepsilon$.
We can uniquely factorize $u = u_1 \dots u_n$ where $u_i \in (\vertdown_{\rt{t}_{i_j}} \cup \bar \vertdown_{\rt{t}_{i_j}})^+$ for some $i_1, \dots, i_n \in [1,k]$
with $i_j \neq i_{j+1}$ for all $j \in [1,n-1]$.
\begin{enumerate}
\item If $n = 1$ then $u \in (\vertdown_{t_{i_1}} \cup \bar \vertdown_{t_{i_1}})^*$
and hence the runs $\sigma_0, \dots, \sigma_m$ are runs in $\cA_{t_{i_1}}$.
Therefore $\sigma$ is a realizable placeholder run in $(p,\rt{t}_{i_1},q)$.
By induction hypothesis the effect $\kappa$ of $\sigma$ is contained in $W_{(p,\rt{t}_{i_1},q)}$.
From $\kappa \in Z_{\rt{t}_{i_1}} \subseteq Y_{\nort{s}}$ and \eqref{w2} we obtain $\kappa \in W_{(p,\nort{s},q)}$.
\item If $n > 2$ then there exists $\ell \in [1,n]$ such that $u_\ell \equiv_\Gamma \varepsilon$ by the properties of a free product of monoids.
There exists a placeholder run $\sigma'$ contained in $\sigma$ of the form
\[
	\sigma' = \sigma_i' (p_i,s_i^{x_i},q_i) \sigma_{i+1} \cdots \sigma_{j-1} ( p_j,s_j^{x_j},q_j) \sigma_j',
\]
where the prefix $\sigma_i' \colon q_{i-1}' \autsteps[w_i'] p_{i}$ is a suffix of $\sigma_i$,
and the suffix $\sigma_j' \colon q_j \autsteps[w_j'] p_{i+1}'$ is a prefix of $\sigma_j$,
and $u_\ell$ is the projection of $w_i' w_{i+1} \dots w_{j-1} w_j'$ to $\vertdown_{\nort{s}} \cup \bar \vertdown_{\nort{s}}$.
Since $\sigma$ is a realizable almost placeholder run in $(p,\nort{s},q)$,
$\sigma'$ is a realizable almost placeholder run in $(q_{i-1}',\nort{s},p_{i+1}')$.
Observe that $(q_{i-1}',\nort{s},p_{i+1}')$ is realizable
since we can replace in $\sigma'$ all placeholders of the form $(p_h,\nort{s},q_h)$ by realizable placeholder runs in $(p_h,\nort{s},q_h)$
to get a realizable placeholder run in $(q_{i-1}',\nort{s},p_{i+1}')$.
Let $\nu$ be the effect of $\sigma'$ and $\mu$ be the effect of $\sigma$ without $\sigma'$.
We can then replace $\sigma'$ in $\sigma$ by the realizable placeholder $(q_{i-1}',\nort{s},p_{i+1}')$
to obtain an almost placeholder run $\sigma''$ in $(p,\nort{s},q)$ with effect $(q_{i-1}',\nort{s},p_{i+1}') + \mu$.
Observe that $\sigma'$ and $\sigma''$ have smaller weight than $\sigma$.
By induction hypothesis we know that $\nu \in W_{(q_{i-1}',\nort{s},p_{i+1}')}$
and $(q_{i-1}',\nort{s},p_{i+1}') + \mu \in W_{(p,\nort{s},q)}$.
By \eqref{w4} the effect $\nu + \mu$ of $\sigma$ belongs to $W_{(p,\nort{s},q)}$, as desired.
\end{enumerate}
\end{enumerate}
\end{proof}

\subsubsection*{A saturation procedure}
Consider the definition of $W_\tau$. In order to argue that an element belongs
to $W_\tau$, we need to apply two kinds of steps alternatingly: (i)~produce new
elements in $W_\tau$ using rules (\ref{w-1})--(\ref{w-4}) and (ii)~observe that
$W_\tau\cap Z_\tau$ has become non-empty so as to enable more rules among
(\ref{w-1})--(\ref{w-4}). Here, our key idea is to use grammars to decide if
applying rules (\ref{w-1})--(\ref{w-4}) proves a set $W_\tau\cap Z_\tau$ non-empty.
Hence, we assume that for a certain set $R\subseteq\rho(\cA)$ we have already established
that $W_\tau\cap Z_\tau\ne\emptyset$ for every $\tau\in R$. Then, we construct grammars
and decide their emptiness to check if this leads to more $\tau$ such that $W_\tau\cap Z_\tau$ is non-empty. 

Let us make this formal. Let $R\subseteq\rho(\cA)$ be a subset. 
We inductively define the tuple $(W^R_\tau)_{\tau\in \rho(\cA)}$, where $W^R_{\tau}\subseteq \N^{\rho(\cA)} + \Z^{V}$ for every $\tau\in\rho(\cA)$. It is the smallest tuple such that
\begin{enumerate}
	\item If $s$ is a leaf, then $E'^R_{(p,\nort{s},q)}\subseteq W^R_{(p,\nort{s},q)}$.
	\item If $r^x$ is above $s^y$ and $(p,s^y,q)\in R$, then $W^R_{(p,s^y,q)}\cap Y_{r^x}\subseteq W^R_{(p,r^x,q)}$.
	\item If $r^x$ is above $s^y$ and $(p,r^x,q)\in R$, then $(p,r^x,q)\in W^R_{(p,s^y,q)}$.
	\item If $(p',r^x,q')+\bu\in W^R_{(p,r^x,q)}$ and $\bv\in W^R_{(p',r^x,q')}$, then $\bv+\bu\in W^R_{(p,r^x,q)}$.
\end{enumerate}

We now perform the procedure outlined above: We start with $R^{(0)}=\emptyset$
and then set $R^{(i+1)}=\{\tau\in\rho(\cA) \mid W_\tau^{R^{(i)}}\cap
Z_\tau\ne\emptyset\}$.  Then we clearly have $R^{(0)}\subseteq
R^{(1)}\subseteq\cdots$ and thus there is some $n$ with $R^{(n+1)}=R^{(n)}$.
The following is immediate from the definition of $W_\tau$.
\begin{proposition}\label{saturation-placeholder-runs}
	For every $\tau\in\rho(\cA)$, we have $W^{R^{(n)}}_\tau=W_\tau$.
\end{proposition}

Thus, in order to decide $\REVREACH$ in $\compP$, it suffices to decide,
given a set $R^{(i)}$ as above, whether $W^R_\tau\cap Z_\tau$ is empty for each
$\tau\in\rho(\cA)$. To do this, we will construct certain grammars for which we
will show that emptiness can be decided in polynomial time. Here, we will use
the fact that each $R^{(i)}$ is obtained from the process above: We will assume
that $R\subseteq\rho(\cA)$ is \emph{admissible}, meaning that there exists some
$i$ with $R=R^{(i)}$.

\subsubsection*{The grammar construction} Let us now show how to construct a
grammar $G=(N,T,P)$ such that $N=\rho(\cA)$ and $L(\tau)=W_\tau^R$ for every
$\tau\in\rho(\cA)$. In our definition of grammars, we require that every
production either has the form (1)~$a\to \bu$ with $a\in N_0$ and $\bu\in
\N^{N_0}+\Z^T$ or (2)~$a\to b$ with $a\in N_i$ and $b\in N_j$ with $|i-j|=1$,
i.e. $a$ and $b$ are nonterminals on neighboring levels. For the construction,
it will be convenient to more generally allow productions $a\to b$
where $a\in N_i$ and $b\in N_j$ for $i\ne j$. If we have such a grammar,
one can introduce intermediate nonterminals to achieve the more restrictive format.
Moreover, this introduction of intermediate nonterminals preserves bidirectedness.

If we allow these more general productions, we need to specify how they
induce derivation steps. The relation $\deriv[0]$ is defined as
before. However, once $\deriv[i']$ is defined for all $i'<i$, we have
$\bv\deriv[i] \bv'$ iff there is an $a\in N_i$ with $\bv(a)>0$ and a production
$a\to a'$ for some $a'\in N_{i'}$ for some $i'<i$, and a $\bu\in L(a')\cap
\N^{N_{[i,k]}}+\Z^{T_{[i,k]}}$ such that $\bv'=\bv-a+\bu$. Here, we had to put
the intersection with $\N^{N_{[i,k]}}+\Z^{T_{[i,k]}}$, because if $a'\in
N_{i'}$ with $i'<i-1$, then $L(a')$ is not necessarily included in
$\N^{N_{[i,k]}}+\Z^{T_{[i,k]}}$.

Recall that $\Gamma=(V,I)$. The set $T$ will consist of the looped vertices of $\Gamma$, but also 
some auxiliary letters defined as follows. We pick an arbitrary linear order $\ll$ on $Q$. Then we have the letters
$\Theta=\{z_{p,\nort{s},q} \mid p\ll q, \text{$s$ is a leaf of $t$}\}$.
In other words, for any two states in $Q$ and each leaf $s$ of $t$, we create one letter in $\Theta$.
We set $T=\{v\in V \mid vIv\}\cup \Theta$.

\begin{figure}
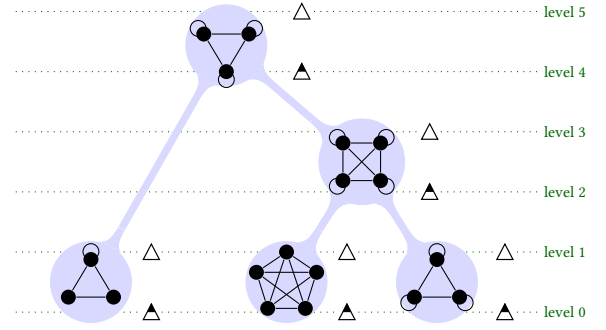

\treeLevels{1}
\caption{Example for choosing the level of $\rt{s}$ and $\nort{s}$ for each subtree $s$ of $t$. The division of the looped nodes into $\bigcup_{i\in[1,k]} T_i$ is obtained by placing each looped node into $T_i$, where $i$ is the level of the $\rt{s}$ directly above it.}
\label{tree-levels}
\end{figure}

First, we divide the sets $N=\rho(\cA)$ and $T$ into levels. The idea
is that types $(p,\nort{s},q)$ with leaves $s$ have level $0$ and the higher a
subtree $s$ is in $t$, the higher its level. Moreover, a type $(p,\rt{s},q)$
should have strictly higher level than $(p,\nort{s},q)$. Since
$\Gamma\in\SCpm_{d,\ell}$, we know that the height of $t$ is
$h(\Gamma)\le\ell$.  
We choose a map $\iota\colon \{\rt{s},\nort{s} \mid \text{$s$ is a subtree of $t$}\}\to [0,k]$ such that
(i)~for every leaf $s$ of $t$, we have $\iota(\nort{s})=0$, (ii)~if $r^x$ is above $s^y$, then $\iota(r^x)>\iota(s^y)$.
This can clearly be defined with some $k\le 2\ell$. Here, we need $2\ell$, because $\rt{s}$ and $\nort{s}$ must be on different levels. See \cref{tree-levels} for an example of how to choose the levels for each $s^x$.
This yields the partition $N=\bigcup_{i=0}^k N_i$ with $N_i=\{(p,s^x,q)\in\rho(\cA) \mid \iota(s^x)=i\}$.  Moreover, we set $T=\bigcup_{i=0}^k T_i$ where for $i\in[1,k]$, we have $v\in T_i$ if and only if
$\iota(\rt{s})=i$, where $s$ is the subtree whose root node contains $v$. Moreover, $T_0=\Theta$.

We begin by describing the productions $a\to\bu$ for $a\in N_0$ and
$\bu\in\N^{N_0}+\Z^T$. For this, we pick some leaf $s$ of $t$. Our goal is to guarantee 
$L(p,\nort{s},q)\ne\emptyset$ if and only if $W^R_{(p,\nort{s},q)}\cap Z_{\nort{s}}\ne \emptyset$ for any $p,q\in Q$.  To achieve this, we use
\cref{compute-coset-cliques}. We construct a valence system $\hat{\cA}_{s}$
over a graph $\tilde{\Gamma}_{s}$ as follows. The graph
$\tilde{\Gamma}_s=(\tilde{V}_s,\tilde{I}_s)$ has vertices $\tilde{V}_s=U_s\uplus L_s\uplus R_s$, where $U_s$ ($L_s$) is the set of unlooped (looped) vertices in $V$ that belong to $s$ and $R_s=\{(p',\nort{s},q')\in R\mid p',q'\in Q\}$. Moreover,
$\tilde{\Gamma}_s$ is a clique, and a vertex is looped if and only if it is
in $L_s\uplus R_s$. In other words, $\tilde{\Gamma}_s$ is constructed by taking the vertices
belonging to $s$ (which already form a clique) and adding for each
$(p',\nort{s},q')\in R_s$ another looped vertex. Then $\hat{\cA}_s$ is obtained from $\cA_s$
by adding, for any $(p',\nort{s},q')\in R_s$, an edge from $p'$ to $q'$ labeled by $(p',\nort{s},q')$,
and an edge from $q'$ to $p'$ labeled $-(p',\nort{s},q')$.
Now \cref{compute-coset-cliques}
allows us to compute $\bv,\bu_1,\ldots,\bu_n\in\Z^{L_s\cup R_s}$ such that
$\bv+\langle\bu_1,\ldots,\bu_n\rangle=\Eff_{\hat{\cA}_s}(p,q)$.
We now want to turn this coset representation into productions for the grammar.
For this, we need some notation. 
Then, given
a vector $\bu\in\Z^{U_s\cup R_s}$, define
$\alpha(\bu)\in\N^{R_s}+\Z^{U}$ by $\alpha(\bu)(p',\nort{s},q')=x+y$, where
\[ x=\begin{cases} \bu(p',\nort{s},q') & \text{if $\bu(p',\nort{s},q')\ge 0$} \\ 0 & \text{otherwise}\end{cases}\]
and
\[
y=\begin{cases} -\bu(q',\nort{s},p') & \text{if $\bu(q',\nort{s},p')<0$} \\ 0 & \text{otherwise}\end{cases}\]
Note that then $\alpha(\bu)\in\N^{R_s}+\Z^{U_s}$ and moreover, for any $p',q'\in Q$, we have $\alpha(\bu)(p',\nort{s},q')-\alpha(\bu)(q',\nort{s},p')=\bu(p',\nort{s},q')-\bu(q',\nort{s},p')$. 

Now we include the productions
\begin{align}
\begin{split}
	(p,\nort{s},q) &\to \alpha(\bv), \\ (p,\nort{s},q)&\to (p,\nort{s},q)+\alpha(\bu_j), \\  (p,\nort{s},q)&\to (p,\nort{s},q)+\alpha(-\bu_j), \label{productions-level-0a}
\end{split}
\end{align}
and
\begin{align}
\begin{split}
(q,\nort{s},p) &\to \alpha(-\bv), \\ (q,\nort{s},p)&\to (q,\nort{s},p)+\alpha(\bu_j), \\ (q,\nort{s},p)&\to (q,\nort{s},p)+\alpha(-\bu_j), \label{productions-level-0b}
\end{split}
\end{align}
for every $j\in[1,n]$.  Moreover, we include the production
\begin{equation} (p,\nort{s},q)\to (p,\nort{s},q)+(p,\nort{s},q)+(q,\nort{s},p) \label{productions-add}\end{equation}
for every $(p,\nort{s},q)\in R$.
Finally, we need productions that allow us to eliminate a pair $(p,\nort{s},q)$ and $(q,\nort{s},p)$.
If $p\ll q$, we add the two productions
\begin{align} (p,\nort{s},q)\to z_{p,\nort{s},q},&& (q,\nort{s},p)\to -z_{p,\nort{s},q}. \label{productions-delete}\end{align}
In addition to these level-$0$ productions, we also add the following for every $r^x$ and $s^y$:
\begin{enumerate}[label=(C\arabic*)]
	\item\label{productions-cross-a} If $r^x$ is above $s^y$ and $(p,s^y,q)\in R$, then we have a production $(p,r^x,q)\to (p,s^y,q)$.
	\item\label{productions-cross-b} If $r^x$ is above $s^y$ and $(p,r^x,q)\in R$, then we have a production $(p,s^y,q)\to (p,r^x,q)$.
\end{enumerate}

We prove the following:
\begin{proposition}\label{correctness-grammar}
	Given a graph $\Gamma$ in $\SCpm_{d,\ell}$ and a valence system $\cA$
	over $\Gamma$, and an admissible $R\subseteq\rho(\cA)$, we can
	construct in polynomial time a bidirected $k$-grammar $G=(N,T,P)$ with $k\le 2\ell$ such
	that $N=\rho(\cA)$ and $L(\tau)\ne\emptyset$ if and only if $W^R_\tau\cap Z_\tau\ne\emptyset$ for each $\tau\in\rho(\cA)$ with $\tau\in N_i$.
\end{proposition}

Before we verify \cref{correctness-grammar}, let us see how it implies \cref{bireach-to-grammars}.
\begin{proof}[Proof of \cref{bireach-to-grammars}]
	By \cref{revreach-placeholder-runs}, we have to decide
	whether $E_{(p,\rt{s},q)}\ne\emptyset$. By
	\cref{equivalence-placeholder-runs}, this is equivalent to
	$W_{(p,\rt{t},q)}\cap Z_{\rt{t}}\ne\emptyset$. In order to decide the
	latter, we apply the saturation procedure described after
	\cref{saturation-placeholder-runs}. We start with $R^{(0)}$ and then
	compute the set $R^{(i+1)}$ from $R^{(i)}$. Now
	\cref{correctness-grammar} allows us to construct for $R^{(i)}$ in
	polynomial time a bidirected $k$-grammar $G$ with $k\le 2\ell$ such
	that $\tau\in R^{(i+1)}$ if and only if $L(\tau)\ne\emptyset$.  Thus,
	by invoking the emptiness problem for $G$, we can compute $R^{(i+1)}$
	from $R^{(i)}$.  In the end, we have reachability if and only if
	$(p,\rt{t},q)\in R^{(n)}$.

	We invoke emptiness of bidirected grammars at most polynomially many
	times, because the saturation must finish in polynomially many steps:
	Each $R^{(i)}$ is a subset of $\rho(\cA)$, which has polynomially many
	elements, we must have $R^{(n+1)}=R^{(n)}$ for some polynomially
	bounded $n$.  

	Finally, since $k\le 2\ell$ in \cref{correctness-grammar}, we indeed
	invoke emptiness of bidirected grammars only for $k\le 2\ell$ if the
	input graphs are drawn from $\SCpm_{d,\ell}$.
\end{proof}

Let us observe that the same proof also yields an exponential space Turing
reduction if the input graphs are drawn from all of $\SCpm$. Here, it is
important to notice that although this is an exponential space reduction, it
yields $k$-grammars where $k$ and the number of productions is still
polynomial. However, the numbers occurring in the productions can be doubly
exponential (and thus use exponentially many bits).
\begin{proposition}\label{bireach-to-grammars-uc-unbounded}
	There is an exponential space Turing reduction from $\REVREACH(\SCpm)$
	to emptiness of bidirected $k$-grammars. Here, the resulting
	$k$-grammars will have $k$ and the number of productions bounded
	polynomially in the input.  Moreover, the numbers occurring in
	productions require at most exponentially many bits.
\end{proposition}

The rest of this section is devoted to proving \cref{correctness-grammar}.

\subsubsection*{Languages generated by the constructed grammar}
We prove that $L(\tau)\ne\emptyset$ if and only $W^R_\tau\ne\emptyset$.
We prove that $L(\tau)\ne\emptyset$ if and only if $W^R_\tau\cap
Z_\tau\ne\emptyset$ by showing that (i)~$W^R_\tau\cap Z_\tau\subseteq L(\tau)$ for every $\tau\in\rho(\cA)$ and
that (ii)~for every $\bu\in L(\tau)$ with $\tau\in N_i$, there exists a
$\bd\in\Delta_i$ such that $\bu+\bd\in W^R_\tau\cap Z_\tau$. Here, $\Delta_i$ is defined as follows. For $i\in[0,k-1]$, we define $\Delta_i\subseteq\N^{\rho(\cA)}$ as the submonoid generated by all $a+\inv{a}$ with $a\in R_{[i+1,k]}$. Moreover, $\Delta_k$ consists just of $\bzero$. 
Thus, we establish the following two facts.

\begin{proposition}\label{grammar-equivalence-a}
	For every $\tau\in\rho(\cA)$, we have $W^R_\tau\cap Z_\tau\subseteq L(\tau)$.
\end{proposition}
\begin{proposition}\label{grammar-equivalence-b}
	For every $\tau\in\rho(\cA)$ with $\tau\in N_i$ and every $\bu\in L(\tau)$, there exists a $\bd\in\Delta_i$ such that $\bu+\bd\in W^R_\tau\cap Z_\tau$.
\end{proposition}

Observe that \cref{grammar-equivalence-a,grammar-equivalence-b} together imply that $L(\tau)\ne\emptyset$ if and only if $W^R_\tau\ne\emptyset$.

We begin with \cref{grammar-equivalence-a}. We show \cref{grammar-equivalence-a} by proving that for every $\bu\in W^R_\tau$
with $\bu\in N_i$, we have $\tau\derivs[i]\bu$ in the grammar. This implies that $W^R_\tau\cap Z_\tau\subseteq L(\tau)$.
We first prove two auxiliary lemmas.
\begin{lemma}\label{eliminate-inverse-pair}
	For every $a\in R_0$, we have $a+\inv{a}\derivs[0]\bzero$.
\end{lemma}
\begin{proof}
	Let $a=(p,\nort{s},q)$ for some leaf $s$ of $t$ and states $p,q\in Q$ and suppose $p\ll q$ (the other case is symmetric).
	We use each of the productions in \cref{productions-delete} once to obtain
	\begin{align*}
	a+\inv{a}=(p,\nort{s},q)+(q,\nort{s},p) &\deriv[0] z_{p,\nort{s},q}+(q,\nort{s},p) \\
	&\deriv[0] z_{p,\nort{s},q}-z_{p,\nort{s},q}=\bzero.
	\end{align*}
\end{proof}

\begin{lemma}\label{derive-effects}
For any $p,q\in Q$ and every leaf $s$  with $a=(p,\nort{s},q)\in R$ and
$\bx\in\Eff_{\hat{\cA}_s}(p,q)$, we have $(p,\nort{s},q)\derivs[0] \bx$.
\end{lemma}
\begin{proof}
Let $s$ be a leaf of $t$ and $\bx\in\Eff_{\hat{\cA}_s}(p,q)$. In the
construction of the grammar, we compute $\bv,\bu_1,\ldots\bu_n\in\Z^{R_s\cup
L_s}$ so that $\Eff_{\hat{\cA}_s}(p,q)=\bv+\langle\bu_1,\ldots,\bu_n\rangle$.
Therefore, we can write
$\bx=\bv+x_1\cdot\bu_1+\cdots+x_n\cdot\bu_n+x'_1\cdot(-\bu_1)+\cdots+x'_n\cdot(-\bu_n)$
for some $x_1,\ldots,x_n,x'_1,\ldots,x'_n\in\N$. Note that we can derive
\[ (p,\nort{s},q)\derivs[0] \alpha(\bv)+\sum_{j=1}^n x_j\cdot\alpha(\bu_j)+x'_j\cdot\alpha(-\bu_j) \]
by using each production $(p,\nort{s},q)\to (p,\nort{s},q)+\alpha(\bu_j)$ exactly $x_j$-times, then $(p,\nort{s},q)\to(p,\nort{s},q)+\alpha(-\bu_j)$ exactly $x'_j$-times and finally the production $(p,\nort{s},q)\to\alpha(\bv)$ once. Let us denote the derived vector by $\by$.

Now observe that the vector $\by-\bx$ can
be written as a sum of vectors $c+\inv{c}$ with $c\in R_0$.
By \cref{eliminate-inverse-pair}, we can derive $c+\inv{c}\derivs[0]\bzero$ for each such $c$.
This implies that $\by\derivs[0]\bx$.
\end{proof}

The first step is to prove this for $\tau=(p,\nort{s},q)$, where $s$ is a leaf.
In this case, we have to show the following.
\begin{lemma}\label{derive-leaf-placeholder-runs}
	For every leaf $s$ of $t$ and every $p,q\in Q$ and every $\bx\in E'^R_{(p,\nort{s},q)}$, we have $\tau\derivs[0]\bu$. 
\end{lemma}
\begin{proof}
	Let 
	\begin{align*} (q_0,w_1,p_1)(p_1,s_1^{x_1},q_1)&(q_1,w_2,p_2)\cdots \\ &(p_m,s_m^{x_m},q_m)(q_m,w_m,p_{m+1}) \end{align*}
	be an almost $R$-placeholder run in $(p,\nort{s},q)$. We want to show that its effect $\bx=(p_1,s_1^{x_1},q_1)+\cdots+(p_m,s_m^{x_m},q_m)+\projup_{\nort{s}}(w_1\cdots w_m)\in\N^{\rho(\cA)}+\Z^T$ satisfies $\tau\derivs[0]\bx$. 

	Since this is an almost $R$-placeholder run, we have $(p_j,s_j^{x_j},q_j)\in R$ for $j\in[1,m]$ and since $R$ is admissible, this implies that $(p_j,\nort{s},q_j)\in R$ for each $j\in[1,m]$. Therefore, the sequence
	\begin{align*} (q_0,w_1,p_1)(p_1,\nort{s},q_1)&(q_1,w_2,p_2)\cdots \\ &(p_m,\nort{s},q_m)(q_m,w_m,p_{m+1}) \end{align*}
	is also an almost $R$-placeholder run. By definition of $\hat{\cA}_s$, this implies that $\by:=(p_1,\nort{s},q_1)+\cdots+(p_m,\nort{s},q_m)+\projup_{\nort{s}}(w_1\cdots w_m)$ belongs to $\Eff_{\hat{\cA}_s}(p,q)$. By \cref{derive-effects}, we have $(p,\nort{s},q)\derivs[0]\by$.

	Since we have a production $(p_j,\nort{s},q_j)\to (p_j,s_j^{x_j},q_j)$ for each $j\in[1,m]$, we can derive $\by\derivs[0]\bx$. Hence we have $(p,\nort{s},q)\derivs[0]\bx$. 
\end{proof}

We are now ready to prove \cref{grammar-equivalence-a}.
\begin{proof}[Proof of \cref{grammar-equivalence-a}]
	As mentioned above, we prove that for every $\bu\in W^R_\tau$ with $\tau\in N_i$, we have $\tau\derivs[i]\bu$.

	We proceed by induction on the number of rule applications used to
	conclude that $\bu\in W^R_\tau$. Put differently, we proceed by induction on $n$ and we assume that $\tau\derivs[i]\bu$ if $\bu$'s membership in $W^R_\tau$ can be derived in $<n$
	steps. And we prove that then every $\bu'$, whose membership can be
	derived in $n$ steps, also satisfies $\tau\derivs[i]\bu'$.

	The induction base is done in \cref{derive-leaf-placeholder-runs}. So suppose $\bu$ belongs to $W^R_\tau$ and can be derived in $n$ steps. We consider each rule in the definition of $W^R_\tau$:
	\begin{enumerate}
		\item Suppose $r^x$ is above $s^y$ and $(p,s^y,q)\in R$ and $\bu\in
			W^R_{(p,s^y,q)}\cap Y_{r^x}$. Then since $\bu\in
			W^R_{(p,s^y,q)}$ and can be derived in $<n$ steps, we know
			by induction that $(p,s^y,q)\derivs[i']\bu$, where $i'$ is the level of $s^y$. 
			
			There exists a production $(p,r^x,q)\to (p,s^y,q)$
			since $(p,s^y,q)\in R$,. Moreover, the fact that
			$\bu\in Y_{r^x}$ allows us to conclude
			$(p,r^x,q)\derivs[i] \bu$.
			
		\item Suppose $r^x$ is above $s^y$ and $(p,r^x,q)\in R$. Then there is a production $(p,s^y,q)\to (p,r^x,q)$ and hence we have $(p,s^y,q)\deriv[i] (p,r^x,q)$, where $i$ the level of $s^y$.

		\item Suppose that $(p',r^x,q')+\bu\in W^R_{(p,r^x,q)}$ and $\bv\in W^R_{(p',r^x,q')}$, where the memberships $(p',r^x,q')+\bu\in W^R_{(p,r^x,q)}$ and $\bv\in W^R_{(p',r^x,q')}$ can be derived in $<n$ steps. We have to show that
			$(p,r^x,q)\derivs[i]\bv+\bu$, where $i$ the level of $r^x$.

			By induction, we know $(p,r^x,q)\derivs[i](p',r^x,q')+\bu$ and that $(p',r^x,q')\derivs[i]\bv$.
			By applying these derivations successively, we obtain $(p,r^x,q)\derivs[i]\bv+\bu$.
	\end{enumerate}
	This establishes that indeed for every $\bu\in W^R_\tau$ with $\bu\in N_i$, we have $\tau\derivs[i]\bu$.
	This implies the inclusion $W^R_\tau\cap Z_{\tau}\subseteq L(\tau)$: If $\bu\in W^R_\tau$, then we know $\tau\derivs[i]\bu$. Since furthermore $\bu\in Z_\tau$, we even have $\bu\in L(\tau)$.
\end{proof}

Our next step is to prove \cref{grammar-equivalence-b}.  We begin with the case that $\tau=(p,\nort{s},q)$ for some leaf $s$ of $t$.
\begin{lemma}\label{grammar-equivalence-base-case}
	Let $s$ be a leaf of $t$ and $p,q\in Q$. For every $\bu\in L(p,\nort{s},q)$, there exists a $\bd\in\Delta_0$ such that $\bu+\bd\in W^R_{(p,\nort{s},q)}\cap Z_{\nort{s}}$.
\end{lemma}
\begin{proof}
	We first observe that if $(p,\nort{s},q)\derivs[0]\bu'$ using only productions \cref{productions-level-0a,productions-level-0b,productions-add,productions-cross-a}, then $\bu'\in E'^R_{(p,\nort{s},q)}$.
	This follows by induction from
	the construction of the productions
	\cref{productions-level-0a,productions-level-0b,productions-add}. For
	this, also note that each element $\alpha(\bu_j)$ and $\alpha(-\bu_j)$
	in \cref{productions-level-0a,productions-level-0b} belongs to
	$\Eff_{\hat{\cA}_s}(p,p)$.

	Now let $\bu\in L(p,\nort{s},q)$. Then $\bu$ is obtained using both
	(i)~productions
	\cref{productions-level-0a,productions-level-0b,productions-add,productions-cross-a} and
	(ii)~the productions in \cref{productions-delete}. This means, there is
	a $\bu'$ obtained from $(p,\nort{s},q)$ using \cref{productions-level-0a,productions-level-0b,productions-add,productions-cross-a} with $(p,\nort{s},q)\derivs[0]\bu'$, such that $\bu$
	is obtained from $\bu'$ by applying productions in
	\cref{productions-delete}. Since $\bu\in L(p,\nort{s},q)$ contains no
	level-$0$ terminal symbols, the productions in
	\cref{productions-delete} must have been applied pairwise. This means,
	we have $\bu'=\bu+\bd'$ for some
	$\bd=(b_1+\inv{b_1})+\cdots+(b_r+\inv{b_r})$ with $b_1,\ldots,b_r\in
	R_0$. Since $R$ is admissible, this implies that for each $j\in[1,r]$,
	we have $W^R_{b_j}\cap Z_{b_j}\ne\emptyset$ and also $W^R_{\inv{b_j}}\cap Z_{\inv{b_j}}\ne\emptyset$.
	Therefore, we can replace in $\bd'$ each $b_j$ by the effect $\bu_j$ of an $R$-placeholder run in $b_j$
	and we can replace $\inv{b_j}$ by $\inv{\bu_j}$: Note that $\inv{\bu_j}$ is the effect of an $R$-placeholder run in $\inv{b_j}$. Let $\bd$ be the resulting vector in $\N^{\rho(\cA)}$.
	Then clearly $\bd\in\Delta_0$. Moreover, we have $\bu+\bd\in W^R_{(p,\nort{s},q)}\cap Z_{\nort{s}}$.
\end{proof}

We are now ready to prove \cref{grammar-equivalence-b}.
\begin{proof}[Proof of \cref{grammar-equivalence-b}]
	We prove the statement by induction on $i$: \cref{grammar-equivalence-base-case} is the case $i=0$.
	Now suppose the statement holds for all $i'<i$. To prove the statement for $i$, we apply another induction by the number of derivation steps to derive an element of $L(\tau)$. More precisely, we show that for any $n\ge 1$, if $\tau\derivs[i]\bu$ in at most $n$ steps, then there exists a $\bd\in\Delta_i$ such that $\bu+\bd\in W^R_{\tau}$. Suppose this holds for all step counts $<n$.

	Let $\tau\derivs[i]\bu$ in $n$ steps with $\tau\in N_{i}$. We distinguish two cases.
	\begin{enumerate}
		\item Suppose $n=1$. Then there is a $\tau'\in N_{i'}$ for some $i'<i$ 
			such that $\tau'\derivs[i']\bu$ and $\bu\in
			L(\tau')\cap S_{i-1}$. Since $\bu\in L(\tau')$, our
			induction hypothesis on $i'$ yields some
			$\bd'\in\Delta_{i'}$ such that $\bu+\bd'\in
			W^R_{\tau'}\cap Z_{\tau'}$. 
			
			Since $\bu\in S_{i-1}$, we
			also have $\bu'\in Y_\tau$. 
			
			Therefore, it follows from
			the definition of $W^R_\tau$ that there is a
			$\bd\in\Delta_i$ with $\bu+\bd\in W^R_\tau$.
			
		\item Suppose $n>1$ and write $\tau=(p,r^x,q)$. Then there is a
			$\bu'$ such that $(p,r^x,q)\derivs[i]\bu'$ in $<n$
			steps and $\bu'=(p',r^x,q')+\bv$ and some $\bx$ with
			$(p',r^x,q')\deriv[i]\bx$ so that $\bu=\bx+\bv$.
			By our induction hypothesis, we know that there is a
			$\bd'\in\Delta_i$ such that $\bu'+\bd'\in
			W^R_{(p,r^x,q)}$. Moreover, there is a
			$\tilde{\bd}\in\Delta_i$ with
			$\bx+\tilde{\bd}\in W^R_{(p',r^x,q')}$. According
			to the last rule in the definition of $W^R_\tau$, this
			implies that $\bv+\bd+\bx+\tilde{\bd}\in
			W^R_{(p,r^x,q)}$. In particular, we have
			$\bu+\bd+\tilde{\bd}=\bv+\bx+\bd+\tilde{\bd}\in
			W^R_\tau$. Since $\bd+\tilde{\bd}\in\Delta_i$,
			this proves our claim.
	\end{enumerate}
\end{proof}

\subsubsection*{Bidirectedness of the constructed grammar}
Now we prove that the grammar is indeed bidirected.
The conditions
\cref{symmetry-production-inverse,symmetry-reverse-nton,symmetry-add-inverse-pairs}
are obvious from the construction.  Moreover, the condition
\cref{symmetry-rhs} follows from $L(\tau)\ne\emptyset$ if and only if $W^R_\tau\cap Z_\tau\ne\emptyset$ for every $\tau\in\rho(\cA)$.
It remains to show the following.
\begin{lemma}\label{details-symmetry-production-reverse}
The constructed grammar satisfies \cref{symmetry-production-reverse} of the bidirectedness conditions.
\end{lemma}
\begin{proof}
The property \cref{symmetry-production-reverse} is clear immediately except for
the productions in
\cref{productions-level-0a,productions-level-0b}. Let us first prove the property for
a production $(p,\nort{s},q)\to \bu$ with $\bu=\alpha(\bv)$ on the left of \cref{productions-level-0a}. Write $\bu=b+\bx+\by$ for $b\in R_0$, $\bx\in\N^{R_0}$, and $\by\in\Z^T$. 

The vector $\bu$ is chosen so that $\bu=\alpha(\bv)$ with $\bv\in \Eff_{\hat{\cA}_s}(p,q)$. By construction of $\hat{\cA}_s$, this implies the existence of a run
\begin{align}
\begin{split} (q_0,w_1,p_1)(p_1,\nort{s}_1,q_1)&(q_1,w_2,p_2)\cdots \\ &(p_m,\nort{s}_m,q_m)(q_m,w_m,p_{m+1})
\end{split}
\end{align}
with $p=q_0$, $q=p_{m+1}$ and $(p_j,\rt{s},q_j)\in R_s$ for $j\in[1,m]$.
Furthermore, there exists a run $q_j\autsteps[w_j]p_{j+1}$ in $\cA_s$ and $[\pi_{\nort{s}}(w_1\cdots w_m)]=1$,
and $(p_1,\nort{s},q_1)+\cdots+(p_m,\nort{s},q_m)=\bx$ and $\by=[\hat{\pi}_{\nort{s}}(w_1\cdots w_m)]$. Let $a=(p,\nort{s},q)$. Moreover, we assume $b=(p_1,\nort{s},q_1)$: If $b=(p_j,\nort{s},q_j)$, the proof is analogous.

Observe that $(p_1,\nort{s},q_1)\in R_s$ by construction of $\hat{\cA}_s$. Moreover, consider the sequence
\begin{align}
\begin{split}(p_1,\bar{w}_1,q_0)(p,\nort{s},q)&(p_{m+1},\bar{w}_m,q_m)(q_m,\nort{s},p_m)\cdots\\&(q_2,\nort{s},p_2)(p_2,\bar{w}_2,q_1).
\end{split}\end{align}
Since $\hat{\cA}_s$ is also bidirected, this corresponds to a run in $\hat{\cA}_s$ from $p_1$ to $q_1$.
This run of $\hat{\cA}_s$ has effect $\bu'=\inv{(\bu-(p_1,\nort{s},q_1))}+(p,\nort{s},q)$
and thus $\bu'\in \Eff_{\hat{\cA}_s}(p_1,q_1)$. According to \cref{derive-effects}, we thus have \[ b=(p_1,\nort{s},q_1)\derivs[0] \bu'=\inv{(\bx+\by)}+a. \]
This concludes the proof for the production $(p,\nort{s},q)\to\alpha(\bv)$.

For productions $(p,\nort{s},q)\to\alpha(\bu_j)$, we can note that if $\bv+\langle \bu_1,\ldots,\bu_n\rangle=\Eff_{\hat{\cA}_s}(p,q)$, then each $\bu_j$ belongs to $\Eff_{\hat{\cA}_s}(p,p)$ and then argue as above.
\end{proof}
\subsection{Properties of bidirected grammars}
In this subsection, we prove \cref{level-bidirectedness}.
We begin with some simple observations.
\begin{lemma}\label{bidirected-add-da}
	If $G=(N,T,P)$ is bidirected, and $a \in R_i$, then
	$a\derivs[i] a+a+\inv{a}$.
\end{lemma}
\begin{proof}
	We proceed by induction on $i$. For $i=0$, this is one of the
	conditions of bidirectedness. For $i>0$ and $a\in R_i$, we know that
	there is some $a'\in R_{i-1}$ such that there
	are productions $a\to a'$ and $a'\to a$. By
	induction, we have $a'\derivs[i-1] a'+a'+\inv{a'}$. Moreover, since
	there is a production $\inv{a'}\to\inv{a}$ and $a\in R_i$ guarantees a
	production $a'\to a$, we obtain $a\derivs[i] a+a+\inv{a}$.
\end{proof}

\begin{lemma}\label{derive-equivalent-sum}
	If $G$ is $i$-bidirected and $a\in R_i$ with $a\derivs[i]\bu'$ for some $\bu,\bu'\in\N^{R_{[i,k]}}+\Z^{T_{[i,k]}}$ with $\bu'\approx_a\bu$, then there is a $\bd\in\Delta_a$ with $a\derivs[i]\bu+\bd$.
\end{lemma}
\begin{proof}
	First, we claim that for any $a,b\in R_i$ with $a\nreachs b$, there
	exists a $\bd\in\Delta_a$ such that $a\derivs[i]
	a+b+\inv{b}+\bd$. If $a\nreachs b$, there is a
	$\bv\in\N^{R_{[i,k]}}+\Z^{T_{[i,k]}}$ with $a\derivs[i]b+\bv$ such that
	$a\deriv[i]b+\bv$. Since $G$ is $i$-bidirected, we have $b\derivs[i]
	a+\bv'$ for some $\bv'\approx_a\inv{\bv}$.  By
	\cref{bidirected-add-da}, we have $b\derivs[i]b+b+\inv{b}$ and thus
	\[ a\derivs[i]b+\bv\derivs[i]b+b+\inv{b}+\bv\derivs[i]a+\bv+\bv'+b+\inv{b}. \]
	Since $\bv'\approx_a\inv{\bv}$, we have $\bv+\bv'\in\Delta_a$. This proves our claim.

	Now consider $a\derivs[i]\bu'$. Since $\bu'\approx_a\bu$, we know that
	there exist $\bd,\bd'\in\Delta_a$ with $\bu+\bd=\bu'+\bd'$.
	By our claim, we know that there is a $\bd''\in\Delta_a$ such that
	$a\derivs[i]a+\bd'+\bd''$.  Thus, we have
	$a\derivs[i]a+\bd'+\bd''\derivs[i]\bu'+\bd'+\bd''=\bu+\bd''$.
\end{proof}

The following lemma tells us that $i$-bidirectedness can be checked as a property
of each derivation step.
\begin{restatable}{lemma}{stepBidirectedness}\label{step-bidirectedness}
	If $G$ is bidirected, then the following are equivalent:
	\begin{enumerate}
	\item $G$ is $i$-bidirected.
	\item for
	every $a\in R_i$ and every step $a\deriv[i]b+\bv$ with $b\in R_i$ and
	$\bv\in\N^{R_i}+\Z^T$, there is a $\bv'\in\N^{R_i}+\Z^T$ with
	$b\derivs[i]a+\bv'$ and $\bv'\approx_a \inv{\bv}$.
	\end{enumerate}
\end{restatable}
\begin{proof}
	Clearly, if $G$ is $i$-bidirected, then the second condition is satisfied, so let us show the converse. 
	First, observe that if the second condition holds, then $\nreachs$ is
	symmetric as a relation on $R_i$: Indeed, if $a\nreach b$, then the condition implies that
	$b\nreachs a$. 

	Now suppose $a\derivs[i]\bu+\bv$ with $\bu\in\N^{R_i}$
	and $\bv\in\N^{R_{[i,k]}}+\Z^{T_{[i,k]}}$. We pick one letter $b\in
	R_i$ occurring in $\bu$ and write $\bu=b+\bu'$ for some $\bu'\in
	\N^{R_i}$. In the derivation $a\derivs[i]\bu+\bv$, we choose the chain $a_0,\ldots,a_n$ of nonterminals that leads to $b$, i.e. $a_0=a$, $a_n=b$ and such that $a_{j+1}$ is created by replacing $a_{j}$ for $j\in[0,n-1]$. Hence, there are vectors 
	$\bx_1,\ldots,\bx_n\in\N^{R_{[i,k]}}+\Z^{T_{[i,k]}}$ so that we use the productions $a_j\deriv[i]a_{j+1}+\bx_{j+1}$ for $j\in[0,n-1]$ to obtain $b$. Then, we have $a_0\deriv[i]a_1+\bx_1\deriv[i]a_2+\bx_2+\bx_1\deriv[i]\cdots\deriv[i]a_n+\bx_n+\cdots+\bx_1$ and also $\bx_1+\cdots+\bx_n\derivs[i] \bu'+\bv$.
	By the second condition and \cref{derive-equivalent-sum}, for each $j\in[1,n]$, there exists a $\bd_j\in\Delta_a$ such that $a_j\derivs[i] a_{j-1}+\inv{\bx_j}+\bd_j$. Putting these together, we obtain
	\begin{align*}
	\begin{split}
	b=a_n &\derivs[i] a_0+\inv{\bx_1}+\bd_1+\cdots+\inv{\bx_n}+\bd_n \\ &=a+\inv{\bx_1}+\cdots+\inv{\bx_n}+\bd,
	\end{split}
	\end{align*}
	where we define $\bd=\bd_1+\cdots+\bd_n$.
	Since $\bx_1+\cdots+\bx_n\derivs[i]\bu'+\bv$, we therefore have
	$b\derivs[i] a+\inv{(\bu'+\bv)}+\bd$ and thus
	\[ \bu=b+\bu'\derivs[i] a+\inv{(\bu'+\bv)}+\bd+\bu'=a+\bu'+\inv{\bu'}+\inv{\bv}+\bd. \]
	Finally, since $\bu\in\N^{R_i}$ and thus $L(c)\ne\emptyset$ for every $c$ occurring in $\bu$ (and also $a\nreachs c$), we know that there is some $\be\in\N^{R_{[i+1,k]}}+\Z^{T_{[i+1,k]}}$ with $\bu\derivs[i]\be$. But then we have $\bu+\inv{\bu}\derivs[i]\be+\inv{\be}$, where $\be+\inv{\be}\in\Delta_a$. Thus, we have $\bu\derivs[i] a+\be+\inv{\be}+\inv{\bv}+\bd$.
	Now with $\bv'=\inv{\bv}+\be+\inv{\be}+\bd$, we clearly have $\bv'\approx_a\inv{\bv}$. Thus, $G$ is $i$-bidirected.
\end{proof}

\levelBidirectedness*
\begin{proof}
	We proceed by induction on $i$, so suppose $G$ is $(i-1)$-bidirected.

	According to \cref{step-bidirectedness}, it suffices to show that for
	$a\deriv[i]b+\bv$ with $a,b\in R_i$ and
	$\bv\in\N^{R_{[i,k]}}+\Z^{T}$, there exists a
	$\bv'\approx_a\inv{\bv}$ with $b\derivs[i] a+\bv'$. If
	$a\deriv[i]b+\bv$, then this is due to a production $a\to \tilde{a}$
	for some $\tilde{a}\in R_{i-1}$, a derivation $\tilde{a}\derivs[i-1]
	\tilde{b}+\bv$, and a production $\tilde{b}\to b$. 
	Since $G$ is
	$(i-1)$-bidirected, there exists a $\tilde{\bv}\approx_{\tilde{a}}\inv{\bv}$
	and a derivation $\tilde{b}\derivs[i-1] \tilde{a}+\tilde{\bv}$. Since
	$\tilde{a}+\tilde{\bv}\approx_{\tilde{a}}\tilde{a}+\inv{\bv}$, we know from
	\cref{derive-equivalent-sum} that we can derive
	$\tilde{b}\derivs[i]\tilde{a}+\inv{\bv}+\tilde{\bd}$ for some
	$\tilde{\bd}\in\Delta_{\tilde{a}}$. Moreover, since
	$\tilde{\bd}\in\Delta_{\tilde{a}}$, there is a $\bd\in\Delta_a$
	with $\tilde{\bd}\derivs[i-1]\bd$. Hence, we have
	$\tilde{b}\derivs[i-1]\tilde{a}+\inv{\bv}+\bd$.

	Note that since $a,b\in R_i$, we also have productions $b\to \tilde{b}$
	and $\tilde{a}\to a$, so that we also have
	$b\deriv[i]a+\inv{\bv}+\bd$ and thus $\bv'=\inv{\bv}+\bd$ is as desired.
\end{proof}

\subsection{Expressing emptiness in terms of cosets}

\begin{lemma}\label{bidirected-add-da-garbage}
	If $G=(N,T,P)$ is $i$-bidirected, and $a \in R_i$ and $b\in R_i$ with $a\nreachs b$. Then there is some $\bd\in\Delta_a$ such that $a\derivs[i] a+b+\inv{b}+\bd$.
\end{lemma}
\begin{proof}
	Since $a\nreachs b$, there is some
	$\bu\in\N^{N_{[i,k]}}+\Z^{T_{[i,k]}}$ with $a\derivs[i] b+\bu$. Since
	$G$ is $i$-bidirected, this implies $b\derivs[i] a+\bu'$ with
	$\bu'\approx_a\inv{\bu}$. By \cref{bidirected-add-da}, we know that
	$b\derivs[i] b+b+\inv{b}$. Hence, we have $a\derivs[i] b+\bu \derivs[i]
	b+b+\inv{b}+\bu\derivs[i] a+\bu+\bu'+b+\inv{b}$. Now observe that since
	$\bu'\approx_a\inv{\bu}$, we have $\bu+\bu'\in\Delta_a$.  
\end{proof}

\bidirectedDaIa*
\begin{proof}
	Consider some $b\in R_{i+1}$ with $a\nreachs b$. Then there must be an $a'\in R_i$ with $a\nreachs a'$ and a production $a'\to b$. By \cref{bidirected-add-da},
	we have $a'\derivs[i] a'+a'+\inv{a'}$ and thus $a'\derivs[i] a'+b+\inv{b}$. Therefore, we have $b+\inv{b}\in H_a$.
\end{proof}

\bidirectedGroupMonoid*
\begin{proof}
	Since $G$ is bidirected, we know that $D_a\subseteq H_a$. Thus, the
	inclusion ``$\supseteq$'' is clear. For the the converse, it suffices
	to show that for every $b\in R_i$ with $a\nreachs b$ and
	$b\deriv[i]\bu$, the vector $-(-b+\bu)$ belongs to the right-hand side.
	We write $\bu=\bv+\bx$ with $\bv\in\N^{R_i}$ and
	$\bx\in\N^{N_{[i+1,k]}}+\Z^{T}$. Since $G$ is $i$-bidirected, we know
	that $\bv\derivs[i] b+\bx'$ such that $\bx'\in\inv{\bx}+D_b$.
	But this implies $-\bv+b+\bx'$ belongs to the monoid on the
	right-hand side. Moreover, the two elements $-(-b+\bu)$ and $-\bv+b+\bx'$ differ in \[ -(-b+\bu)-(-\bv+b+\bx')=-\bx-\bx'\in -\bx-\inv{\bx}+D_b\subseteq D_b\]
	and since $a\nreachs b$ and $G$ is $i$-bidirected, we also have $b\nreachs a$
	and thus $D_b=D_a$. Thus $-(-b+\bu)$ belongs to the right-hand side.
\end{proof}

\emptinessNotRHS*
\begin{proof}
	Since (2) is a direct consequence of the definition and previous lemmas, it remains to prove~(1).
	Suppose $L(a)\ne\emptyset$. Then there is a derivation
	$a\deriv[i]\bu_1\derivs[i]\bu$ with
	$\bu\in S_i$. In particular, the first derivation step must be due to
	some production $a\to a'$ with $\bu_1\in L(a')$. This means $a'\in
	R_{i-1}$ and thus $\bu_1\in L_{a'}$ according to
	\cref{grammars-to-groups}. Write $\bu_1=b_1+\cdots+b_n+\bv$, where $b_1,\ldots,b_n\in R_i$ and $\bv\in\N^{N_{[i+1,k]}}+\Z^{T_{[i,k]}}$.
	Since $\bu_1\derivs[i]\bu$, we know that for each $j\in[1,n]$, we have $b_j\derivs[i]\bv_j$ for some $\bv_j\in\N^{N_{[i+1,k]}}+\Z^{T_{[i,k]}}$ such that $\bu=\bv_1+\cdots+\bv_n+\bv$. This means in particular that $\bv_j\in M(b_j)$ and thus
	\begin{align*}
	\begin{split}
	S_i\ni \bu&=\bu_1 + (-b_1+\bv_1)+\cdots+(-b_n+\bv_n) \\
	&\in L_{a'}+\langle -b_1+M_{b_j}\rangle+\cdots+\langle -b_n+M_{b_n}\rangle.
	\end{split}
	\end{align*}
	
	Moreover, since $b_1+\cdots+b_n+\bv=\bu_1\in L(a')$, the nonterminals $b_1,\ldots,b_n$ satisfy $a'\nreachs b_j$ for $j\in[1,n]$. Thus, $\bu$ belongs to $K_{a'}$.

	Conversely, suppose $K_{a'}\ne\emptyset$ for some production $a\to a'$.
	Then $a'\in R_{i-1}$ and thus $L_{a'}=L(a')+D_{a'}$ according to
	\cref{grammars-to-groups}. Consider some $\bu\in K_{a'}$ and write
	$\bu=\bu_1+\bv_1+\cdots+\bv_n$ with $\bu_1\in L_{a'}$ and
	$\bv_j=(-1)^{\varepsilon_j}(-b_j+\bx_j)$ for some $b_j\in R_i$,
	$a'\nreachs b_j$, and $\bx_j\in M_{b_j}$.

	In order to construct a derivation, we have to eliminate those summands
	$\bv_j$ with $\varepsilon_j=1$, because they do not directly correspond
	to derivation steps. Hence, we define $\hat{\bv}_1,\ldots,\hat{\bv}_n$
	as follows. For $j\in[1,n]$, if $\varepsilon_j=0$, then $\hat{b}_j=b_j$
	and $\hat{\bx}_j=\bx_j$; if $\varepsilon_j=1$, then $\hat{b}_j=\inv{b_j}$ and $\hat{\bx}_j=\inv{\bx_j}$. Then $\hat{\bv}_j=-\hat{b}_j+\hat{\bx}_j$.
	Then each difference $\hat{\bv}_j-\bv_j$ is either $\mathbf{0}$ or
	$(-\inv{b}_j+\inv{\bx})-(-1)(-b_j+\bx_j)=-\inv{b_j}-b_j+\inv{\bx_j}+\bx_j$. Therefore, with $\bb=b_1+\inv{b_1}+\cdots+b_n+\inv{b_n}$ and
	$\bg=\bx_1+\inv{\bx_1}+\cdots+\bx_n+\inv{\bx_n}$, we have
\[ \bu_1+\hat{\bv}_1+\cdots+\hat{\bv}_n+\bb=\bu_1+\bv_1+\cdots+\bv_n+\bg. \]
Since $\bu_1\in L(a')+D_{a'}$ and $\bb\in\Delta_{a'}$, we have $\bu_1+\bb\in L(a')+D_{a'}$, so that \cref{bidirected-add-da-garbage} implies that we can write $\bu_1+\bb=\by-\bd$ with some $\by\in L(a')$ and $\bd\in \Delta_{a'}$. This implies
\begin{equation} \by+\hat{\bv}_1+\cdots+\hat{\bv}_n = \bu_1+\bv_1+\cdots+\bv_n+\bg+\bd. \label{a}\end{equation}
Now notice that since $\bu_1+\bb=\by-\bd$, we know that $\bb$ must occur in $\by$. In particular, $\hat{b}_1+\cdots+\hat{b}_n$ must occur in $\by$.
Hence, we have a
derivation $a'\derivs[i] \by+\hat{\bv}_1+\cdots+\hat{\bv}_n$. By
\cref{a}, we have thus derived
$\bu_1+\bv_1+\cdots+\bv_n+\bg+\bd=\bu+\bg+\bd$. Since $\bu\in S_i$
and $\bg\in\N^{N_{[i+1,k]}}\subseteq S_i$, it remains to eliminate the
level-$i$ nonterminals in $\bd$. However, since $\bd\in\N^{R_i}$, we can pick
for each $f$ occurring in $\bd$ some $\be_f\in L(e)$, $\be\in\N^{N_{[i+1,k]}}+\Z^{T_{[i+1,k]}}$. Then, if $\bd=f_1+\cdots+f_r$, then we have $a\derivs[i]\bu+\bg+\be_{f_1}+\cdots+\be_{f_r}\in\N^{N_{[i+1,k]}}+\Z^{T_{[i+1,k]}}$. We therefore have $L(a)\ne\emptyset$.
\end{proof}

\subsection{From coset circuits to linear Diophantine equations}\label{appendix-circuits-to-matrices}

\upperBoundsPExptimeExpspace*
\begin{proof}
	 We begin with~(1) and~(2). By \cref{bireach-to-grammars}, it suffices
	 to show that emptiness for bidirected $k$-grammars is decidable in
	 $\EXPTIME$ if $k$ is part of the input, and in $\compP$ for fixed $k$.

	 Given a bidirected $k$-grammar, we construct a coset circuit $C$ as
	 described in \cref{construct-coset-circuit} by alternating emptiness
	 checks for gates and building new gates.  To check gates for
	 emptiness, recall that the circuit has depth $\le ck$ for some
	 constant $c\in\N$. Thus, it remains to decide emptiness of a gate $g$
	 in exponential time, resp.\ in polynomial time for bounded depth
	 circuits.  For each gate $g$ in this coset circuit, we compute a
	 matrix representation for the coset $C(g)$. As argued above, the
	 resulting matrix $\bA\in\Z^{s\times t}$ and vector $\bb\in\Z^s$ will
	 satisfy $s,t\le(r+1)^{ck}$, where $r$ is the largest in-degree of a
	 gate in $C$. The magnitude of $\bA$ and $\bb$ is at most the magnitude
	 of the matrices in the leaves of $C$, which means the entries of $\bA$
	 and $\bb$ only require polynomially many bits.  Therefore, the size of
	 $\bA$ and $\bb$ is polynomial in $(r+1)^{ck}$.  Hence, by
	 \cref{ild-in-p}, we can decide emptiness of $C(g)$ in time polynomial
	 in $(r+1)^{ck}$.

	For (3), we proceed slightly differently.  Since now even $d$ is part of the
	input, we use \cref{bireach-to-grammars-uc-unbounded} instead of
	\cref{bireach-to-grammars}. Then, the dimensions of all vectors and the
	number of productions in the grammar are still polynomial.  Therefore,
	our coset circuit has linear depth and the matrices labeling its leaves
	have entries that require exponentially many bits.  The resulting
	matrix for each gate therefore has an exponential number of rows and
	columns and its entries require at most exponentially many bits. Thus,
	we can again apply \cref{ild-in-p} to check emptiness of the gate in
	exponential time. All together, we obtain an exponential space
	algorithm.
\end{proof}

\end{document}